\documentclass[11pt]{article}
\usepackage{graphicx} %

\usepackage{mathpazo}
\usepackage{pdfpages}
\usepackage{bm}
\usepackage[margin=1in]{geometry}
\usepackage{amsmath,amssymb}

\usepackage[utf8]{inputenc}
\usepackage{color,amsmath,amsfonts,fullpage,amsthm,mathrsfs,fancyhdr,verbatim,framed,hyperref,xspace,mdframed,tabularx,booktabs,enumitem,comment,float}
\usepackage{cleveref}
\usepackage[dvipsnames]{xcolor}
\usepackage{graphicx}
\usepackage{amsthm}
\usepackage{mathpazo}
\usepackage{amsmath}
\usepackage{amssymb}
\usepackage{mathtools}
\usepackage{thmtools,thm-restate}
\usepackage{fec} 
\usepackage{soul}
\usepackage{authblk}

\newtheorem{theorem}{Theorem}
\numberwithin{theorem}{section}
\newtheorem{Lemma}[theorem]{Lemma}

\newtheorem{remark}[theorem]{Remark}
\newtheorem{proposition}[theorem]{Proposition}
\newtheorem{assumption}[theorem]{Assumption}
\newtheorem{corollary}[theorem]{Corollary}

\newtheorem{definition}[theorem]{Definition}

\newtheorem{problem}[theorem]{Problem}

\usepackage{tikz}
\usetikzlibrary{arrows.meta,positioning,fit,calc,backgrounds}
\newcount\Comments
\newcommand{\kibitz}[2]{\ifnum\Comments=1\textcolor{#1}{#2}\fi}

\definecolor{lightteal}{rgb}{0.5, 0.8, 0.8}
\usepackage[style=alphabetic, maxbibnames=99, backref=true, backrefstyle=three]{biblatex}
\bibliography{biblio}

\newcommand{\ground}[1]{{#1^{\mathrm{(ground)}}}}
\newcommand{\wkb}[1]{{#1^{\mathrm{(WKB)}}}}

\title{Mechanisms for Quantum Advantage in Global Optimization of Nonconvex Functions}
\author{Dylan Herman\footnote{These authors contributed equally.}}
\author{Guneykan Ozgul\protect\footnotemark[1]}
\author{Anuj Apte}
\author{Junhyung Lyle Kim}
\author{Anupam Prakash}
\author{Jiayu Shen}
\author{Shouvanik Chakrabarti\thanks{shouvanik.chakrabarti@jpmchase.com}}
\affil{Global Technology Applied Research, JPMorganChase, New York, NY 10001, USA}
\date{October 2025}

\begin{document}

\maketitle

\begin{abstract}
We present new theoretical mechanisms for quantum speedup in the global optimization of nonconvex functions, expanding the scope of quantum advantage beyond traditional tunneling-based explanations. As our main building-block, we demonstrate a rigorous correspondence between the spectral properties of Schr\"{o}dinger operators and the mixing times of classical Langevin diffusion. This correspondence motivates a mechanism for separation between quantum and classical algorithms on functions with unique global minimum: while quantum algorithms operate on the original potential, classical diffusions correspond to a Schr\"{o}dinger operators with a WKB potential having nearly degenerate global minima. We formalize these ideas by proving that a real-space adiabatic quantum algorithm (RsAA) achieves provably efficient, polynomial-time optimization for broad families of nonconvex functions. First, for block-separable functions, we show that RsAA maintains polynomial runtime while known off-the-shelf algorithms require exponential time and structure-aware algorithms exhibit arbitrarily large polynomial runtimes. These results leverage novel non-asymptotic versions of well-known semiclassical spectral analysis results. Second, we use recent advances in the theory of intrinsic hypercontractivity to demonstrate polynomial runtimes for RsAA on appropriately perturbed strongly convex functions that lack global structure, while off-the-shelf algorithms remain exponentially bottlenecked.
In contrast to prior works based on quantum tunneling, these separations do not depend on the geometry of barriers between local minima. Our theoretical claims about classical algorithm runtimes are supported by rigorous analysis and comprehensive numerical benchmarking. These findings establish a rigorous theoretical foundation for quantum advantage in continuous optimization and open new research directions connecting quantum algorithms, stochastic processes, and semiclassical analysis.
\end{abstract}

\newpage
\tableofcontents
\newpage

\section{Introduction}
\subsection{Motivation}
A continuous optimization problem~\cite{nemirovskii1983problem,nesterov2018lectures} is defined by $\min_{x \in \Ccal} f(x)$ where $f \colon \R{d} \to \R{}$ is a continuous function, and $\Ccal \subseteq \R{d}$ specifies the set of feasible solutions. The theory and practice of solving these problems have been a major theme of algorithmic research, due to their central role in various disciplines, including machine learning, scientific computing, operations research, and financial engineering.  In order to design efficient algorithms for these problems, one must typically assume additional properties such as the \emph{convexity} of the objective function and constraint set. The setting where we forgo convexity and instead make much weaker assumptions concerning the continuity of the objective and/or its derivatives is known as nonconvex optimization.

Nonconvex optimization is a very powerful primitive for modeling diverse problems due to its extreme flexibility. This flexibility comes at the cost of generally efficient algorithms, and in fact even the optimization of general quadratic functions can easily be shown to be NP-Hard~\cite{vavasis1995complexity}. One of the primary difficulties is the potential existence of an exponentially large number of stationary points, which trap local optimization methods.
For this reason, provably efficient algorithms for nonconvex optimization target either the easier task of finding a local minimum, or assume some general proxies for convexity such as the Polyak-\L{}ojasiewicz condition~\cite{karimi2016linear}. Outside of these cases, nonconvex functions can be globally optimized only in some special cases.

The ubiquity and importance of optimization problems has naturally led to widespread interest in developing quantum algorithms to solve them~\cite{abbas2024challenges}. It has often been suggested that quantum algorithms can provide a computational advantage over classical algorithms for globally optimizing nonconvex functions. The common intuition for this advantage comes from the physical phenomenon of \emph{quantum tunneling}, i.e., the ability of a quantum particle to pass through a potential barrier more effectively than a particle following a classical stochastic dynamics. This phenomenon has been well studied in the discrete setting and has led to the construction of some cost functions where an adiabatic quantum algorithm exhibits a provable advantage over simulated annealing~\cite{farhi2000quantum,reichardt2004quantum}. This initial observation has led to many theoretical explorations of adiabatic and annealing algorithms for combinatorial and continuous nonconvex optimization~\cite{crosson2021prospects}. These ideas have also motivated the design of specialized hardware that are designed for the analog simulation of quantum annealing protocols \cite{boixo2014evidence}. There have also been many studies of strategies to improve these algorithms such as the use of optimized annealing schedules \cite{herr2017}, optimal control mechanisms such as Bang-Bang control \cite{bapat2019}, and the addition of counter-diabatic terms to stabilize the simulation \cite{passarelli2023, delCampo2013}.

Despite the significant progress in the study of quantum annealing and adiabatic algorithms, there remain significant challenges in their theoretical understanding. Firstly, the mechanisms for tunneling through potential barriers faster than simulated annealing are dependent on the shape of the barriers between local minima~\cite{reichardt2004quantum,Liu_2023}. In multiple dimensions, the requirements on the shape of the barrier are not immediately clear and this makes it challenging to concretely identify families of potentials for which the adiabatic annealing algorithms offer an advantage over simulated annealing. More seriously, in the cases where there are demonstrated advantages over classical algorithms, it has been shown that the speedup over \emph{quantum Monte-Carlo} algorithms such as Simulated Quantum Annealing is only polynomial \cite{crosson2016simulated,crosson2021prospects}. From a purely technical point of view, the lack of theoretical tools to bound spectral gaps of the Hamiltonians involved in the adiabatic algorithm limits the range of problems for which these arguments can be applied. As a consequence, it often appears that quantum annealing algorithms should be seen as heuristic methods that are resistant to end-to-end theoretical analysis for interesting functions. 

In this paper, we seek to address these challenges. We identify mechanisms for quantum advantage in nonconvex optimization that go beyond the tunneling intuition that has been discussed in prior works. Specifically, these mechanisms are not reliant on the specific shape of barriers between local minima. As a consequence, we are able to provide evidence %
that they persist even against quantum Monte-Carlo algorithms.
From a technical point of view, our analysis does not directly depend on tunneling amplitudes between local minima.

Building upon these mechanisms we construct general classes of nonconvex functions for which the runtime of an adiabatic algorithm can be bounded explicitly. These runtime bounds are built upon several technical tools for analyzing the spectral gaps of Schr\"{o}dinger operators, which may be of separate interest. The quantum algorithm that we analyze will be fixed and for fairness, will not be adjusted based on the class of function being optimized.

We also discuss the possibility of classical algorithms to optimize these function classes. We consider both off-the-shelf classical algorithms and those that take advantage of problem structure. The first class of functions we analyze exhibits exponential advantage over off-the-shelf global optimization algorithms, but can be optimized by efficient algorithms that take advantage of problem structure. We then show how to maintain the polynomial quantum runtime while removing this global structure. This yields a family of nonconvex functions that can be provably optimized in polynomial time by a quantum algorithm, but cannot be classically optimized in polynomial time by any algorithm we are aware of. In the next subsections, we specify our problem setup and then describe the main technical results.

\subsection{Setting up the Problem}
\paragraph{Problem Definition:}
We define the problem of unconstrained continuous optimization in the following standard form.
\begin{problem}[Unconstrained Continuous Optimization]
\label{prob:optimization}
Let $f : \mathbb{R}^d \rightarrow \mathbb{R}$ be a continuous function and $\mathcal{X}^{\star} \neq \emptyset$ the set of global minimizers. Given an input point $x_0 \in \mathbb{R}^d$ and $R \in \mathbb{R}_+$ such that $\ell_2(x_0, \mathcal{X}^{\star}) \leq R$, find a point in $\tilde{x} \in \mathbb{R}^d $ such that
\begin{align*}
    f(\tilde{x}) - f(x^{\star})\leq \epsilon, \quad \forall x^{\star} \in \mathcal{X}^{\star}.
\end{align*}
\end{problem}
We remark that although our goal is to solve unconstrained optimization problem above, in the actual implementation of the algorithm, we typically use a sufficiently large bounded domain $\mathcal{X}\subset \mathbb{R}^d$ that contains the global minima. Therefore, we will work on a bounded domain $\mathcal{X}$ unless stated otherwise.

For our quantum algorithms, we will assume access to the function via a quantum binary evaluation oracle.
\begin{restatable}[$\epsilon_f$-accurate binary oracle]{definition}{BinaryOracleDefn}
\label{defn:binary_oracle}
    Let $f : \mathcal{X} \rightarrow \mathbb{R}$, where $\mathcal{X} \subseteq \mathbb{R}^d$. The unitary $O_{f}$ is an $\epsilon_f$-accurate binary oracle for $f$ if for all computational basis states $\ket{x}\ket{y}$
    \begin{align*}
        O_f\ket{x}\ket{y} = \ket{x}\ket{y\oplus \widetilde{f}(x)},
    \end{align*}
    and $\lVert \widetilde{f}(x) - f(x)\rVert_{\infty} < \epsilon_f$. 
\end{restatable}
Unless otherwise specified we will use $G$ to denote the Lipschitz constant of the cost function $f$, i.e., $G$ is a parameter such that $\norm{\nabla f(x)}_2 \le G$ for all $x \in \Xcal$. Note that in general, we allow $G$ to depend on $d$ which will typically be the case for fast growing functions in $\R{d}$.

\paragraph{Fixing the Quantum Algorithm:} We wish to understand the mechanisms that separate quantum and classical dynamics for nonconvex optimization. In order to make a fair apples-to-apples comparison against both black-box and structure-aware classical algorithms, we investigate exactly the same quantum algorithm throughout the paper, irrespective of the particular function class under investigation. The algorithm we will consider is an adiabatic annealing algorithm that tracks the ground state of a parameterized Schr\"{o}dinger operator. We call this algorithm the \emph{real space adiabatic algorithm} and describe it as follows.
\begin{definition}
    [Real-space Adiabatic Algorithm (RsAA)]
    \label{defn:schrodinger-anneal}
   The Real-space Adiabatic Algorithm for a cost function $f \colon \Xcal \to \R{}$  with bounded domain $\Xcal$ tracks the ground state of the Schr\"odinger operator
   \begin{equation*}
       H(\lambda) = -\Delta + \lambda^2 {f}, \qquad \lambda \in [0,\lambda_{\max{}}],
   \end{equation*}
   where $\Delta$ is the Dirichlet Laplacian for $\Xcal$.
\end{definition}
The algorithm and its complexity are analyzed in detail in Section~\ref{sec:algorithm_analysis}. We state the main result on its complexity below. We note that the complexity below is stated in terms of queries to the binary oracle $f$, however for simplicity we will often refer to the runtime of the algorithm. This is because for sufficiently complex functions, the number of function queries is the largest contributing factor to the total runtime or gate complexity and determines their asymptotics.
\begin{theorem}[Theorem \ref{thm:adiabatic_simulation} Informal]
\label{thm:adiabatic_simulation_inform}
     Suppose $f : \mathcal{X}\rightarrow \mathbb{R}$ is $G$-Lipschitz function with a bound $\Lambda \geq \lVert f \rVert_{\infty}$ and  has support in a compact domain $\mathcal{X}$ with $x_0 + [-2R, 2R]^d \subseteq \mathcal{X}$, and that for $\lambda \in [0, \lambda_{\max}]$, the spectral gap of $H(\lambda)$ is lower bounded by $\delta_{\min}$. There is a digital quantum algorithm (Algorithm \ref{alg:adiabatic_schro}) that starts from the discrete uniform superposition and simulates the adiabatic evolution \eqref{defn:schrodinger-anneal}
     using  $\mathcal{O}\left(\textup{poly}(d, \lambda_{\max}, \Lambda, 1/\delta_{\min}, 1/\rho)\right)$ queries to an $\epsilon_f = \widetilde{\mathcal{O}}\left(\textup{poly}(\rho, \delta_{\min}, 1/\lambda_{\max}, 1/\Lambda)\right)$ accurate binary oracle, $n=\Ocal\left(d^2\cdot \polylog(1/\epsilon, 1/\rho, R, G, \lambda_{\max}, \Lambda)\right)$ qubits and $\widetilde{\Ocal}\left(\poly(d, \lambda_{\max} ,\Lambda, 1/\rho)\right)$ gates, and outputs an $n$-qubit quantum state $|\Psi\rangle$ such that \begin{align*}
        \mathbb{P}_{|\Psi\rangle}[ \lVert X - y\rVert < \epsilon] > \mathbb{P}_{\lvert \Phi_{\lambda_{\max}}\rvert^2} \left[ \lVert X - y\rVert < \frac{\epsilon}{2} \right] - \rho,
    \end{align*}
    for any fixed $y \in \mathcal{X}$, where $\Phi_{\lambda_{\max}}$ is the ground state of $H(\lambda_{\max})$.
\end{theorem}
For a digital state, we define $\mathbb{P}_{|\Psi\rangle}$ to be the measurement distribution in the computational basis, where basis states have been mapped to grid points (See Section \ref{sec:quantum_runtime_separable} for more details). For a continuous wave packet, $\mathbb{P}_{\lvert \Phi_{\lambda_{\max}}\rvert^2}$ is the probability distribution in the position basis.

The runtime of RsAA is usually dominated by the  inverse spectral gap $\delta_{\text{min}}^{-1}$ throughout the adiabatic trajectory, which will be analyzed theoretically in later sections to demonstrate the polynomial runtime of RsAA for certain classes of objective functions. We note that we intentionally do not make any problem specific modifications to the RsAA algorithm. Thus the polynomial dependencies in the query complexity and precision requirements stated in Theorem~\ref{thm:adiabatic_simulation_inform} are probably quite pessimistic, since the adiabatic schedule and the analysis are not optimized. We make this choice for conceptual simplicity in presenting the main mechanisms for speedup. Additionally, because we are primarily concerned here with large or even super-polynomial speedups, the exact degree of the polynomial cost of RsAA is not significant. Still, the formal version of Theorem \ref{thm:adiabatic_simulation_inform} does enable some further improvements in performance, which we do take advantage of in later sections. It remains of interest to carefully analyze the RsAA algorithm in order to apply it for other practical problems.

\paragraph{Classical Algorithms for Comparison:} We do not make any attempt in this paper to prove query complexity separations between our quantum algorithms and \emph{all} classical algorithms. In fact, our goal is primarily to characterize mechanisms for advantage that are not clearly characterized by complexity theoretic arguments. We will instead take the approach of considering a \emph{specific} broad class of classical algorithms that are often used as benchmarks and investigate the cost of optimizing the relevant function classes using these algorithms. There are two significant classes of algorithms:
\begin{enumerate}
    \item \textbf{Black-box or off-the-shelf algorithms} are fixed algorithms that only require black-box access to the function and can be applied to almost any objective function. Examples of black-box optimization algorithms include first-order local optimization algorithms, such as (accelerated)~gradient descent~\cite{nesterov_acceleration,nesterov2018lectures,boyd2004convex}, Adam~\cite{kingma2014adam} and RMSProp~\cite{tieleman2012lecture}; derivative-free algorithms such as COBYLA and BOBYQA~\cite{powell2007view}; quadratic programming methods such as Sequential Quadratic Programming~\cite[Chapter 18]{nocedal2006numerical}; quasi-Newton methods such as BFGS~\cite{broyden1970convergence,fletcher1970new,goldfarb1970family,shanno1970conditioning}, L-BFGS~\cite{liu1989limited} and L-BFGS-B~\cite{byrd1995limited,zhu1997algorithm}; and interior-point methods such as IPOpt~\cite{wachter2006implementation}. This class also includes algorithms that are specifically designed for global optimization of nonconvex functions; such as perturbed/stochastic gradient-descent~\cite{jin2017escapesaddlepointsefficiently,jin2018accelerated}, hybrid local-global schemes such as basin-hopping~\cite{wales1997global}, simulated annealing algorithms such as dual-annealing~\cite{xiang1997generalized,xiang2000efficiency}, and genetic algorithms such as differential evolution~\cite{storn1997differential}. Finally, there are also exact algorithms based on branch-and-bound methods such as the commercial MIP solvers Gurobi~\cite{gurobi} and CPLEX~\cite{cplex2009v12}.
    \item \textbf{Structure-aware algorithms} specifically exploit the global structure of the function classes under consideration, such as separability or local convexity. We analyze two such algorithms: a convexity-honing algorithm that combines low-temperature Langevin dynamics with gradient descent, and a Hessian-based algorithm that recovers hidden rotational structure.
\end{enumerate}
Not all of these classical algorithms can be analyzed in a fully rigorous manner. Nevertheless, we investigate a broad selection of state-of-the-art classical algorithms spanning diverse optimization paradigms, which we argue serve as representative benchmarks for analyzing the performance in global optimization. 

In order to theoretically motivate quantum/classical comparisons, we will particularly consider classical algorithms that are based on a continuous-time classical dynamics known as the \emph{Langevin Diffusion} with variable noise rate, which is described by the Stochastic Differential Equation
\begin{align}
    \label{eq:learning-rate-sde}
    \mathrm{d}X_t = -\nabla f(X_t)\mathrm{d}t+\sqrt{s}\mathrm{d}B_t,
\end{align}
where $f \colon \R{d} \to \R{}$ and $B_t$ is the Wiener process in $\R{d}$. Setting $s=2$ yields the usual Langevin diffusion, which provides an insight into the behavior of a classical particle under deterministic forces dictated by the potential $f$ and random fluctuations. In this sense, Langevin diffusion captures the mechanics of classical systems at a finite temperature making it a natural classical analogue to compare with Schr\"odinger evolution.

It is well known that the stationary measure of the Langevin diffusion has a density function proportional to $\exp(-2 f(x)/s)$. Therefore, setting $s = 2/\beta$ makes the stationary measure equal to the Gibbs measure with potential $f$ and inverse temperature $\beta$. It can be shown under mild assumptions on $f$ that a sample from the Gibbs measure at $\beta \sim \Ocal(d/\epsilon)$ (Lemma \ref{lem:beta-lower-bound}) %
suffices to solve the optimization problem with constant probability. Thus the Langevin diffusion gives us a flexible black-box continuous-time optimization framework, which has been the subject of significant theoretical study. Furthermore, it can be seen as the continuous-time limit of well known discrete-time optimization algorithms such as Stochastic/Perturbed Gradient Descent~\cite{shi2023learning}. To be more specific, consider the simple Euler-Maruyama discretization of~\eqref{eq:learning-rate-sde} with step size $s$,
\begin{align}
\label{eq:sgld-update}
x_{k+1} = x_k - s \nabla f(x_k) + s \xi_k
\end{align}
where $\xi_k$ is an isotropic Gaussian noise term. Then, one can simulate the SDE by iterating~\eqref{eq:sgld-update}. The resulting algorithm gives the more standard form of stochastic gradient descent (SGD) with learning rate $s$ and Gaussian noise. The primary result we will use on the Langevin diffusion is the following:
\begin{restatable}[Theorem 2 \cite{shi2023learning}]{theorem}{langsgdgap}
\label{thm:lang_sgd_gap}
Suppose $f$ is confining (see Definition \ref{defn:confining_func}), \begin{align*}
    \forall \beta > 0, \lim_{\lVert x \rVert \rightarrow \infty}\left(\beta\norm{\nabla f(x)}^2 - \Delta f(x)\right) = \infty,
\end{align*} $f$ has at least two local minima, and $f$ is a Morse function. Then there exists a $\beta_0$, depending on $f$, such that $\forall \beta > \beta_0$, the inverse-relaxation time $\delta^{(C)}(\beta)$ of Langevin dynamics under $f$ satisfies
\begin{align*}
    \delta^{(C)}(\beta) =(\alpha + o(1/\beta))e^{-\beta H_f},
\end{align*}
where $H_f$ is a quantity known as the Morse saddle barrier that depends only on the geometry of the function $f$, and $\alpha >0$ is a constant depending on $f$.
\end{restatable}
Roughly, one can view the value of $H_f$ as the height of the largest barrier in $f$. Hence if, as discussed, we need $\beta = \mathcal{O}(d/\epsilon)$ for global optimization, then the above leads to an exponential in $d$ mixing time for Langevin dynamics and a corresponding exponential cost for its discretized variants such as stochastic/perturbed gradient descent (SGD). In Section~\ref{sec:sgd_lower_bound} we provide a full technical analysis of the runtime of SGD, formalizing the above intuitions.

Of course, all the candidates for classical optimization cannot be rigorously analyzed through the lens of Langevin diffusion. For these algorithms, we provide a detailed view in Section~\ref{sec:numerical-benchmarking}. We provide intuitive explanations for the bottlenecks faced by these algorithms, which we supplement with numerical benchmarks for the global optimization methods, Basin-hopping (Section~\ref{sec:basin-hopping-benchmarking}), Dual-Annealing (Section~\ref{sec:dual-annealing-benchmark}), and Differential-Evolution (Section~\ref{sec:differential-evolution-benchmark}).

The final off-the-shelf algorithms we consider are \emph{Quantum Monte-Carlo} and specifically Simulated Quantum Annealing~(SQA)~\cite{crosson2016simulated,crosson2021prospects} which has been shown to dequantize some of the exponential separations between quantum and simulated annealing. We will provide conceptual arguments in Section~\ref{sec:sqa} that show that such a phenomenon is unlikely for the function classes considered here, and hence, the exponential separations are likely to persist against SQA. These arguments leverage the fact that unlike previous works on quantum algorithms for nonconvex optimization, our results do not depend explicitly on the shape and size of barriers between local minima.

\begin{figure}
    \centering
    \includegraphics[width=0.9\linewidth]{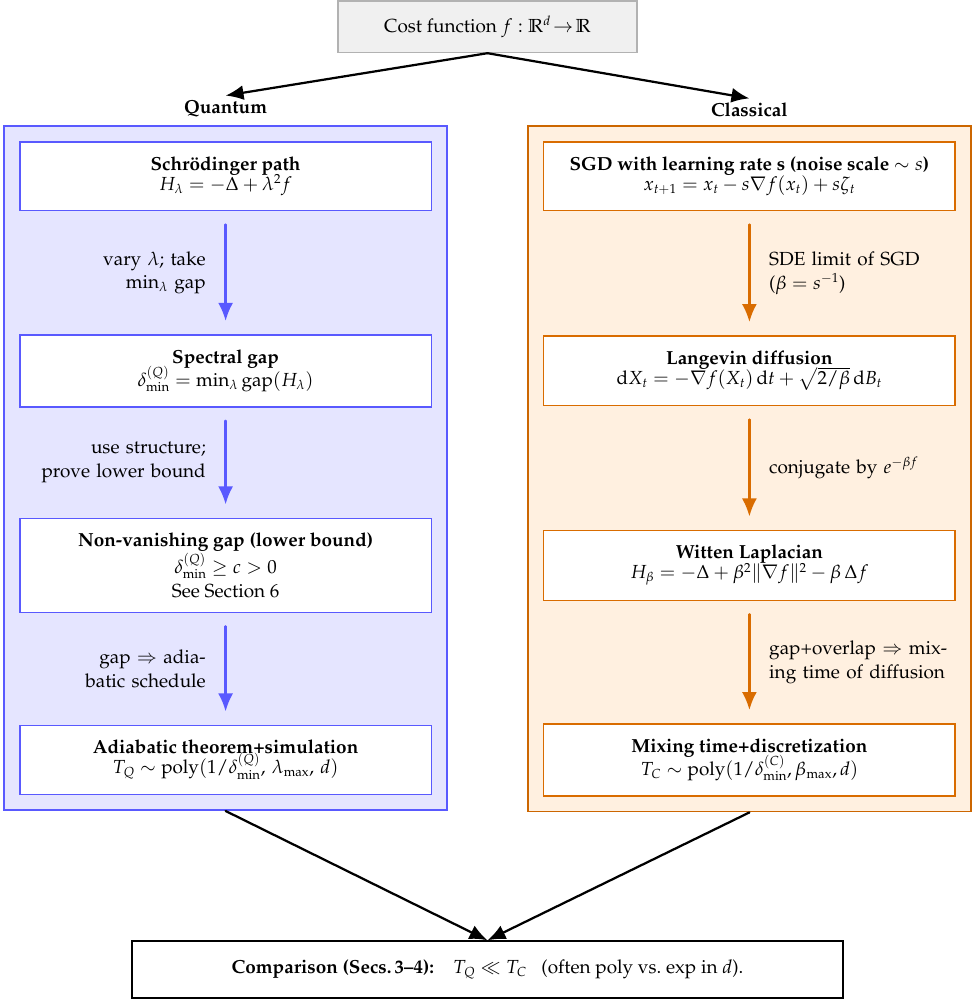}
    \caption{Correspondence between Quantum Dynamics (RsAA) and Langevin Diffusions (SGD)}
    \label{fig:placeholder}
\end{figure}

\subsection{Overview of Results}

\subsubsection{Connecting Schr\"odinger Operators to Classical Diffusion}
Our primary mechanism for separation between quantum and classical dynamics for optimization is based on an explicit two-way correspondence between Langevin diffusion and Schr\"{o}dinger operators, sometimes referred to as Stochastic Quantization \cite{damgaard1987stochastic, brooks2019sharptunnelingestimatesdoublewell}. Specifically, we show that the Schr\"{o}dinger operator is spectrally equivalent to the Langevin diffusion that samples from its ground state. In this setup, the negative log-density of the ground state plays the role of an effective potential, which we term as the \emph{ground state potential} for the Langevin diffusion.
\begin{definition}[Ground State Potential]
    \label{def:ground_state_potential}
    Let $H = -\Delta + f$ be a Schr\"{o}dinger operator in $d$-dimensions whose ground state is given by $\psi_1$. The function $\ground{f} \colon \R{d} \to \R{}$ defined by $\ground{f} = -\log \psi_1^2$ is defined as the ground state potential corresponding to $f$.
\end{definition}
The mixing time of the Langevin diffusion and the running time of RsAA are determined by the spectral analysis of the corresponding operators. As a consequence in order to characterize the running time of quantum annealing, it is sufficient to consider the Langevin diffusion for the ground state potential. Specifically, we have the following statement:
\begin{restatable}{proposition}{InformalQAToSGD}
\label{prop:informal-qa-to-sgd}
    The minimum spectral gap of $H(\lambda) = -\Delta + \lambda^2 f$ for $\lambda \in [0,\lambda_{\max}]$ is $\Omega(\delta_{\min})$ if and only if the relaxation time of the Langevin diffusion for the ground state potential $\ground{f_\lambda}$ for all $\lambda \in [0,\lambda_{\max}]$ is $\Ocal(\delta_{\min}^{-1})$.
\end{restatable}
The above proposition allows us to show a quantum advantage via RsAA if we can identify a potential that is difficult to optimize via classical methods, but for which the ground state potential for all relevant $\lambda$ leads to a rapidly relaxing Langevin diffusion. In other words, the ground state potential completely determines the performance of the adiabatic algorithm and characterize the speedup. Intuitively, it has been observed that the ground state of a Schr\"{o}dinger operator has better smoothness properties than the potential itself which makes it more amenable to local algorithms such as those based on Langevin diffusion. In fact, it has been argued that the eigenfunctions of a Schr\"{o}dinger operator only see a regularized version of the original potential~\cite{steinerberger2021regularized}. This behavior can also be observed numerically in low dimension (See figure~\ref{fig:strongly-convex-3d}).  In the general case, this is not directly helpful in a quantitative sense, as computing the ground state potential requires us to effectively solve the eigenvalue problem for the Schr\"{o}dinger operator, however it is possible analyze the ground state potential for specific class of potentials.

We also establish this correspondence in the other direction. We show that the Langevin diffusion for a potential $f$ at inverse-temperature $\beta$ is spectrally equivalent to a Schr\"{o}dinger operator called the \emph{Witten Laplacian} given by $-\Delta + \beta^2 \norm{\nabla f}^2 - \beta\Delta f$. We call the potential term $\wkb{f_\beta} = \beta^2 \norm{\nabla f}^2 - \beta\Delta f$ in the above operator the \emph{WKB effective potential}, since the ground state of this operator is given by the WKB ansatz $e^{-\beta f}$ as shown in Section \ref{sec:backward-wkb}. This yields an alternative interpretation of the same mechanism for quantum speedup:
\begin{restatable}{proposition}{InformalSGDToQA}
\label{prop:informal-sgd-to-qa}
    The mixing time of the Langevin diffusion for $f \colon \R{d} \to \R{}$ is at least $\Ocal(\delta^{-1})$ where $\delta$ is the spectral gap of the Witten Laplacian $-\Delta + \beta^2 \norm{\nabla f}^2 - \beta\Delta f$.
\end{restatable}
We can therefore identify potentials for which algorithms based on Langevin diffusion are slower than RsAA if the Witten Laplacian can be demonstrated to have a smaller spectral gap than the original Schr\"{o}dinger operator. This complements the intuition that the effective ground state potential has a different geometry than the original potential.  

\subsubsection{Mechanism: The Impact of Unique Global Minima}

Based on the above discussion, we present a mechanism that separates the RsAA algorithm from classical algorithms based on the Langevin diffusion. It is well-known in the mathematical physics literature \cite{simon1983semiclassical,simon1984semiclassicaltunneling} that in the semiclassical regime $(\lambda \to \infty)$, there is a significant difference between potentials with a unique global minimum and those with multiple global minima. In particular, as $\lambda \to \infty$ the spectrum of a potential with a unique global minimum approaches that of a quantum harmonic oscillator defined by the quadratic approximation at the global minimum. Therefore, the spectral gap is eventually non-decreasing as $\lambda \to \infty$. In contrast, when there are multiple global minima, we are in the so called \emph{tunneling regime}. In this case, the spectral gap as $\lambda \to \infty$ takes the form $\exp(-\lambda S_f)$, where $S_f$ is a quantity that depends only on the function and captures the effect of instantons %
tunneling between the degenerate minima \cite{Mario2015}. As a consequence, the gap is monotonically decreasing with $\lambda$ as $\lambda \to \infty$. 

Regardless of whether the global minimum is unique or not, the spectral gap of the Langevin diffusion takes the form $\exp(-\beta H_f)$ where $H_f$ depends only on the cost function and monotonically decreases as $\beta \to \infty$. This behavior can be understood via the WKB effective potential described above: 
\begin{align*}
    \wkb{f_\beta} = \beta^2 \norm{\nabla f}^2 - \beta\Delta f.
\end{align*}
Note if $x$ is a critical point of $f$,  then $\wkb{f_\beta}(x) = -\beta\Delta f(x)$, and hence all minima are at first-order critical points of $f$. This also implies that all critical points with, roughly, the same ``curvature'' will be indistinguishable in $\wkb{f_\beta}$. Specifically, as long as there are two critical points with the same maximal $\Delta f$, $\wkb{f_\beta}$ will have multiple global minima.
As a consequence, the WKB effective potential for large $\beta$ increasingly resembles a potential with non-unique global minima, even if $f$ had a unique global minimum. The Witten Laplacian is therefore in the tunneling regime for $\beta \to \infty$, leading to the gap falling exponentially with $\beta$.  %

To see how this affects optimization, we note that in the general case, we need to choose $\beta = \Theta(d)$ for optimization. This leads to a classical gap estimate of $\exp(-d H_f)$ which is exponentially falling in $d$. For RsAA, we must also choose $\lambda = \mathrm{poly}(d)$. However, when $f$ has a unique global minimum the gap does not fall exponentially with $\lambda$, so the spectral gap in the semiclassical regime is not exponentially small in $d$. We note, that the spectral gap for some intermediate $\lambda$ \emph{can} be exponentially small in $d$, and therefore this alone does not constitute a proof of quantum advantage. However, if we can lower bound the gap at intermediate $\lambda$ by different arguments then we do have a proof of advantage over Langevin diffusion which remains obstructed due to the gap of the Witten Laplacian exponentially falling with $d$. Therefore, to analyze the runtime of RsAA algorithm, one first need to understand the regime in which the semiclassical approximation is valid and then analyze the earlier phase in the evolution where avoided crossings can potentially introduce exponentially small spectral gaps. Our applications are specifically chosen to understand such behavior of the spectral gap throughout the adiabatic evolution. Our analysis techniques may be an independent tool for understanding spectral gap of more general quantum Hamiltonians.

\subsubsection{Application 1: Block Coordinate-Separable Functions}
We now construct a set of functions for which the quantum runtime can be bounded by extensions of semiclassical arguments. Our construction is a rigorous and generalized version of an example constructed by Leng et al. \cite{leng2023quantum} to show a performance separation between quantum and classical algorithms for nonconvex optimization. Our generalization is to \emph{rotated block separable functions}: $f: \mathbb{R}^d \rightarrow \mathbb{R}$ that satisfy the property that there exists a rotation $U \in \text{SO}(d)$ s.t. $f(Ux) = \sum_{i=1}^{k} g_i (\hat{x}_i)$, where the $g_i : \mathbb{R}^{d_i} \rightarrow \mathbb{R}$ depend on disjoint subsets of coordinates of $f$. The separable structure allows for the gap of the $d$-dimensional cost function to be analyzed by studying the gaps of $k$ potentials of lower dimension $d_i, i =1, \dots, k$. Leng et al. consider the rotated \emph{completely separable} setting where each $g_i(\hat{x}_i)$ is a one-dimensional, asymmetric double-well potential obeying some regularity conditions. Since the gap of the operator cannot be bounded by purely asymptotic semiclassical arguments,~\cite{leng2023quantum} only provided numerical evidence for the gap in finite $\lambda$ regime. Furthermore, no theoretical analysis is provided for separation from classical algorithms. 

We provide in Section \ref{sec:semiclassical_analysis} a completely rigorous analysis of the spectral gap and runtime via the first \emph{non-asymptotic} versions of the semiclassical arguments of Simon~\cite{simon1983semiclassical}, which may be of independent interest. We also characterized $\lambda_{\text{max}}$ under milder assumptions to show the convergence of RsAA to the global minimum of $f$ without using Agmon's estimate which causes problems for general $d$ dimensional functions as Agmon's theorem involves function related constants that can potentially depend on $d$ other geometric properties of $f$. In particular, our proof for $\lambda_{\text{max}}$ does not require separability assumption and applicable in more general settings. These results are used (in Section \ref{sec:quat_clas_block_sep}) to show that RsAA can optimize block-separable functions with \emph{constant-sized blocks}, i.e. $\forall i, d_i = \mathcal{O}_d(1)$, in polynomial time (Theorem \ref{thm:quant_runtime_for_block_sep}). We provide a detailed theoretical analysis that perturbed/stochastic gradient dynamics require exponential time to optimize this class of functions~(Section~\ref{sec:sgd_lower_bound}).

 We note that several benchmark functions that have been observed to be hard to optimize by off-the-shelf classical algorithms are separable. These observations are supplemented by our numerical benchmarks in Section~\ref{sec:numerical-benchmarking}. These functions include the Levy function and the Rastrigin function. Our analysis provides rigorous evidence that an off-the-shelf quantum algorithm (RsAA) optimizes each of these functions in polynomial time.

\paragraph{Structure-aware Classical Algorithms:} Despite these separations the construction above has the limitation that it can be optimized in polynomial time by algorithms that are specifically designed to take advantage of the function structure. However, the degree of polynomial in the runtime as a function of $d$ can be arbitrarily-large for some coordinate block-separable functions. The degree of the polynomial in the quantum runtime is upper bounded by a universal constant for all functions in the class. Hence, for optimizing block-separable functions, there is still an \emph{arbitrarily-large polynomial} separation between RsAA and all of the classical algorithms that we consider. We analyze two structure-aware classical algorithms:
\begin{enumerate}
    \item \textbf{Convexity-Honing Algorithm}: A local algorithm that takes advantage of the local convexity of separable functions in a $\ell_\infty$ norm ball around the global minimum. This allows us to run a Langevin diffusion at a lower inverse temperature of $\beta \sim \log(d)$ followed by a deterministic gradient algorithm initialized at a sample from the diffusion. While the runtime of this algorithm is polynomial in $d$, the exponent of the polynomial can be shown to be function dependent, and can be arbitrarily high depending on the function. %
    \item \textbf{Hessian Algorithm}: A global algorithm that uses some number of function evaluations to find $U$ and rotate the function so that it can be expressed as $\sum_{i=1}^d g_i(\hat{x}_i)$. The rotation can then be inverted and the resulting function can be optimized by optimizing each component function to optimality separately. %
\end{enumerate}

\subsubsection{Application 2: Beyond Coordinate-Separability}

The functions described in the previous section provide evidence of quantum advantage against off-the-shelf classical algorithms, but can be optimized in polynomial (albeit with arbitrarily large degree in $d$ and/or $1/\epsilon$) by structure-aware algorithms. We now seek to identify functions for which we can analyze the spectral gaps without invoking the global structure of coordinate-separability discussed in the previous section.
To this end, we introduce new technical tools that allows us to rigorously demonstrate RsAA to have inverse-polynomial spectral gaps for cost functions that are hard to optimize. We have already discussed that the spectral properties of the function under consideration are determined by the ground state potential. We can therefore consider adding perturbations to a function with a well-behaved ground state potential so as to maintain the properties of the ground state while eliminating global structure and further obstructing classical algorithms.

Our primary mathematical ingredient for this argument uses recent results by Gross~\cite{gross2025invariance} regarding the invariance of the \emph{intrinsic hypercontractivity} of Schr\"{o}dinger operators under perturbations of the cost function. The intrinsic hypercontractivity of a Schr\"{o}dinger operator is a property of the corresponding Dirichlet form that is stronger than the spectral gap, and in fact is equivalent to a log-Sobolev inequality~\cite{gross_1975} for the ground state measure. If we can add a perturbation to a family of functions for which $H(\lambda)$ for all relevant $\lambda$ is intrinsically hypercontractive while maintaining the conditions to maintain the hypercontractivity, we have a new function that can be optimized in polynomial time via RsAA. In fact, as a consequence of hypercontractivity we also obtain strong concentration inequalities on the ground state measures, in addition to a spectral gap bound. This gives us sharper bounds on $\lambda_{\max{}}$ than can be obtained from general arguments.

To demonstrate the robustness of ground states in quantum systems, we consider a class of nonconvex functions that can be expressed as a sum of a strongly convex component $h$ and a perturbative term $g$. Additionally, we impose the following growth conditions: the function $h$ grows at least as fast as $C_f\|x\|^{2k}$ for $k\geq 1$ ensuring a strong radial localization,  while the perturbation $g$ grows at most as $\frac{C_g}{d}\|x\|^{k+1}$, limiting its growth compared to $h$. Under these assumptions, we analyze the ground state of the Schrödinger operator with potential $\lambda^2f = \lambda^2(h+g)$ and show that it remains robust to the presence of the perturbation $g$. This robustness is formalized through the property of intrinsic hypercontractivity of the ground state, which in turn guarantees the existence of a large spectral gap and strong concentration properties. These results are discussed in greater depth in Section~\ref{sec:beyond_coordinate_sep_funcs} with some of the technical details deferred to Section~\ref{sec:hypercontractivity}. 

The significance of this finding becomes apparent when considering the limitations of classical optimization methods. The perturbation $g$ is allowed to introduce barriers in the potential landscape that are at least constant in size (i.e., $|H_f| = \Ocal(1)$) for large $x$. In such scenarios, classical stochastic gradient descent (SGD) can encounter multiple local minima that are separated by barriers of size $\Ocal(d)$, especially as the inverse temperature parameter $\beta$ scales with the dimension ($\beta = \Theta(d)$). These barriers can trap classical algorithms, preventing efficient exploration and optimization of the landscape. Indeed, to the best of our knowledge, there is no classical algorithm capable of optimizing such functions in polynomial time.

In contrast, the  quantum ground state maintains hypercontractivity and a robust spectral gap even in the presence of significant nonconvexity. This quantum advantage highlights a fundamental difference in the behavior of quantum and classical systems: while classical methods are hindered by the presence of multiple minima and high barriers, quantum dynamics is not very sensitive to the shape or size of these barriers.

\subsection{Related Work}
In recent years there have been many works on accelerating optimization algorithms by developing quantum analogues of classical sub-components. There are too many results to reasonably list here, so we refer to~\cite{augustino2021quantum,dalzell2023quantum} for a comprehensive survey. More closely related to our work are previous attempts to characterize quantum dynamical mechanisms for nonconvex optimization, which we discuss below:
\begin{itemize}[leftmargin=*]
    \item One of the first works to attempt a theoretical characterization of Schr\"{o}dinger operators for nonconvex continuous optimization is by Leng et al.~\cite{leng2023quantumhamiltoniandescent}, who analyzed an algorithm called Quantum Hamiltonian Descent (which is similar to RsAA except that it is not necessarily adiabatic). The runtime/cost analysis in~\cite{leng2023quantumhamiltoniandescent} is primarily numerical, and in fact some of the benchmark functions evaluated in that paper---including the Levy, Rastrigin, Csendes, Michalewiz, Skyblinker, and Alpine 1 functions---can be proven based on our results, to be optimizable in $\poly(d)$ time by RsAA. Chen et al. \cite{chen2025quantum} proposed an open-system variant called Quantum Langevin Dynamics which is based on the quantum master equation, and prove similar asymptotic convergence results as~\cite{leng2023quantumhamiltoniandescent}, relying primarily on numerics for algorithmic bounds.
    \item A further step towards a theoretical separation was taken by~\cite{leng2023quantum} who proposed a concrete class of rotated separable cost functions and argued for a polynomial quantum runtime. The arguments here were based on numerically justified assumptions on the spectral gap. We expand on this work in several ways: firstly, we give concrete and rigorous theoretical bounds on the spectral gap that holds for a generalized class of block co-ordinate separable functions that contains as a subset all the functions analyzed in~\cite{leng2023quantum}. Secondly, we identify a mechanism for speedup in terms of the uniqueness of the global minimum in the ground state versus the WKB effective potential. Finally, we provide an analysis of off-the-shelf classical algorithms such as SGD and SQA, as well as of structure-aware algorithms. We note also that the generalization of the function class is necessary even to preserve a separation against off-the-shelf global optimizers; as we will argue in Section~\ref{sec:classical-algorithms} the \emph{basin-hopping} algorithm with appropriate choice of parameters can efficiently optimize the functions considered in~\cite{leng2023quantum}.
    \item A recent preprint by Leng et al.~\cite{leng2025sub} describes a sub-exponential query speedup for optimization. However, their model differs fundamentally from ours in several important respects. Their proof proceeds via reduction from a known separation between stoquastic adiabatic evolution and classical algorithms for finding an exit node in a graph. In their setup, the exit node is encoded in the cost function while the graph structure is embedded in the adiabatic path that locates this node. In contrast, our quantum algorithm operates in the black-box setting where function access occurs solely through evaluation oracles. Consequently, their sub-exponential separation does not apply to our model, since their problem structure resides in the adiabatic path construction rather than in the cost function oracle itself (which merely encodes the exit node label). This makes our results not directly comparable to theirs. Moreover, our work has a complementary philosophy to~\cite{leng2025sub}. Rather than seeking provable speedups against \emph{all} classical algorithms for a highly specialized function class, we identify mechanisms for provable speedup against \emph{specific well-studied} classical algorithms for broad and natural function classes that include benchmarks of independent interest.
    \item Finally, we discuss the work of Liu et al.~\cite{Liu_2023} who also study a mechanism for quantum speedup in nonconvex optimization based on an algorithm called the Quantum Tunneling Walk. There are two primary differences from our work: firstly, Liu et al. require all the local minima to be approximately equal. Therefore, there is no direct challenge in global optimization itself, as all local minima are nearly optimal globally. The claim of advantage instead, is in exploring the set of local minima by tunneling from one to the other. This highlights the second difference from our work, as \cite{Liu_2023} considers primarily the tunneling regime and the mechanism for speedup is based primarily on a difference between classical and quantum tunneling amplitudes. In contrast, the mechanisms in our paper do not directly consider tunneling amplitudes or the shape of the potential barriers between stationary points.
\end{itemize}

\section{Mechanism: Unique Global Minimum}
\label{sec:main-mechanism}
In this section, we present the separation mechanism between the RsAA algorithm and classical optimization algorithms that are based on Langevin diffusion.
\subsection{Connecting Schr\"odinger Operators to Diffusion}
It has long been recognized that Schr\"odinger operators admit a probabilistic interpretation through their relation to stochastic processes and this interpretation has been used widely in physics and applied mathematics for calculations related to stability, tunneling, and scattering phenomena~\cite{simon1984semiclassicaltunneling,functional_integration,AST_1985__132__203_0}. Suppose that for a normalized quantum state $\Psi(x,t)$, we interpret $\rho(x,t)=\Psi^{\dagger}(x,t)\Psi(x,t)$ as the probability density function evolving according to a Markov processes with transition kernel $P$ so that 
\begin{equation}
    \rho(x,t) = \int_{\mathbb{R}^d} \mathrm{d}x'\rho(x',0) P_t(x',x).
\end{equation}
Then $\rho$ satisfies the following Fokker–Planck–Smoluchowski equation (also known as forward Kolmogorov equation),
\begin{align}
\label{eq:eq:fokker-planck}
    \partial_t \rho(x,t) & =  \Delta \rho(x,t)- \nabla \cdot (2 \rho(x,t)b(x,t) )
\end{align}
with the initial conditions $\rho(x,0)=\rho_0(x)$. Equivalently this evolution can also be described as the well known Langevin diffusion following the stochastic differential equation (SDE):
\begin{equation}
\label{eq:langevin-diffusion}
    \mathrm{d}X_t = b(x,t) \mathrm{d}t+\sqrt{2}\mathrm{d}B_t 
\end{equation}
where $ b(x,t)$ corresponds to drift vector field and $B_{t\geq 0}$ is the standard Brownian motion. Solving for the drift term by considering the continuity equation gives an alternative formalism to quantum mechanics~\cite{AST_1985__132__203_0}. The connection between Schr\"odinger and Stochastic evolution is referred as Stochastic Quantization in quantum-field theory \cite{damgaard1987stochastic}.

In the next sections, we will show that this connection implies that Schr\"odinger operator $H$ is spectrally equivalent to the dynamics governed by equation~\eqref{eq:langevin-diffusion} with drift term $b = \nabla \log \psi_1^2$ up to an isometric transformation. Importantly, this correspondence allows spectral quantities of $H$---such as the spectral gap---to be translated into well studied probabilistic notions such as mixing times, hypercontractivity, and concentration estimates for the associated Markov process. In this way, methods from stochastic analysis can provide a framework for studying both spectral and regularity properties of ground states.
For example, a long standing fundamental gap conjecture (fundamental gap theorem as of 2011) for Schr\"odinger operators with convex potentials on a convex domain can also be proven using probabilistic methods such as reflection coupling of certain SDEs \cite{gong2015probabilistic}, simplifying the original proof of Andrews and Clutterbuck~\cite{andrews2011prooffundamentalgapconjecture}. Similarly, Steinerberger~\cite{steinerberger2021regularized} used the analysis of the hitting time of Feynman-Kac type stochastic processes to argue that the Schr\"odinger operators have a smoothing effect on the potential which is linked to the robustness of the quantum ground states.

It is also crucial to note that the spectral properties of the Schrödinger operator—such as its eigenvalues and eigenfunctions—are fundamentally determined by the ground state wavefunction. This is because the generator of the Langevin diffusion process, as described in~\eqref{eq:langevin-generator}, will be entirely specified by $b=\nabla \log \psi_1^2$, and thus the ground state, $\psi_1$, encodes the full spectrum of the quantum system. Consequently, one may also ponder whether it is possible to simulate~\eqref{eqn:ground_state_langevin} when the quantum gap is $\Omega\left(\frac{1}{\text{poly}(d, \lambda)}\right)$. This is because in principle one could simulate the SDE by an appropriate discretization scheme provided an oracle to the gradient of ground state potential. In fact, this result provides a continuous analog of the results of Bravyi et al.~\cite{Bravyi_2022} (see also \cite[Theorem 30]{huang2024certifyingquantumstatessinglequbit}) that demonstrate that sampling from the ground states of gapped discrete Hamiltonians is only polynomially harder than computing its amplitudes. However, constructing such an oracle is expected to be as hard as solving the eigenvalue equation of the Schr\"odinger operator itself. Hence, if the SDE with the drift vector field $\nabla \log \psi_1^2$ has better spectral properties than the SDE with the drift term $\nabla f$, one can possibly exploit this by simulating the Schr\"odinger equation on a quantum computer. Next, we make the connection to Langevin dynamics more rigorous in the following subsections.

\subsection{Forward Direction: Ground State Transformation}
\label{sec:forward_ground_state_transf}
We first establish the connection to Langevin diffusion starting from a given Schr\"odinger operator $H$.  We provide the following preliminary for Markov semi-groups and continue introducing additional tools in Section~\ref{sec:hypercontractivity} as needed.
\begin{definition}[Markov Process]
   A Markov process $\mathcal{M} = (\mathcal{L}(P),\mu)$ is a stochastic process $(X_t)_{t\in \mathbb{R}^{+}}$ such that for  any $s, s+t\in \mathbb{R}^{+}$ and a bounded measurable function $f$, 
\begin{align*}
    \mathbb{E}[f(X_{t+s})| \{X_r\}_{r\leq s}] = (P_{t}f)(X_{s}).
\end{align*} 
The measure $\mu$ is called stationary or invariant if for all bounded functions $f$, 
\[
\mu(P_t f) = \mu
\]
and $\{P_t\}_{t\in \mathbb{R}^{+}}$ forms a semi-group on $L^2(\mu)$ and infinitesimal generator of this semi-group is defined as 
\begin{equation*}
    \mathcal{L} f := \lim_{t\downarrow 0}\frac{P_tf -f }{t}.
\end{equation*}
\end{definition}

Recall the SDE in \eqref{eq:langevin-diffusion}, and let $\phi =-\log \psi_1$; it is well known that Langevin diffusion converges to the stationary distribution
 \begin{equation}
    \mu(x)= \frac{\exp(-2\phi(x))}{\int_{\mathbb{R}^d}  \mathrm{d}x \exp(-2\phi(x))},
\end{equation}
when $\phi$ is a confining function, which we define below.
\begin{definition}[Confining Function]
\label{defn:confining_func}
    A function $\phi$ is said to be confining  when $\lim_{\|x\|\to \infty} \phi(x) = \infty$ and for all $s>0$,
\begin{equation*}
  \int_{\mathbb{R}^d}e^{-\frac{2\phi(x)}{s}}\mathrm{d}x<+\infty.  
\end{equation*}
\end{definition}
From the associated Fokker-Planck equation in \eqref{eq:eq:fokker-planck}, it follows that generator of SDE in \eqref{eq:langevin-diffusion} is 
\begin{equation}
    \label{eq:langevin-generator}
    \mathcal{L} = \Delta +\nabla \log\psi_1^2 \cdot \nabla 
\end{equation}

\begin{definition}[Dirichlet Form]
A Markov semi-group $\mathcal{M}=(\mathcal{L}(P), \mu)$ has the corresponding Dirichlet form defined as
\begin{equation}
\label{eq:dirichlet-form}
    \mathcal{E}(f,g) = -\langle f, \mathcal{L} g \rangle_\mu. 
\end{equation}
\end{definition}
An explicit calculation from~\eqref{eq:langevin-generator} gives that the Dirichlet form of the Langevin diffusion is 
\begin{equation}
\label{eq:langevin-Dirichlet}
    \mathcal{E}(f,g) = \int_{\mathbb{R}^d}  \mathrm{d}x\mu(x)\langle\nabla f,\nabla g \rangle.
\end{equation}
The significance of the Dirichlet form is that the spectral gap $\delta$ of the Langevin diffusion can be related to the variance by the Poincar\'e inequality:
\begin{equation}
 \label{eq:poincare-spectral-gap}
  \mathrm{Var}_{\mu}[f] = \frac{1}{\delta} \mathcal{E}(f,f).  
\end{equation}
To be able to use \eqref{eq:poincare-spectral-gap} to estimate the spectral gap of the Schr\"odinger operator $H$, we need to show that $H$ is related to the generator of the Langevin SDE with effective potential $-\log \psi_1^2$.

However, the quadratic form associated with the Schr\"{o}dinger operator $H = -\Delta + V$ is not even a Dirichlet form of~\eqref{eq:langevin-Dirichlet}. That is, the equation
\begin{equation}
    \langle Hf,g \rangle_{m} =  \int_{\mathbb{R}^d} \langle \nabla f, \nabla g\rangle \mathrm{d}m
\end{equation}
does not hold unless $V=0$ where the ${m} = \psi_1^2$ is the probability measure for the ground state $\psi_1$. Nevertheless, we can obtain such an operator $\hat{H}$ from $H$ by using the so-called~\cite{gross2025invariance}. The following theorem shows that a Schr\"odinger equation is spectrally equivalent to Langevin diffusion process up to a unitary transformation.
\begin{restatable}[Ground State Transformation]{theorem}{groundStateTransform}
\label{thm:ground-state-transform}
    Let $H = -\Delta + V$ on $L^2(\mathbb{R}^d)$ with $V \in C^2(\mathbb{R}^d)$ bounded below. Assume $H$ has the ground state $\psi_1 \in C^2(\mathbb{R}^d)$ with a simple eigenvalue $E_1$, i.e. 
    \[
    H\psi_1 = E_1 \psi_1, \quad \|\psi_1\|^2 = 1.
    \]
    Define the measure $m = \psi_1^2$ and the isometry $U: f\mapsto f\psi_1$. Then, the operator 
    \[
    \hat{H}= -U^{-1}(H-E_1)U
    \]
    is self-adjoint in $L^2(m)$ and the generator of the Langevin diffusion 
    \[
          \mathrm{d}X_t = \nabla \log(m)\mathrm{d}t+\sqrt{2}\mathrm{d}B_t
    \]
    on $\mathbb{R}^d$ and the Dirichlet form is 
    $
    \mathcal{E}(f,f) = -\langle f\hat{H},f\rangle_m$.
\end{restatable}
\begin{proof}
    Consider $\psi_1 = e^{-\phi}$. Plugging this to eigenvalue equation, $H\psi_1 = E_1\psi_1$:
    \begin{align*}
        (-\Delta +V)e^{-\phi} &= -\nabla \cdot (-e^{-\phi}\nabla \phi)+e^{-\phi}V\\
        &=(\Delta\phi - \|\nabla \phi\|^2+V)e^{-\phi}\\
        &= E_1e^{-\phi}
    \end{align*}
    which gives the well known WKB equation,
\begin{equation}
\label{eq:wkb-equation}
    -\Delta \phi + \|\nabla \phi\|^2 = V-E_1. 
\end{equation}
Let $\textrm{d}m = \psi_1^2(x)\textrm{d}x$.
Using this, next we can compute the following.
\begin{align*}
\int_{\mathbb{R}^d} \| \nabla (f\psi_1)\|^2 \mathrm{d}x &= \int_{\mathbb{R}^d} \|\psi_1 \nabla f + f\nabla \psi_1 \|^2 \mathrm{d}x=\int_{\mathbb{R}^d} \|\psi_1 \nabla f - \psi_1 f\nabla \phi\|^2 \mathrm{d}x \\
&= \int_{\mathbb{R}^d} \|\nabla f-f \nabla \phi\|^2 \psi_1^2(x)\mathrm{d}x\\
&= \int_{\mathbb{R}^d}(\|\nabla f\|^2 + f^2 \|\nabla \phi\|^2)\psi^2_0(x)\mathrm{d}x-\int_{\mathbb{R}^d}2 (f \nabla f \nabla \phi) \psi_1^2(x)\mathrm{d}x
\end{align*}
where we use the definition of $\psi_1$. The last term can be bounded using integration by parts,
\begin{align*}
    \int_{\mathbb{R}^d}2(f\nabla f\nabla \phi)\psi_1^2(x)\mathrm{d}x &= \int_{\mathbb{R}^d} \nabla\cdot(f^2\psi_1^2(x)\nabla \phi)\mathrm{d}x - \int_{\mathbb{R}^d}f^2\Delta\phi\psi^2_1(x)\mathrm{d}x+2\int_{\mathbb{R}^d}f^2\|\nabla \phi\|^2\psi_1^2(x)\mathrm{d}x\\&=\int_{\mathbb{R}^d}f^2(-\Delta\phi+2\|\phi\|^2)\psi^2_1(x)\mathrm{d}x
\end{align*}
Plugging back into previous equation,
\begin{align*}
\int_{\mathbb{R}^d} \| \nabla (f\psi_1)\|^2 \mathrm{d}x &=  \int_{\mathbb{R}^d}(\|\nabla f\|^2 + f^2 \|\nabla \phi\|^2)\psi_1^2(x)\mathrm{d}x-\int_{\mathbb{R}^d}f^2 (-\Delta \phi +  2\|\nabla \phi\|^2)\psi_1^2(x)\mathrm{d}x\\
&= \int_{\mathbb{R}^d}\|\nabla f\|^2 \mathrm{d}m-\int_{\mathbb{R}^d}f^2 (-\Delta \phi +  \|\nabla \phi\|^2)\mathrm{d}m\\
&= \mathcal{E}_m(f,f)+ \langle f, U^{-1}(E_1-V)Uf\rangle_{m} 
\end{align*}
where the last step follows from WKB equation~\eqref{eq:wkb-equation}. This gives
\[
\mathcal{E}_m(f,f) = \langle \nabla(f\psi), \nabla (f\psi) \rangle +\langle f,U^{-1} (V-E_1)Uf\rangle_{m}.
\]
The first term on right hand-side is 
\[
\int_{\mathbb{R}^d} \langle \nabla(f\psi), \nabla (f\psi)\rangle \mathrm{d}x = -\int_{\mathbb{R}^d} \langle \Delta f\psi,  f\psi\rangle \mathrm{d}x
\]
by integration by parts. Let $$\hat{H} =  \langle \nabla (f\psi), \nabla (f\psi) \rangle +\langle f, U^{-1}(V-E_1)Uf\rangle_{m} = U^{-1}(H-E_1)U.$$ 
By the definition of Dirichlet form~\eqref{eq:dirichlet-form},
\[
\mathcal{E}_m(f,f) = -\langle f, \hat{H} f \rangle_{m}.
\]
Moreover, the generator can be derived from $\hat{H}$,
\begin{align*}
    U^{-1}(H-E_1)Uf &=  U^{-1}(-\Delta+ V-E_1)(e^{-\phi}f)\\
    &=U^{-1}(\psi_1\Delta f+2\psi_1 \nabla \phi .\nabla f+\psi_1 
    f \Delta \phi - \psi_1 f \|\nabla \phi\|^2   +(V-E_1)\psi_1f )\\
    &=\nabla f +2\nabla \phi.\nabla f + f \Delta \phi -f(\nabla \phi)^2+ Vf-E_1f\\
    &= -\Delta f + 2 \nabla \phi \cdot  \nabla f \\
    &= -\Delta f -\nabla \log\psi_1^2\cdot \nabla f\\
    &=-\mathcal{L}.
\end{align*}
This concludes the proof.
\end{proof}

Note that Proposition \ref{prop:informal-qa-to-sgd} follows immediately from Theorem \ref{thm:ground-state-transform}, where $-\log(m)$ is $f_{\lambda}^{(\text{ground})}$ (Definition \ref{def:ground_state_potential}). Specifically, the ground state transformation is an isometry between $L^2(\mathbb{R}^d)$ and $L^2(m)$, and so it preserves the spectrum. It is well-known, that the spectral gap of the generator of a reversible Markov semigroup, like for Langevin diffusion, corresponds to the inverse relaxation time.

\subsection{Backward Direction: WKB Equation}\label{sec:backward-wkb}

In this section, we give the reverse direction of Theorem~\ref{thm:ground-state-transform} in the sense that we derive the associated Schr\"odinger operator associated with the given Langevin SDE. We note that the following theorem is a different form of Theorem~\ref{thm:ground-state-transform}, yet we prove it for completeness and its usefulness.
\begin{restatable}[WKB Potential]{theorem}{wkbpotentialThm}
\label{thm:wkbpotentialthm}
     Let $\phi:\mathbb{R}^d\to\mathbb{R}$ be a smooth function. Define the isometry $U: f\mapsto f e^{-\phi}$. Then, the generator of the Langevin SDE
     \[
     \mathrm{d}X_t = -2\nabla \phi \mathrm{d}t +\sqrt{2}\mathrm{d}B_t
     \]
     is given by $U^{-1}(H-E_1)U$ where $$H = -\Delta + \|\nabla \phi\|^2-\Delta \phi + E_1$$ and $\psi_1=e^{-\phi}$ is the ground state of $H$ with eigenvalue $E_1$.
\end{restatable}
\begin{proof}
    Consider the of infinitesimal generator of the Langevin
dynamics with drift $\nabla \phi$
\begin{equation*}
  -\mathcal{L} =\; \,\Delta - 2\,\nabla \phi \cdot \nabla.
\end{equation*}
Conjugating $-\mathcal{L}$ by the weight $\psi_1 = e^{-\phi}$ produces
a Schr\"odinger-type operator:
\begin{equation*}
  H_0 = -\psi_1\,\mathcal{L}\,\psi_1^{-1}
  =- e^{- \phi}\,\mathcal{L}\,e^{\phi}.
\end{equation*}
By applying to a test function and using the product rule we obtain, 
\begin{equation*}
  e^{-\phi}\Big(-\Delta + 2\,\nabla \phi\cdot\nabla\Big)e^{\phi }
  = -\Delta +  \|\nabla \phi\|^2 -\ \Delta \phi,
\end{equation*}
so that
\begin{equation*}\label{eq:WittenSchrodinger}
  H_0 = -\,\Delta +  \|\nabla \phi\|^2 - \Delta \phi.
\end{equation*}

Following Witten~\cite{Witten:1982im}, we can factor $H$ as a sum of
adjoint first-order pieces. Define
\begin{equation*}
  A_i = \partial_i + \,\partial_i \phi,
  \qquad
  A_i^\dagger = -\,\partial_i + \,\partial_i \phi,
  \qquad i=1,\dots,d.
\end{equation*}
Then
\begin{equation*}
  A_i^\dagger A_i
  = \big(-\partial_i + \,\partial_i V\big)\big(\partial_i + \,\partial_i \phi\big)
  = -\,\partial_i^2 +  (\partial_i \phi)^2 - \,\partial_i^2 \phi,
\end{equation*}
and summing over $i$ yields
\begin{equation*}
  \sum_{i=1}^d A_i^\dagger A_i
  = -\Delta +  \|\nabla \phi\|^2 - \Delta \phi
  = H_0.
\end{equation*}

Since $H_0=\sum_i A_i^\dagger A_i \ge 0$, the ground-state energy is $\ge 0$.
Solving $A_i \phi_0 = 0$ for all $i$ gives the (zero-energy) ground state
\begin{equation*}
  \psi_1(x) \propto e^{-\phi(x)}, \qquad H_0\,\psi_1=0.
\end{equation*}
Letting $H = H_0 + E_1$, conjugating back by $\psi_1^{-1}=e^{ \phi}$ recovers the Langevin generator with identical spectrum:
\begin{equation*}
  \psi_1^{-1} H \psi_1
  = e^{\phi}\,H\,e^{-\phi}
  = -\,\Delta + 2\,\nabla \phi \cdot \nabla
  =- \mathcal{L}.
\end{equation*}
This concludes the proof.
\end{proof}
Proposition \ref{prop:informal-sgd-to-qa} now follows from Theorem \ref{thm:wkbpotentialthm} for basically the same reason that \ref{prop:informal-qa-to-sgd} followed form Theorem \ref{thm:ground-state-transform}, with the directions reversed. The difference is that now the ground state or stationary measure is used to construct the Schr\"odinger operator potential. Together, Propositions \ref{prop:informal-qa-to-sgd} and \ref{prop:informal-sgd-to-qa} lead to a two-way correspondence, via isometry, between Langevin diffusion and Schr\"odinger operators.

\subsection{Separation Mechanism}
\label{subsec:separation_mechanism}
Having established the stochastic machinery, we next use the stochastic quantization introduced in Theorems~\ref{thm:ground-state-transform} and \ref{thm:wkbpotentialthm} to understand the separation between the classical and quantum dynamics from a more mathematical perspective. To be more specific, we compare the quantum Schr\"odinger operator
\begin{align}
\label{eqn:schrodinger_op_lambda}
    H(\lambda) = -\Delta + \lambda^2f(x)
\end{align}
to the generator of the classical Langevin diffusion process
\begin{align}
\label{eqn:classical_lang_poten_dyn}
    \mathrm{d}X_t = -\nabla f(X_t)\mathrm{d}t + \sqrt{\beta}\mathrm{d}B_t
\end{align}
which is the continuous analog of certain classical optimization algorithms, such as SGD. 

To compare both operators on the same footing, we use the spectral equivalence between the Schr\"odinger dynamics and Langevin diffusion, which samples from its ground state due to Theorem~\ref{thm:ground-state-transform}. Specifically, the spectral gap of the operator in \eqref{eqn:schrodinger_op_lambda} is equivalent to the gap of the Langevin diffusion
\begin{align}
\label{eqn:ground_state_langevin}
    \mathrm{d}Y_t = \nabla\log(\psi_{1, \lambda}^2(Y_t))\mathrm{d}t + \sqrt{2}\mathrm{d}B_t,
\end{align}
where $\psi_{1, \lambda}$ is the ground state of $H(\lambda)$. Hence the role of $\nabla f$ has been changed to  the log-density of the ground state for the quantum dynamics. Note that the ground state $\psi_{1, \lambda}$ can be taken to be all positive due to stoquasticity. Now the comparison between the quantum and classical dynamics reduces to comparing the relaxation time of diffusion equations~\eqref{eqn:classical_lang_poten_dyn} and ~\eqref{eqn:ground_state_langevin}. To characterize the relaxation time of the diffusion equations, we use the following theorem on the spectral gap.

As will be shown in Lemmas~\ref{lem:hypercontractive-sufficient-lambda} and \ref{lem:beta-lower-bound}, it is typical to set $\beta$ and $\lambda$ to be $\Theta(d)$ so that both dynamics converge to a distribution where an approximate global minima can be obtained with at least constant probability. Thus, the classical relaxation time will be at least exponentially long in $d\cdot H_f$ where $H_f$ will be the barriers between two local minima and maxima of $f$. Exact characterization of $H_f$ will also be discussed in Section~\ref{sec:sgd_lower_bound} in the context of Morse theory. On the other hand, the size of the barrier in the ground state potential $\log(\psi_{1,\lambda}^2(Y_t))$ can be significantly different than $H_f$ leading to different relaxation time of~\eqref{eqn:ground_state_langevin} highlighting the separation.

Alternatively, one can convert the classical diffusion into a Schr\"odinger operator by Theorem ~\ref{thm:wkbpotentialthm},
\begin{equation}
    \label{classical-wkb}
    \widetilde{H}(\beta) = -\Delta + f^{\wkb{}}_\beta
\end{equation}
where $ f^{\wkb{}}_\beta=\beta^2\lVert \nabla f \rVert^2 - \beta \Delta f$ is the WKB potential. We occasionally refer to the Hamiltonian $\widetilde{H}$ as the Witten Laplacian.
Hence, the classical/quantum comparison can be reduced to comparing the spectrum of the Schr\"odinger operators $H(\lambda)$ and $\widetilde{H}(\lambda)$. 
In the following sections (Sections \ref{sec:quat_clas_block_sep} and \ref{sec:beyond_coordinate_sep_funcs}), we will rigorously show that for certain objective functions $f$, the spectral gap  of Schr\"odinger operator is $\Omega(\frac{1}{\text{poly}(d)})$ for all $\lambda>0$, whereas the spectral gap of $\widetilde{H}$ falls exponentially as $\beta$ increases.

To showcase this difference, we can consider the one-dimensional bi-quadratic function that was used in \cite{leng2023quantum}:
\begin{align}
\label{eqn:quartic_function}
    f(x) = x^4 - (x - 1/32)^2 + c,
\end{align}
where  $c$ is constant chosen to make $f$ equal 0 at its global minimum. To understand the difference between the two potentials $f$ and $f^{\wkb{}}$, it is instructive to look at their critical points in 
Figure \ref{fig:degeneracy-separation-2d}.

\begin{figure}[htbp]
    \centering
    \includegraphics[width=1.0\linewidth]{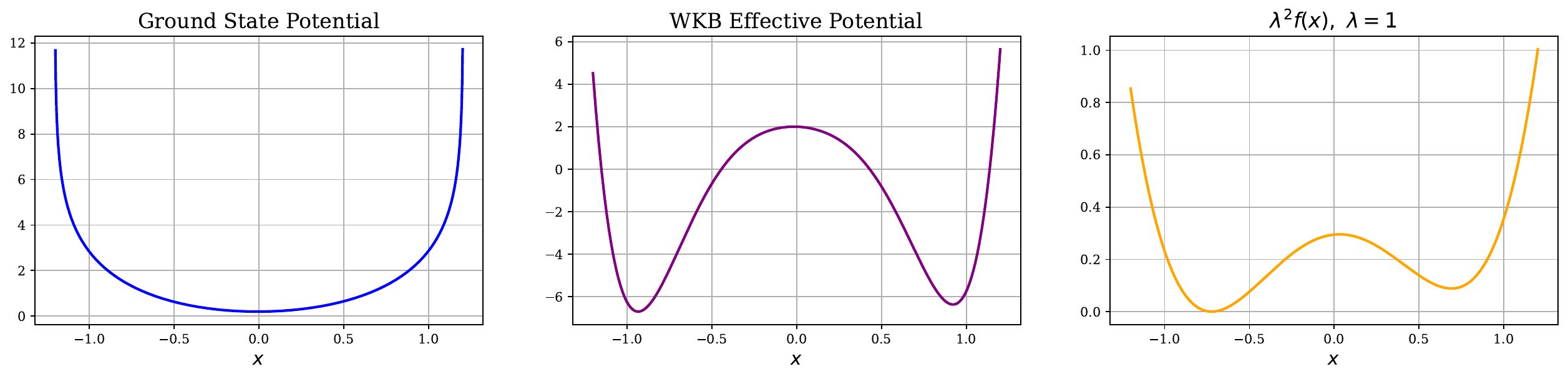}
    \includegraphics[width=1.0\linewidth]{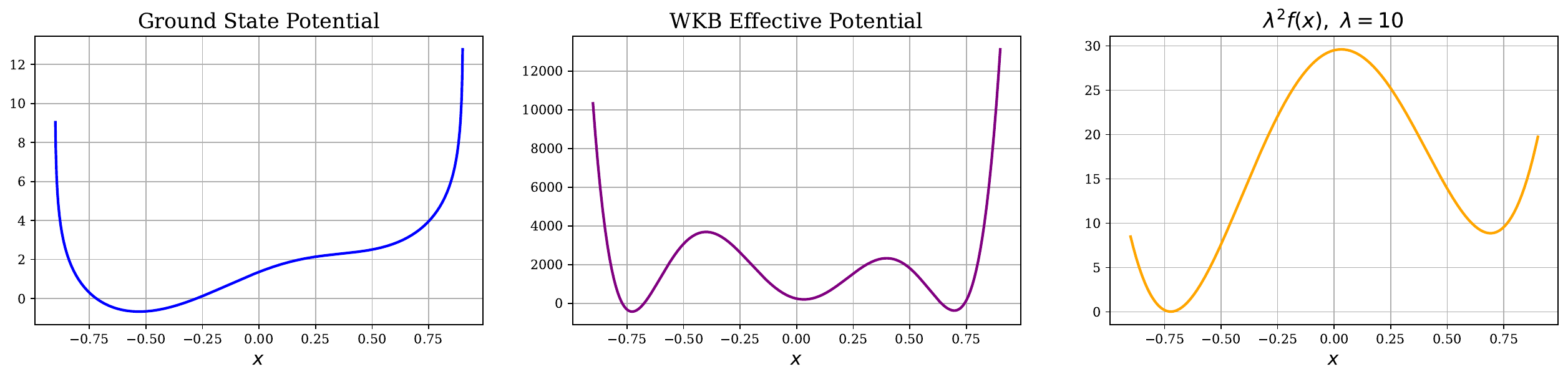}
    \caption{\textit{Visualization of the unique/multiple global minima separation for the potential} $f(x) = x^4 - (x - 1/32)^2 + 0.296$. Quantumly, as $\lambda$ increases, the ground state potential concentrates near the global minimum $x^{\star}$ around $-0.7$. Classically, WKB effective potential eventually contains multiple global minima, and it cannot distinguish between $x^{\star}$ and the local minimum of $f(x)$ around $0.7$.}    
    \label{fig:degeneracy-separation-2d}
\end{figure}
It is clear from the Figure \ref{fig:degeneracy-separation-2d} that the  function $f^{(\text{WKB})}$ becomes more degenerate (contains multiple global minima) whereas the function $f$ contains only a single global minimum and starts concentrating around its unique global minimum. 

The eventual appearance of multiple global minimizers can also be seen by just looking at the form of $f^{(\text{WKB})}$ as explained in the introduction. The first term of the WKB potential is clearly zero at any first-order stationary point. That is, when $\nabla f(x) =0$, the WKB potential reduces to $-\beta \Delta f(x)$, which implies that all stationary points with the same ``curvature'' will be treated as equal minima of $f^{(\text{WKB})}$. This can cause the WKB potential to have additional global minima, which implies it will no longer be able to distinguish $x^{\star}$ from the local minima of $f$, leading to multiple global minima. This is also in-agreement with results from algebraic topology, where the Witten Laplacian is used to classify the critical points of a given index \cite{cycon1987schrodinger}. 

It is also known that the precense or absence of multiple global minima affects the spectrum of the Hamiltonian in a very drastic way.  A result from semiclassical analysis, due to Simon \cite{simon1983semiclassical}, states that for such $f$ with unique global minimum and corresponding Schr\"odinger operator, the $k$-th eigenvalue $E_k(\lambda)$, $k \geq 1$, of $H(\lambda)$ satisfies
\begin{align}
\label{eqn:semiclass_nondegen}
    \lim_{\lambda \rightarrow \infty} E_k(\lambda)/\lambda = e_k(\lambda),
\end{align}
where $e_k(\lambda)$ is the corresponding $k$-th eigenvalue of 
\begin{align*}
\tag{QHO Hamiltonian}
    H_a(\lambda) = -\Delta + \frac{\lambda^2}{2} \langle(x- x^{\star}), \nabla^2f(x^{\star})(x-x^{\star})\rangle. 
\end{align*}
In other words, up to $o(\lambda)$ corrections, $H(\lambda)$ has the same spectrum as a quantum harmonic oscillator (QHO) centered at the global minimizer. For the quantum spectral gap, $\delta^{(Q)}(\lambda)$, the above implies that for sufficiently large $\lambda$,
\begin{align*}
    \delta^{(Q)}(\lambda) \geq \sqrt{\sigma_{\min}(\nabla^2f(x^{\star}))}\lambda  - o(\lambda).
\end{align*}
Hence, if the global minimizer has a non-degenerate Hessian, then the Schr\"odinger operator has an $\Omega(\lambda)$ spectral gap for sufficiently large $\lambda$. 

On the other hand, the setting of multiple global minima is known as the \emph{tunneling regime}. This is because the ground state now has support in multiple wells, and thus there is a larger probability, than in the unique minimum case, of the state tunneling between global-minimum wells. Still, due to the  wells being degenerate (same minimum value, similar shape), the gap is mostly determined by the tunneling amplitude, which while larger than in the unique-minimum case, is still very small. Specifically, using additional techniques from semiclassical analysis, it is known that for large $\beta$, the gap is falling exponentially in $\beta$ \cite{simon1984semiclassicaltunneling}. This exponential scaling is also evident from the spectral gap of the diffusion process by Theorem~\ref{thm:lang_sgd_gap}.

The discrepancy just discussed can easily be observed numerically (See Figure \ref{fig:gap-comparison}) by plotting the spectral gap of $ -\frac{1}{\lambda}\Delta + \lambda f(x)$ which has a gap that is $1/\lambda$ factor smaller than \eqref{eqn:schrodinger_op_lambda}. Importantly, the quantum gap attains its minimum for small $\lambda$ and then asymptotes to a constant. On the other hand,  the exponential decay of $\delta^{(C)}(\lambda)$ in terms of $\beta$ persists even for one-dimensional functions.

\begin{figure}[H]
    \centering
    \includegraphics[width=0.8\linewidth]{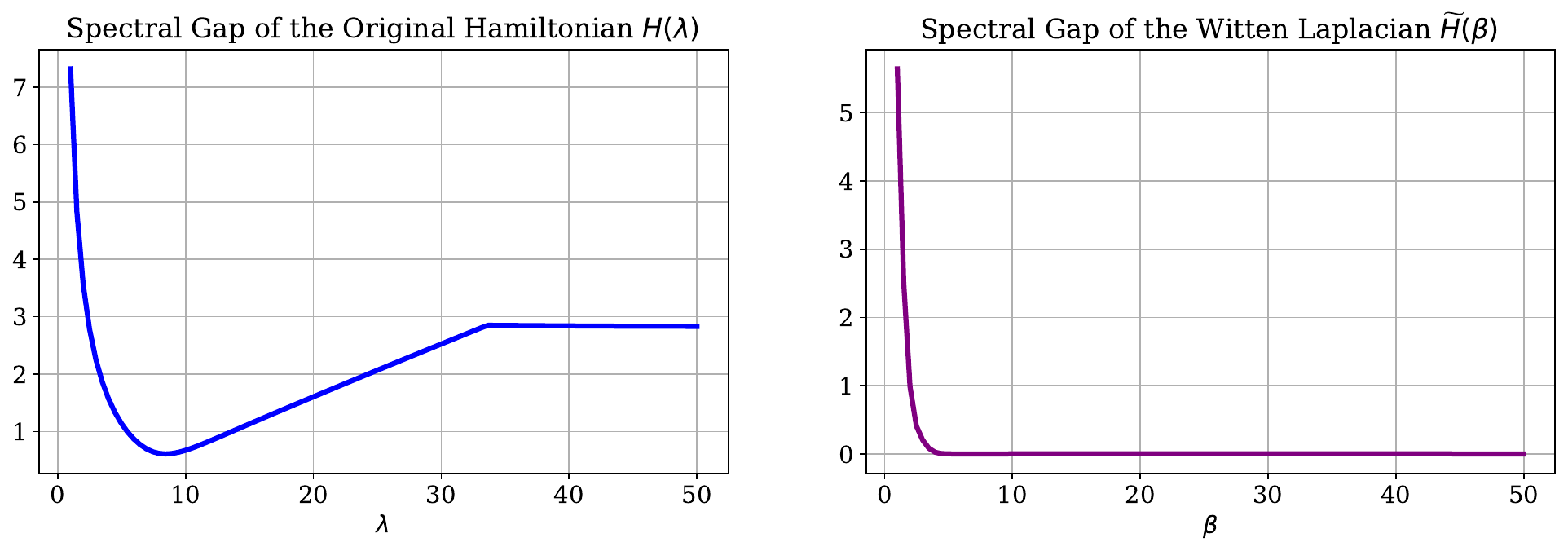}
    \caption{\textit{Comparison of the spectral gaps of the original Hamiltonian with potential $f$ and the Witten Laplacian with the effective WKB potential.}}
    \label{fig:gap-comparison}
\end{figure}

One issue with the existing result of Simon presented above, for the unique-minimum case, is that it treats $d$ to be a constant  and is only valid for an unspecified ``sufficiently large $\lambda$''. Thus this result is not directly useful for making  computational statements about optimization as it is not clear for which values of finite $\lambda$, this approximation is valid. In Section \ref{sec:local_spec_compare}, we strengthen this previous result by making the dependence on all problem parameters explicit and deriving what is the ``sufficiently large $\lambda$" for $d$-dimensional functions via proving the following theorem.
\begin{theorem}[Theorem \ref{thm:semiclassical_local_taylor_version} informal]
\label{thm:semiclassical_local_taylor_version_inf}
For $\lambda = \Omega(d^{5})$,  the spectral gap of the operator in Equation \eqref{eqn:schrodinger_op_lambda}
satisfies
\begin{align*}
\delta^{(Q)}(\lambda)
        \geq \sqrt{\sigma_{\min}(\nabla^2f(x^{\star}))}\lambda - \mathcal{O}(\lambda^{4/5}).
\end{align*}
\end{theorem}
The above implies that even for $d$ growing asymptotically, the gap remains lower bounded by the spectral gap of a QHO centered at $x^{\star}$, and also provides the precise corrections to the gap in terms of $\lambda$ and $d$. However, it seems to indicate that for the  semiclassical approximation to be valid,  $\lambda$ needs to grow as a large polynomial in $d$, although the degree of the polynomial can likely be improved.

We emphasize that this result does not imply that for any function with unique global minimum, the quantum gap does not decay with $d$. In general, there is a critical threshold in terms of $d$ beyond which this result is valid, even if it is a lower threshold than what Theorem \ref{thm:semiclassical_local_taylor_version} requires. Hence phenomenon such as first-order phase transitions can happen before the critical value and introduce exponentially decaying gaps early on. Nevertheless, we will show in the next section (Section \ref{sec:quat_clas_block_sep})  that for functions that obey a certain separability property and have a unique global minimum, no such quantum phase transition can occur. Interestingly, the result will provide a non-trivial gap lower bound for constant $\lambda$. Moreover, we prove in Section~\ref{sec:sgd_lower_bound} that the classical dynamics~\eqref{eqn:classical_lang_poten_dyn} take at least an exponentially long time to mix to its stationary distribution for functions with the  aforementioned separability property. 

Beyond separability, we also identify settings where the quantum ground state potential is smooth and non-degenerate whereas the associated WKB potential is degenerate and highly irregular making the classical dynamics mix very slowly. This irregularity can also be seen in low dimensions (Figures~\ref{fig:levy-3d}, \ref{fig:rastirigin-separable}, and \ref{fig:rastrigin-nonseparable}). By using recently developed hypercontracitivy results, we will rigorously show that the spectrum of the quantum ground-state potential remains robust to certain perturbations for all $\lambda$, whereas the same perturbation introduces multiple global minima in the WKB potential.

The above discussion leads to a new perspective on how quantum algorithms could outperform classical algorithms for global optimization beyond the tunneling explanations that usually depend on the specific shape of the objective function. Rather than considering specific tunneling amplitudes between the local minima, our mechanism uses the uniqueness of the global minimum to show that quantum particle is not in even in the tunneling regime. For completeness and comparison, in Section \ref{sec:tunneling_degeneratecase} we discuss a quantum/classical separation for continuous optimization based on tunneling. One will noticed that, unlike in the unique minimum case, the separation is highly sensitive to the shape of the barrier.

\section{Quantum/Classical Separation for Block-Separable Functions}
\label{sec:quat_clas_block_sep}

In this section, we will be investigating the potential for a quantum/classical runtime separation for optimizing rotated-versions of functions with the following separability property.
\begin{restatable}[Block/Completely Separable Function]{definition}{separableFunc} \label{def:separable-func}
A function $f: \mathcal{X} \rightarrow \mathbb{R}, \mathcal{X} \subseteq \mathbb{R}^d$, is a block separable function if there exists some integer $k \leq d$, and a partition of the coordinates of $f$: $\hat{x}_1, \dots, \hat{x}_k$ such that $f(x) = \sum_{i=1}^{k} g_i(\hat{x}_i)$, $g_i : \mathcal{X}_i \rightarrow \mathbb{R}$, $\mathcal{X}_i \subseteq \mathbb{R}^{d_i}$, and $\sum_{i=1}^k d_i = d$. Furthermore, if $k =d$, then we say that $f$ is completely separable.
\end{restatable}
As was just mentioned, we also consider rotated versions of block-separable functions. The rotation does not impact the runtime of RsAA or classical algorithms based on Langevin diffusion due to rotation invariance.
\begin{restatable}[Rotated Block-Separable Function]{definition}{rotSeparableFunc} \label{def:rotated-sep-func}
A function $f: \mathcal{X} \rightarrow \mathbb{R}, \mathcal{X} \subseteq \mathbb{R}^d$, is a rotated block-separable function if there exists $U \in \text{SO}(d)$ such that $f(Ux)$ is block separable.
\end{restatable}

We will apply the main mechanisms discussed in Section~\ref{sec:main-mechanism} and concretely show that RsAA can efficiently optimize rotated, block-separable functions that satisfy the following additional assumption.
\begin{assumption}[Constant Block Separability]
\label{assump:constant_block}
We say that a $d$-dimensional function $f$ satisfies the Constant Block Separability assumption, if it is block separable and the dimension of each block, $d_i$, is $\mathcal{O}_d(1)$. Additionally, there is no other dependence that any fixed $g_i$ has on $d$.
\end{assumption}
If each $d_i$ is a constant, and each $g_i$ does not have additional $d$ dependence, then given the rotation $U$ one can simply optimize a rotated block separable function by optimizing each $g_i$ separately. The addition of the rotation acts to ``hide'' this structure. 

Even under Assumption \ref{assump:constant_block}, block-separable functions can become highly nonconvex, and are even considered as benchmarks in many optimization suites \cite{Jamil_2013}. In fact, at least six of the test functions used for the benchmarking of quantum Hamiltonian descent (QHD) \cite{leng2023quantum, leng2023quantumhamiltoniandescent} against off-the-shelf classical algorithms were even \emph{completely separable}. For the majority of these separable functions, SGD and other classical algorithms appeared to struggle substantially (Section \ref{sec:classical-algorithms}). We note that the definition of separability that we use here is a generalization of the one considered in \cite{leng2023quantum}.

In the case of algorithms based on Langevin diffusions, we can prove an exponentially large, in $d$, lower bound for optimizing functions satisfying Assumption \ref{assump:constant_block}. This is a result of the mechanisms discussed in Section \ref{sec:main-mechanism} and is proven in Section \ref{sec:sgd_lower_bound}. In Section \ref{sec:quantum_runtime_separable}, we prove that RsAA can find an $\epsilon$-approximate minimizer of any function satisfying Assumption \ref{assump:constant_block} using $\mathcal{O}(d^6/\epsilon^4)$ queries to a noisy-binary quantum oracle. This leads to a potential exponential separation against all of the previously discussed off-the-shelf classical algorithms. However, there are structure-aware algorithms, discussed in Section \ref{sec:classical_algs_separable}, that can efficiently optimize functions satisfying Assumption \ref{assump:constant_block}. Still, we show the degree of the polynomial runtime can be made to be arbitrarily large, without impacting the polynomial in the quantum runtime. Hence, while the separation is not exponential against all classical algorithms that we consider, it is an arbitrarily-large polynomial separation.

In the next subsection, we will prove that quantum runtime is $\mathcal{O}(\text{poly}(d, 1/\epsilon))$ for functions satisfying Assumption \ref{assump:constant_block}, leading to the previously mentioned quantum/classical separation in $d$. As a corollary, this rigorously proves the polynomial runtime for RsAA stated in \cite{leng2023quantum} for completely-separable functions.

\subsection{Provable Quantum Runtime for Rotated Block-Separable Functions}
\label{sec:quantum_runtime_separable}

In Section \ref{subsec:separation_mechanism}, we presented an informal version (Theorem \ref{thm:semiclassical_local_taylor_version_inf}) of our spectral gap bound for functions with unique global minimum, proven in Section \ref{sec:semiclassical_analysis}. Generally, this result requires $\lambda$ to grow with $d$ to be applicable. However, for rotated block-separable functions the result can apply even for constant $\lambda$. This is because if $f(x)$ decomposes as 
\begin{align*}
    f(Ux) = \sum_{i=1}^{k}g_i(\hat{x}_i),
\end{align*}
for some $U \in \text{SO}(d)$,
then the corresponding Schr\"odinger operator
\begin{align*}
    H(\lambda) = -\Delta + \lambda^2\sum_{i=1}^{k}g_i(\hat{x}_i)
\end{align*}
tensorizes across dimension, i.e.
\begin{align*}
    H(\lambda) = \sum_{i=1}^{k} \left(-\Delta_i + \lambda^2 g_i(\hat{x}_i)\right),
\end{align*}
where $-\Delta_i$ is the Laplacian restricted to the coordinates $\hat{x}_i$. As a consequence, the spectral gap of $H$ is lower bounded by the minimum of the spectral gaps of 
\begin{align*}
    -\Delta_{i} +  \lambda^2 g_i(\hat{x}_i),
\end{align*}
which is $d_i$-dimensional. 

Under Assumption \ref{assump:constant_block}, the $d$ dependence in Theorem \ref{thm:semiclassical_local_taylor_version_inf} becomes constant. This (specifically the formal version, Theorem \ref{thm:semiclassical_local_taylor_version}) leads to the following corollary.
\begin{restatable}[Block-Separable Gap Bound for Constant $\lambda$]{corollary}{corgapblocksep}
\label{cor:separable_gap}
  Suppose $f: \mathcal{X} \rightarrow \mathbb{R}$ is a rotated block-separable function in $C^3(\Xcal)$, has a unique global minimum $x^{\star}$, at least one other local minimum, and Assumption \ref{assump:constant_block} holds. Let $g_i$ be such that $\sigma_{\min}(\nabla^2g_i(\hat{x}^{\star}_i)) > 0$ is minimal and $y^{\star} \in \mathcal{X}_i$ be the closest local minimum of $g_i$ to $\hat{x}_i^{\star}$. Let $\lambda_{\star} \in \mathbb{R}_+$ be any constant such that
  \begin{align*}
      &\frac{\sigma_{\min}(\nabla^2g_i(\hat{x}^{\star}_i))}{2} \lambda_{\star}^{1/5} - \frac{\lambda_{\star}^{-1/5}}{6}  - \frac{3}{\sqrt{2}} d_i\sqrt{\sigma_{\max}(\nabla^2g_i(\hat{x}^{\star}_i))} > 0\\
      &\lambda_{\star} > \frac{1}{\lVert \hat{x}^{\star}_i - y^{\star}\rVert},
  \end{align*}
  where $\gamma = \sup_{\lVert \hat{x}^{\star}_i - x\rVert \leq \lambda^{-2/5}} \lVert \nabla^3g_i(\hat{x}^{\star}_i)\rVert_{\text{op}}$.
  Then, for any $\lambda \geq \lambda_{\star}$ the spectral gap, $ \delta^{(Q)}(\lambda)$, of the operator 
  \begin{align*}
     H(\lambda) = -\Delta + \lambda^2 f(x)
  \end{align*}
  satisfies
  \begin{align*}
      \delta^{(Q)}(\lambda) \geq \sqrt{\sigma_{\max}(\nabla^2g_i(\hat{x}_{\star}))}\lambda - \frac{7}{3}\lambda^{4/5} - \frac{\gamma}{3}\lambda^{4/5}.
  \end{align*}
\end{restatable}
Under the hypotheses of $d_i = \mathcal{O}_d(1)$, there is clearly also a non-zero $\lambda_{\star}$ that satisfies the above inequalities. To show that the quantum adiabatic algorithm can efficiently solve the problem, we still need to argue that the gap is large for $\lambda < \lambda_{\star}$. This turns follows from some additional arguments, that when combined with Corollary \ref{cor:separable_gap}, lead to the following.
\begin{restatable}[Spectral Gap Bound for Block-Separable Functions]{theorem}{thmgapblocksep}
\label{thm:gap_bound_for_sep}
 Suppose $f: \mathcal{X} \rightarrow \mathbb{R}$ is a rotated block-separable function satisfying Assumption \ref{assump:constant_block}, and  $\mathcal{X}$ is a compact subset of $\mathbb{R}^d$. Furthermore, suppose the conditions of Corollary \ref{cor:separable_gap} are satisfied. Then, $\forall \lambda \geq 0$, the spectral gap $\delta^{(Q)}(\lambda)$ of the operator
 \begin{align*}
     H(\lambda) = -\Delta + \lambda^2 f(x)
 \end{align*}
 satisfies
 \begin{align*}
     \delta^{(Q)}(\lambda) \geq c\max\{\lambda \cdot \textup{sgn}(\lambda - \lambda_{\star}), 1\},
 \end{align*}
 for some constants $\lambda_{\star} > 0$ and $c > 0$.
\end{restatable}
This result is proven in Section \ref{sec:semiclassical_block_separable}. 

We can now proceed to bound the runtime. The simulation algorithm, as discussed in Section \ref{sec:algorithm_analysis}, tracks the the evolution of a continuous-quantum state with support on rescaled $\ell_{\infty}$ ball with radius $\frac{1}{2}$. Hence, for the digital quantum state $|\Psi\rangle$ meant to approximate a continuous state, we will use the notation $\mathbb{P}_{|\Psi\rangle}[ X \in A]$ to denote the probability of observing $X \in [-\frac{1}{2}, \frac{1}{2}]^{d}$ in a measurable set $A \subset \mathbb{R}^d$. Here, $x$ is also a grid point in $[-\frac{1}{2}, \frac{1}{2}]^{d}$ associated with a unique discrete computational basis state.
\begin{theorem}[Quantum Runtime for Block-separable Functions]
\label{thm:quant_runtime_for_block_sep}
    Let $f : \Xcal \rightarrow \mathbb{R}$ be a rotated block-separable function in ${C}^3(\Xcal)$ with a unique global minimizer $x^{\star}$, satisfying Assumption \ref{assump:constant_block},  $\mu_{\star} I \preceq \nabla^2f({x}^{\star}) \preceq L_{\star}I$, and $\gamma \geq  \sup_{\mathcal{B}_2({x}_{\star}, 1)} \lVert \nabla^3 f({x})\rVert_{\text{op}}$ for known $L_{\star}, \mu_{\star}, \gamma, \max_i d_i$. Suppose we are given a point $x_0$, such that $x^{\star} \in x_0 + [-2R, 2R]^{d} \in \Xcal$, $R = \mathcal{O}(1)$. Then there is a digital quantum algorithm that outputs a quantum state $|\Psi\rangle$ such that
    \begin{align*}
         \mathbb{P}_{|\Psi\rangle}[f(X) - f(x^{\star}) \leq \epsilon] \geq \frac{3}{5}.
    \end{align*}
    The algorithm starts from the discrete uniform superposition over $x_0 + [-2R, 2R]^d$, uses $\mathcal{O}\left(d^6/\epsilon^4\right)$ %
    queries to an $\epsilon_f = \widetilde{\Ocal}(\frac{\epsilon^4}{d^6})$ %
    accurate binary oracle, $\Ocal\left(d^2\cdot \polylog(d , 1/\epsilon)\right)$ qubits, and $\widetilde{\Ocal}\left(\poly(d, 1/\epsilon)\right)$ gates.
\end{theorem}
\begin{proof}
To determine the query complexity we apply our simulation theorem, Theorem \ref{thm:adiabatic_simulation}.
Note that due to the separability of $f$, we can determine $T_{\text{adiabatic}}$ for each component $i$, and then divide $\rho_{\text{adiabatic}}$ by $d$. Hence, for the adiabatic time, we need to upper bound the quantity $\theta$ from Theorem \ref{thm:adiabatic_simulation} for any component $g_i$:
\begin{align*}
\theta = \lambda_{\max}^2\delta^{-2}(1) + 12\int_0^1  \delta^{-3}(s) \lambda^4_{\max} \mathrm{d}s,
\end{align*}
where  $\Lambda = \mathcal{O}(1)$ only in $\theta$.
From Theorem \ref{cor:separable_gap}, we have
\begin{align*}
\int_0^1 \lambda^4_{\max}\delta^{-3}(s)\mathrm{d}s &\leq c\int_{0}^{(\lambda_{\star}/\lambda_{\max})^2} \lambda^4_{\max} \mathrm{d}s + \int_{(\lambda_{\star}/\lambda_{\max})^2}^{1} c\lambda^4_{\max} \left(\frac{1}{\lambda_{\max}^3 s^{3/2}}\right)\mathrm{d}s\\
& = \mathcal{O}(\lambda_{\max}^2),
\end{align*}
for some constant $c > 0$ and $\lambda_{\star} \geq 1$. Hence since $\delta^{-2}(1)  = \mathcal{O}(1/\lambda_{\max}^2)$, $\theta = \mathcal{O}\left(\lambda_{\max}^2\right)$. Hence $T_{\text{adiabatic}} = \mathcal{O}\left(d\lambda_{\max}^2\right)$. There is an additional $\Lambda$ factor in the query complexity, which cannot be taken to be $\mathcal{O}(1)$ this time, but is $\mathcal{O}(d)$ by block-separability and $d_i = \mathcal{O}(1)$.
Hence the query complexity of RsAA is $\mathcal{O}\left(\frac{\lambda_{\max}^4d^2}{\rho_{\text{adiabatic}}}\right)$. Lastly, we need to determine a value of $\lambda_{\max}$ that is sufficient for the ground state to output an $\epsilon$-approximate minimizer with constant probability.

Let $x^{\star} = f(x^{\star}) = 0$, and suppose we want to lower bound the probability of observing an $x$ such that $f(x) \leq \epsilon$.

Consider the ball $\mathcal{B}_2(0, \sqrt{\epsilon/(\gamma + L_{\star})})$. Then
\begin{align}
\label{eqn:quadratic_upper_bound}
f(x) 
&\leq \langle x, \nabla^2f(0), x \rangle + \sup_{y \in \mathcal{B}_2(0, \sqrt{\epsilon/(L_{\star} + \gamma)})}\lVert \nabla^3f(y)\rVert\lVert x \rVert^3 \nonumber\\
&\leq \langle x, (\nabla^2 f(0) + \gamma I) x\rangle.
\end{align}

Note that Lemma \ref{lem:hypercontractive-sufficient-lambda} with $r = \sqrt{\epsilon/(\gamma + L_{\star})}$ and $L = \gamma + L_{\star}$, gives
 \begin{align*}
    \lambda_{\max} = \mathcal{O}\left(\frac{d\sqrt{L_{\star} + \gamma}}{\epsilon}\right) = \mathcal{O}(d/\epsilon)
\end{align*}
suffices to make 
\begin{align*}
    \mathbb{P}_{\lvert \Phi_{\lambda_{\max}}\rvert^2}[ f(X) \geq \epsilon] < \frac{1}{5}.
\end{align*}

Theorem \ref{thm:adiabatic_simulation} provides that the outputted digital state $|\Psi\rangle$ satisfies:
    \begin{align*}
         \mathbb{P}_{|\Psi\rangle}[ X \in \mathcal{B}_2(x^{\star}, \tilde{\epsilon})] > \mathbb{P}_{\lvert \Phi_{\lambda_{\max}}\rvert^2}[ X \in \mathcal{B}_2(x^{\star}, \tilde{\epsilon}/2)] - \rho_{\text{sim}} - \rho_{\text{adiabatic}}.
    \end{align*}
Then we can take $\rho_{\text{adiabatic}} + \rho_{\text{sim}} = \frac{1}{5}$. Also from \eqref{eqn:quadratic_upper_bound}, we have $\tilde{\epsilon} = \frac{\sqrt{\epsilon}}{\sqrt{L_{\star} + \gamma}}$ gives that above implies $f(x) \leq \epsilon$. Thus we can combine with our earlier discussion to get that the query complexity is $\mathcal{O}\left(d^6/\epsilon^4\right)$.

\end{proof}

One will note that query complexity obtained for optimizing the completely separable function \cite{leng2023quantum} was $\widetilde{\mathcal{O}}(d^3/\epsilon^2)$, which is quadratically better than ours. This quadratic reduction comes from a more clever choice of annealing schedule, which, for simplicity, we choose not to optimize here.

As a visual aid, we present some numerical results (Figures \ref{fig:levy-3d} and \ref{fig:rastirigin-separable}) showcasing the significant difference between the quantum ground state potential and WKB potential for two completely-separable functions. These functions are generally considered to be hard for off-the-shelf classical solvers. In both cases, one can observe the regularity of  the quantum ground-state potential and the appearance of many global minimizers in the WKB potential (as discussed in Section \ref{subsec:separation_mechanism}). Still, such functions can be optimized by a simple coordinate descent. In Figure \ref{fig:rastrigin-nonseparable}, we introduce a rotation
\begin{align} \label{eq:rastirigin-rotation}
\begin{pmatrix}
z_1 \\
z_2
\end{pmatrix}
=
\begin{pmatrix}
\cos\theta & -\sin\theta \\
\sin\theta & \cos\theta
\end{pmatrix}
\begin{pmatrix}
x \\
y
\end{pmatrix},
\end{align}
which implies that standard coordinate descent no longer works without knowing the rotation matrix. Due to the rotational invariance of both RsAA and Langevin dynamics, the landscape characteristics remain unchanged. Still, we previously mentioned that additional structure-aware algorithms exist that can optimize the function in Figure \ref{fig:rastrigin-nonseparable} efficiently, albeit with arbitrarily-high polynomial runtime.

\begin{figure}
    \centering
    \includegraphics[width=1.0\linewidth]{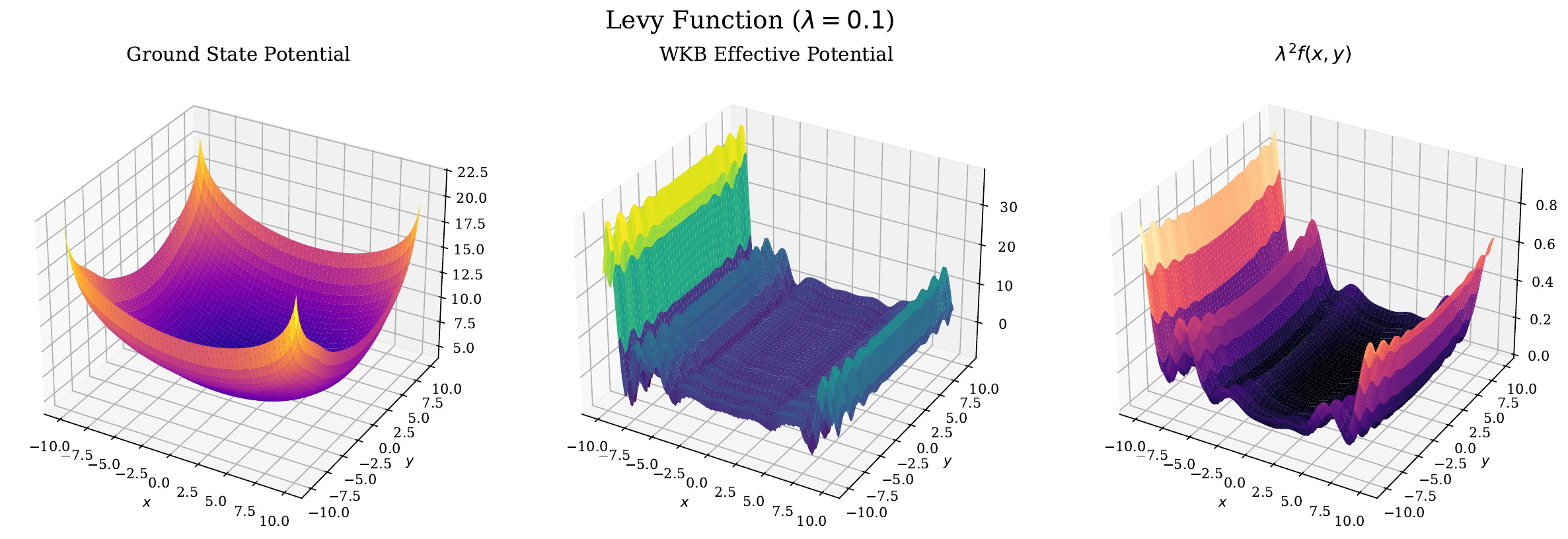}
    \includegraphics[width=1.0\linewidth]{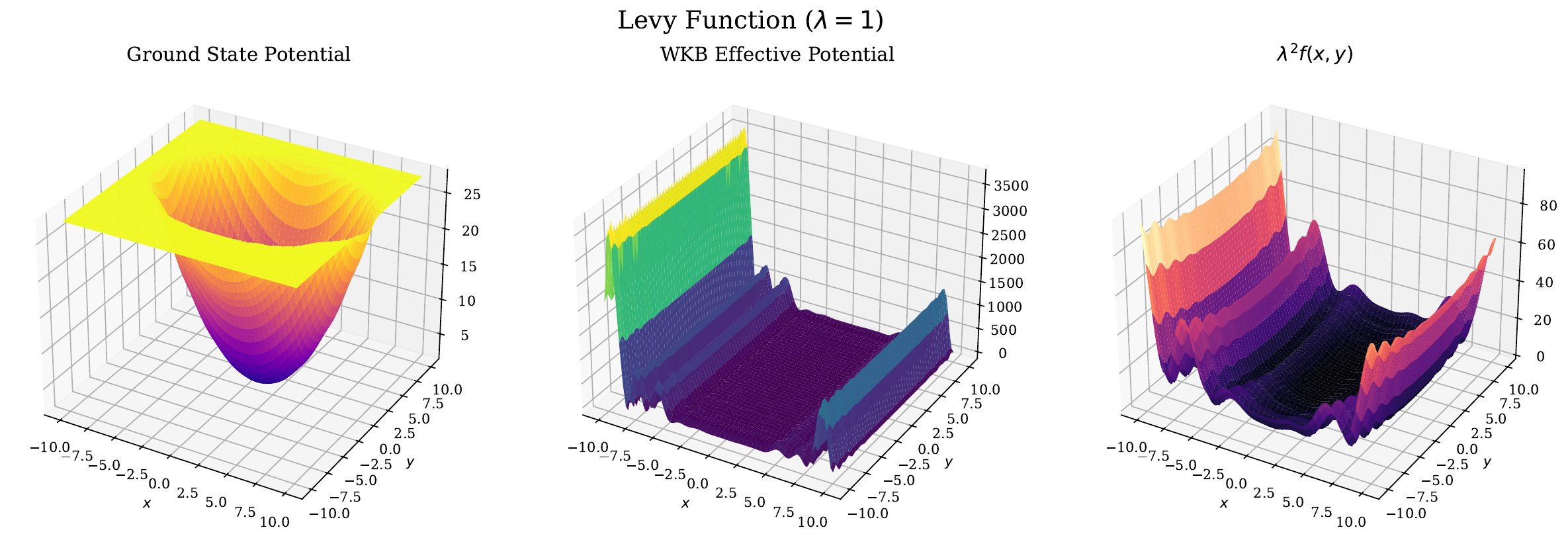}
    \includegraphics[width=1.0\linewidth]{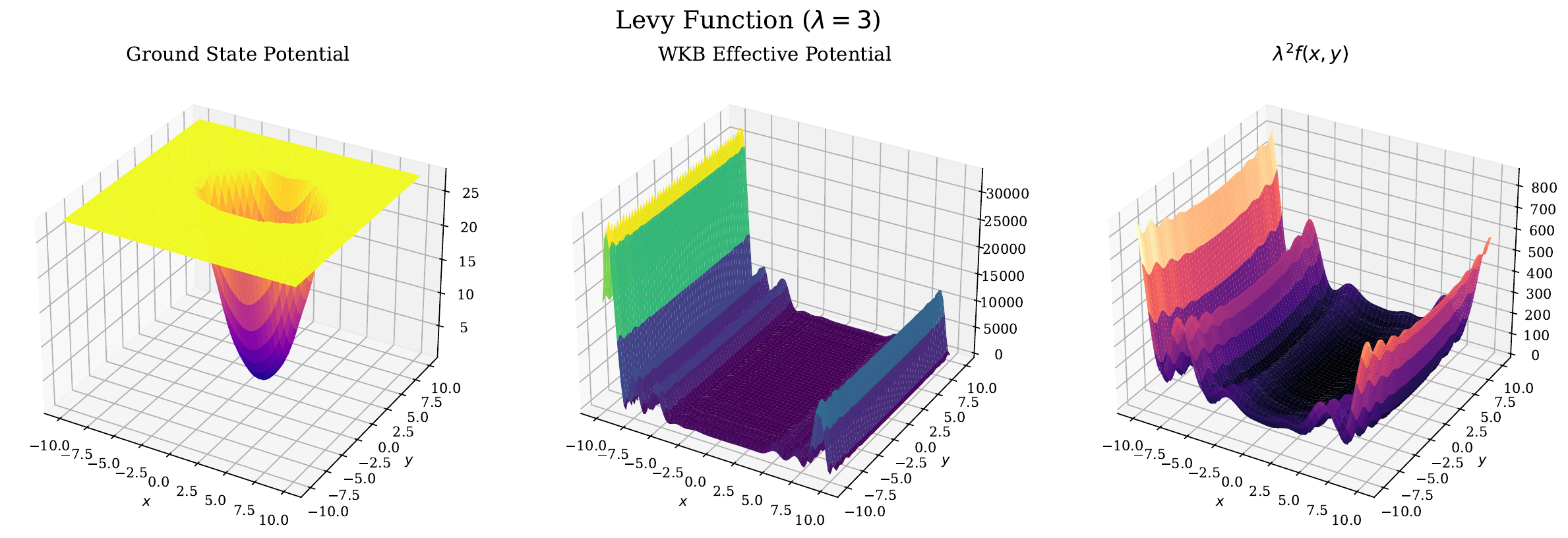}
    \caption{\textit{Visualization of the unique/multiple global minima separation for the Levy function.} Similarly to Figure~\ref{fig:degeneracy-separation-2d}, the ground state potential concentrates near the global minimum, whereas WKB effective potential has multiple global minimizers. The exact form is: $f(x, y) = \sin^2(\pi w_1) + (w_1 - 1)^2 \left[ 1 + 10 \sin^2(\pi w_1 + 1) \right]
+ (w_2 - 1)^2 \left[ 1 + \sin^2(2\pi w_2) \right]$ where $w_1 = 1 + \frac{x - 1}{4}$ and $w_2 = 1 + \frac{y - 1}{4}$.}
    \label{fig:levy-3d}
\end{figure}

\begin{figure}
    \centering
    \includegraphics[width=1.0\linewidth]{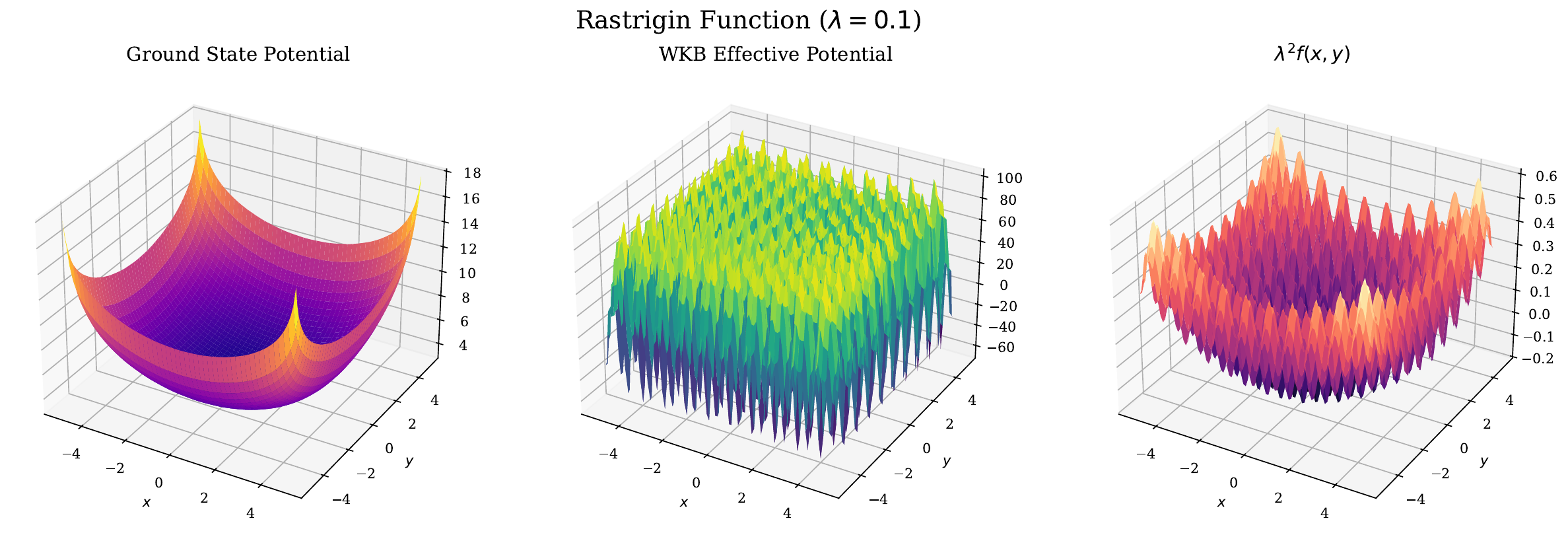}
    \includegraphics[width=1.0\linewidth]{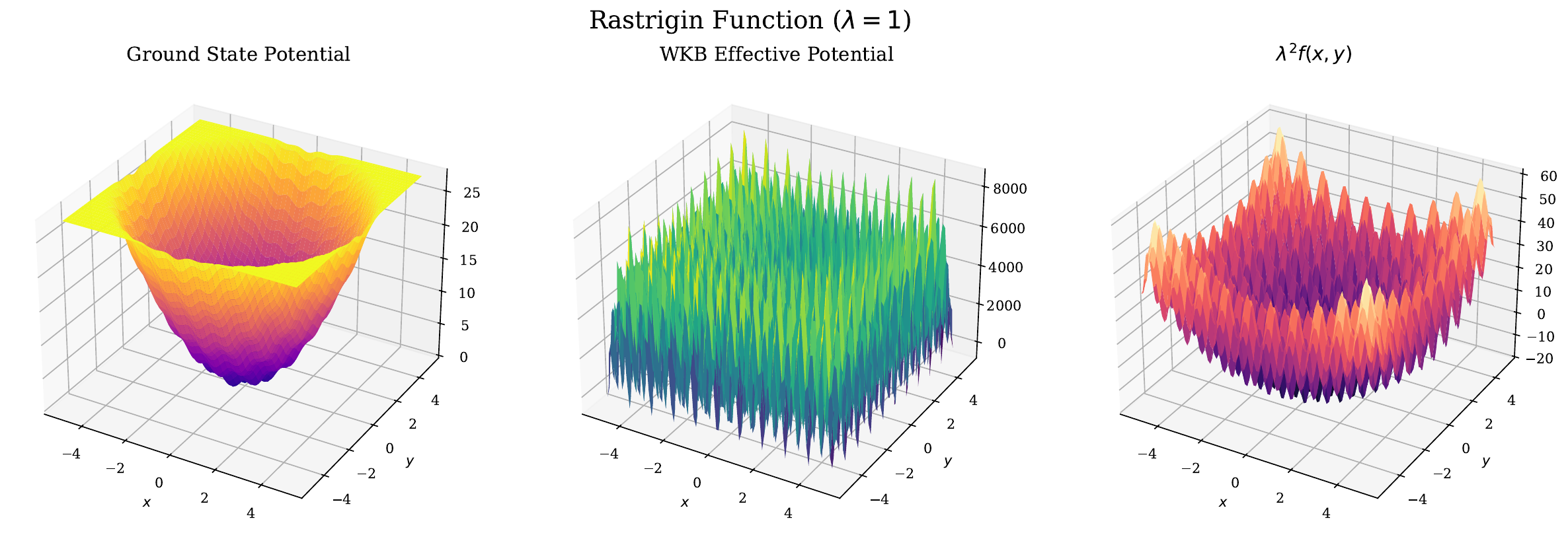}
    \includegraphics[width=1.0\linewidth]{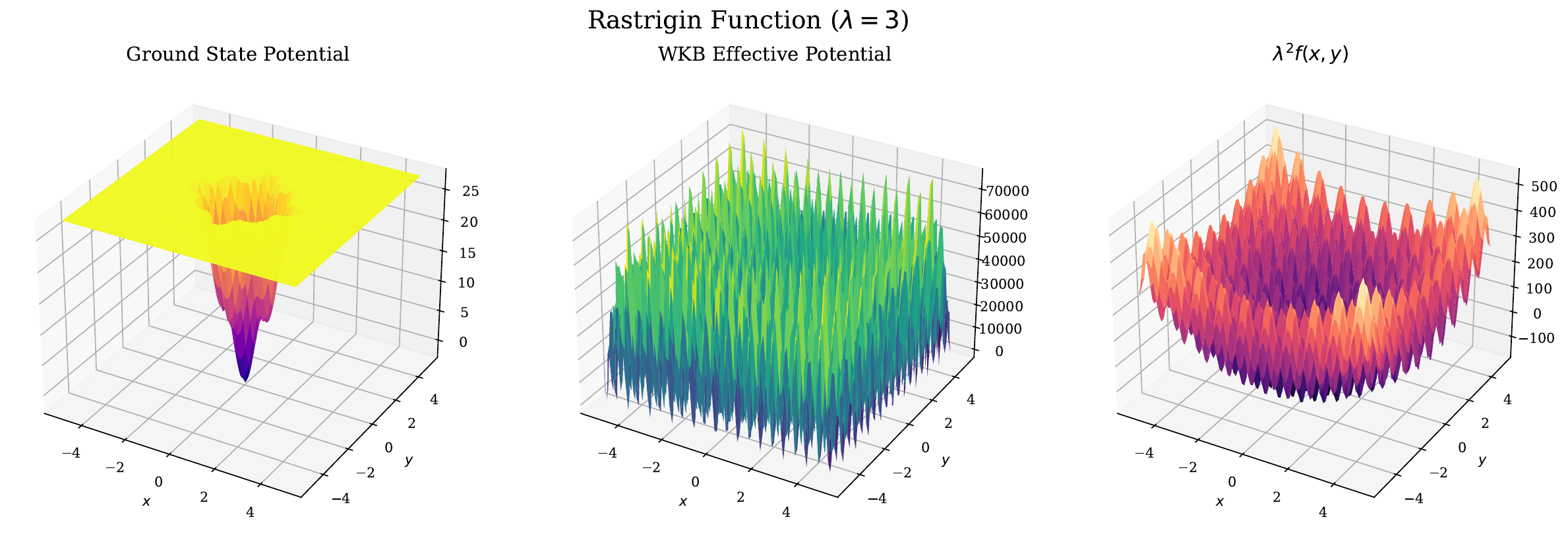}
    \caption{\textit{Visualization of the unique/multiple global minima separation for the separable version of Rastrigin function:} $f_{\text{sep}}(x, y) = \left[ x^2 - 10 \cos(2\pi x) \right] + \left[ y^2 - 10 \cos(2\pi y) \right] + 20$.}
    \label{fig:rastirigin-separable}
\end{figure}

\begin{figure}
    \centering
    \includegraphics[width=1.0\linewidth]{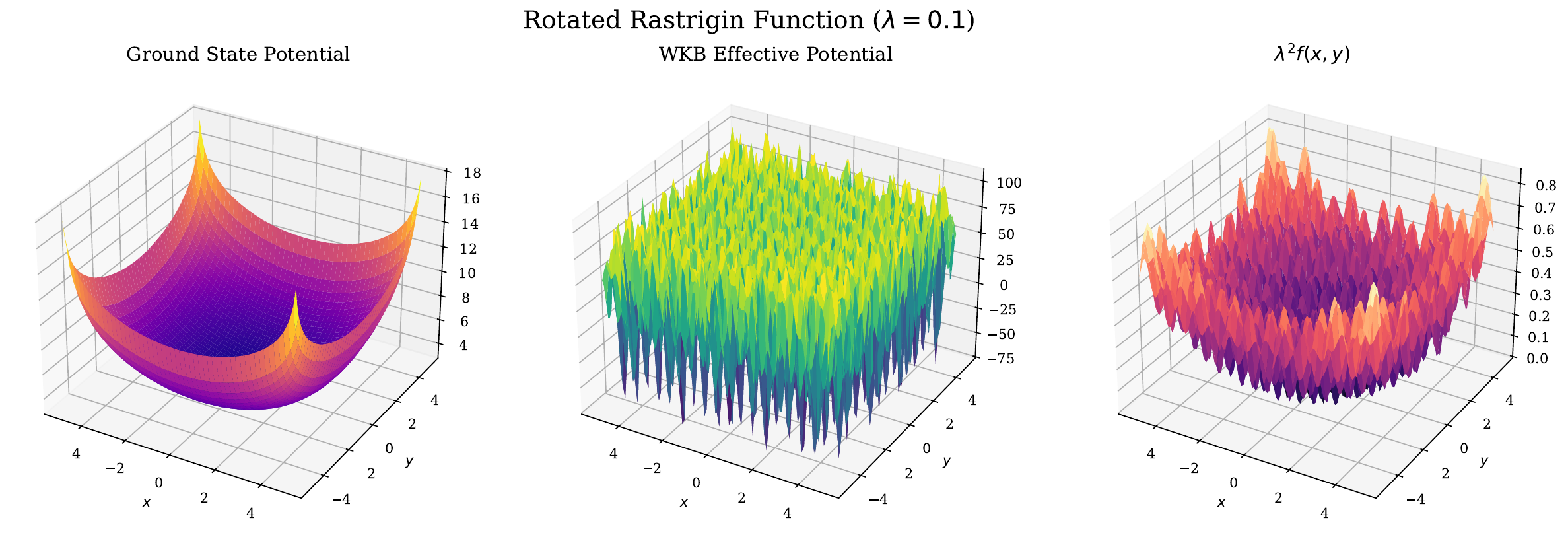}
    \includegraphics[width=1.0\linewidth]{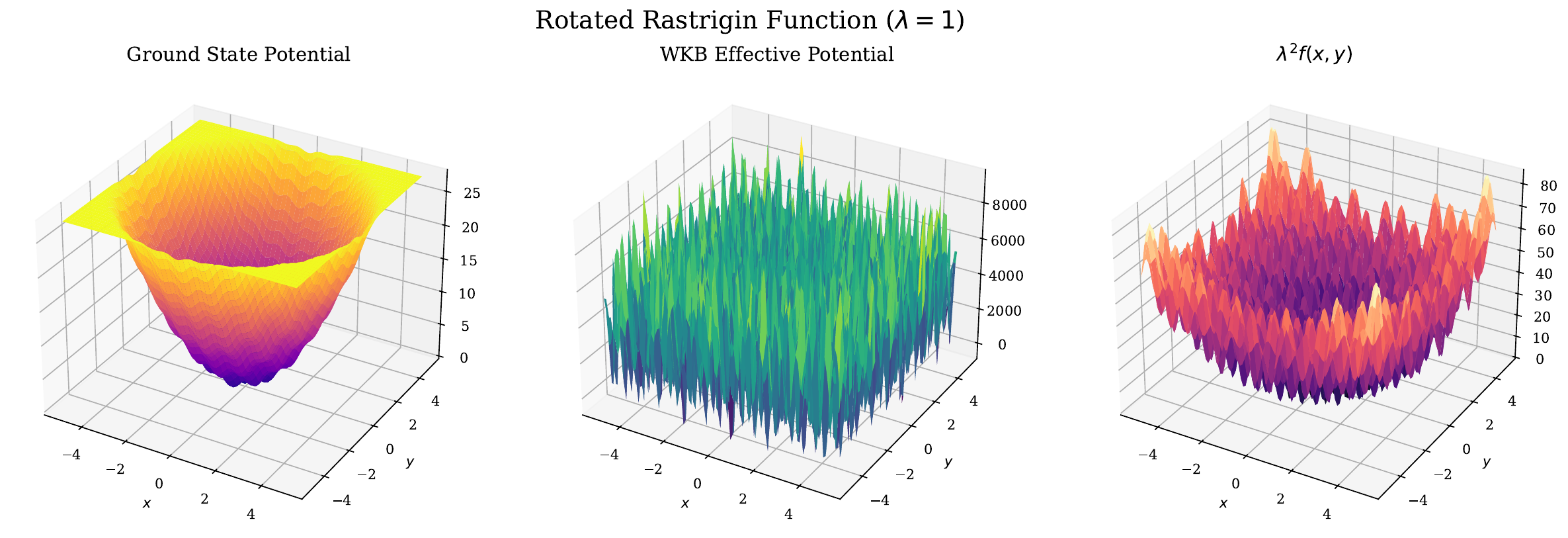}
    \includegraphics[width=1.0\linewidth]{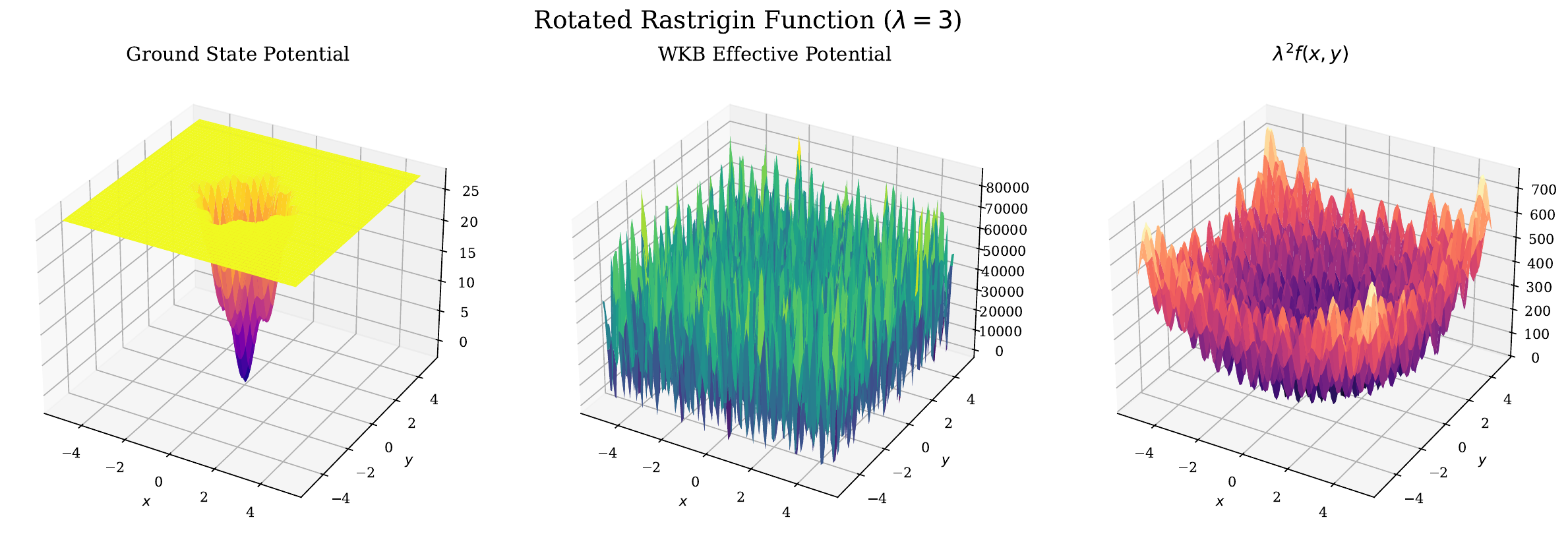}
    \caption{\textit{Visualization of the unique/non-unique minima separation for the rotated version of Rastrigin function:} $f_{\text{rot}}(x, y) = z_1^2 - 10 \cos(2\pi z_1) + z_2^2 - 10 \cos(2\pi z_2) + 20$ where $[z_1 ~ z_2]$ are rotated based on \eqref{eq:rastirigin-rotation} with $\theta = \pi/6$.}
    \label{fig:rastrigin-nonseparable}
\end{figure}

\subsection{Structure-aware Classical Algorithms}

\label{sec:classical_algs_separable}

We now discuss the previously mentioned structure-aware algorithms, which do optimize rotated block-separable functions satisfying Assumption \ref{assump:constant_block} in $\mathcal{O}(\text{poly}(d, 1/\epsilon))$ time. However, we emphasize that both structure-aware algorithms suffer from an \emph{arbitrarily-large} polynomial runtime in terms of $d$. Specifically, the degree of the polynomial can be made large by making simple modifications to the function. In contrast, the quantum algorithm, RsAA, \emph{does not} suffer from this, i.e. the polynomial dependence of RsAA does not change when modifying these same function characteristics.
\paragraph{Convexity-Honing Algorithm} The first classical algorithm leverages the fact that any separable function is convex in an $\ell_\infty$ ball around the global minimum. Thus $\beta$ (in Equation \eqref{eqn:classical_lang_poten_dyn}) only needs to be
    high enough to ensure that a point in this convex region can be found with high probability, after which regular gradient descent can be used to complete the optimization. Specifically, if $f$ is block separable, then we show that $\beta = \Theta(\log(d))$ suffices for $\mu_{\beta}(x) \propto e^{-\beta f(x)}$ to output a point in a convex region around the global minimum $x^{\star}$ with constant probability. From Theorem \ref{thm:lang_sgd_gap}, the relaxation time of Langevin dynamics is at most polynomial in $d$ in this case. After running Langevin dynamics, we then run a simple gradient descent. Once in the convex region, the gradient descent will succeed in finding an $\epsilon$-approximate global minimizer with probability one. Hence, the whole procedure succeeds with constant probability. As mentioned earlier, due to rotation invariance of Langevin dynamics and gradient descent, the above applies to rotated block-separable functions as well.
    
    This algorithm is carefully analyzed in Section~\ref{sec:convex_honing}, where it is shown that the runtime is polynomial but the exponent can be made arbitrarily high by choosing an appropriate function. The following theorem upper bounds the runtime.
    \begin{restatable}[Convexity-Honing Algorithm Runtime]{theorem}{convexHoning}
    \label{thm:convex_hone_thm}
        Let $f : \Xcal \rightarrow \mathbb{R}$ be a rotated block-separable function satisfying Assumption \ref{assump:constant_block}, and $\Xcal$ is bounded. Then there is a classical algorithm that can find an $\epsilon$-approximate minimizer of $f$, with constant probability using $\widetilde{\mathcal{O}}(d^{c})$, where $c = \Theta(H_f)$, queries to a first-order oracle for $f$.
    \end{restatable}

\paragraph{Hessian Algorithm}  The other algorithm is a classical algorithm that takes advantage of separability directly. One can easily show that, at all points, the Hessian of a rotated block-separable $f$ is conjugate to the Hessian of a function $g(x) = \sum_{i=1}^{k}g(\hat{x}_i)$ for some $k$. The Hessian of the latter is block-diagonal with respect to the subsets indexed by $i=1, \dots, k$. Given that this rotation is identically applied to the Hessian of $g$ at all points to get the Hessians of $f$, we can recover the block-diagonal structure by computing the Hessian of $f$ at two distinct points. Once the rotation is recovered, we can run a grid-search on each $g_i$ to recover the global optimum, given Assumption \ref{assump:constant_block} this runs in time independent of $d$.

    We analysis this Hessian-based algorithm in Section~\ref{sec:hessian_algorithm} for block separable functions. The runtime of this algorithm can be bounded by a polynomial with a function independent exponent for $d$, but a function dependent exponent for $\epsilon$. This function-dependent exponent, like for the previous algorithm, can be made to be arbitrarly high:
      \begin{restatable}[Hessian Algorithm Runtime]{theorem}{hessianAlgo}
      \label{thm:hessian_algo}
        Let $f : \Xcal \rightarrow \mathbb{R}$ satisfy Assumption \ref{assump:constant_block}, and $\Xcal$ a bounded subset of $\mathbb{R}^d$. Then there is a classical algorithm that can find an $\epsilon$-approximate minimizer of $f$ with $\mathcal{O}\left((d/\epsilon)^{c}\right)$ queries to a first-order oracle for $f$, where $c$ is a constant depending only on $f$.
    \end{restatable}

\section{Quantum/Classical Separation beyond Coordinate-Separability}
\label{sec:beyond_coordinate_sep_funcs}
We now investigate a potential quantum/classical separation for optimizing functions that do not necessarily satisfy Definition~\ref{def:separable-func}, explored in the previous section. To do this, we use another aspect of the separation mechanism from Section \ref{subsec:separation_mechanism}: the quantum ground state is more robust to non-trivial perturbations that preserve the uniqueness of the global minimum. In contrast, such a perturbation can cause a classical diffusion process to get stuck because of the appearance of multiple local minima in potential, and hence multiple global minima in the WKB potential.

We first construct an example function which is strongly-convex inside a very large region but highly nonconvex outside this domain. Then, using perturbation analysis, we prove that the spectral gap is not affected by this perturbation as the ground state has very small support outside. Motivated by this phenomenon, next we construct a class of functions which can be nonconvex on the entire domain. Under certain growth assumptions, we show that the ground state is hypercontractive which is a spectral property of the ground state that implies both large spectral gap and other regularities such as sub-Gaussian tails. In this section, we present our results and postpone the detailed discussion on hypercontractivity to Section~\ref{sec:hypercontractivity}. Lastly, we discuss the run-time of classical algorithms for these example functions.   

\subsection{Removing Separability}
\label{sec:removing_separability}

In Section \ref{sec:quat_clas_block_sep} on block-separable functions, the separability property was used to show that the Schr\"odinger operator tensorizes. This implies the gap only depends on lower-dimensional Schr\"odinger operators corresponding to each block of the separable function. This enabled us to remove the dimension dependence in $\lambda_{\star}$ in Theorem \ref{thm:semiclassical_local_taylor_version_inf}. Here, we present an alternative gap bound that gets around this issue and applies to quadratically-enveloped, non-separable functions.

\begin{restatable}[Spectral Gap Perturbation Bound for Quadratically-enveloped Functions]{theorem}{pertQuatEnv}
\label{thm:pert_quad_envelop}
Consider $f : \mathcal{X} \rightarrow \mathbb{R}^d$. Suppose $f$ is three-times continuously differentiable and $c_f$-strongly convex for all $x$ inside an $\ell_2$ ball of radius $r = \Omega(\sqrt{d\max\left(\ln(d/\lambda_{\star}), \ln(d)\right)})$ around the unique global minimizer $x^{\star}$.
Additionally suppose that $f$ is quadratically-enveloped everywhere:
\begin{align*}
   \frac{c_f}{2}\lVert x - x^{\star} \rVert^2 \leq f(x) - f(x^{\star}) \leq \frac{C_f}{2}\lVert x - x^{\star} \rVert^2, \forall x \in \Xcal.
\end{align*}
Then for all $\lambda \geq \lambda_{\star}$:
\begin{align*}
\delta^{(Q)}(\lambda) \geq \lambda\sqrt{c_f} - o_d(1).
\end{align*}
\end{restatable}
The above effectively shows that if $g(x)$ is a quadratically-enveloped function, and $f(x) = g(x) + h(x)$ for perturbation $h(x)$ with  support far enough from the global minimizer of $g$, then the gap of Schr\"odinger operator for $f$ remains lower bounded by that of $g$. The quadratically-enveloped assumption is used determine a bound on the localization of the first-excited and ground states of $-\Delta + \lambda^2f$ around $x^{\star}$. From this, we determine the distance, $r$, from $x^{\star}$ where the perturbation is allowed. This result is discussed more in Section \ref{sec:semiclassical_quad_env} and proven in the appendix (Section \ref{sec:appendix_semi_classical}).

We can apply Theorem \ref{thm:pert_quad_envelop} to upper bound the quantum runtime for a particular non-separable function. Consider some $y  \in \mathbb{R}^d$ and the following function $f$ over the domain 
\begin{align*}
    \mathcal{X} = y + [-2c\sqrt{d\ln(d)}, 2c\sqrt{d\ln(d)}]^d,
\end{align*}
and given in the following piece-wise form,
\begin{align} \label{eq:func-perturb-outside}
f(x) = 
\begin{cases}
    \frac{L}{2}\lVert x - y\rVert^2 , &\text{if}~~ x \in \mathcal{B}_2(0, c\sqrt{d\ln(d)}), \\
    \frac{L}{2}\lVert x -y \rVert^2 + h(\vec{\theta}(x-y))\sin\left(d\frac{\lVert x -y \rVert^2-c^2d\ln(d)}{2}\right) + 1 &\text{otherwise},
\end{cases}
\end{align}
where $L = \mathcal{O}_d(1)$, $\vec{\theta}(x)$ outputs the direction of $\vec{x}$, $h(\vec{\theta}_{\star}) = 1$, for some particular direction $\vec{\theta}_{\star}$, and 
 $\lVert h \rVert_{\infty} = \mathcal{O}_d(1)$. Note that this function is continuous and differentiable almost everywhere. The following theorem bounds the quantum runtime.

\begin{theorem}[Quantum Runtime for Optimizing \eqref{eq:func-perturb-outside}]
\label{thm:runtime_perturb_outside}
    Let $f$ be as defined in \eqref{eq:func-perturb-outside}, where $L$ is known. Then there is a digital quantum algorithm that outputs a quantum state $|\Psi\rangle$ such that
    \begin{align*}
         \mathbb{P}_{|\Psi\rangle}[f(X) - f(x^{\star}) \leq \epsilon] \geq \frac{3}{5}.
    \end{align*}
    The algorithm starts from the discrete uniform superposition over $x_0 + \big[-2c\sqrt{d\ln(d)}, 2c\sqrt{d\ln(d)}\big]^d$, uses $\widetilde{\mathcal{O}}\left(d^{13}/\epsilon^{4}\right)$ queries to an $\epsilon_f = \widetilde{\Ocal}(\frac{\epsilon^{4}}{d^{13}})$ accurate binary oracle, $\Ocal\left(d^2\cdot \polylog(d , 1/\epsilon)\right)$ qubits, and $\widetilde{\Ocal}\left(\poly(d, 1/\epsilon)\right)$ gates.
\end{theorem}
\begin{proof}
The function satisfies the hypotheses of Theorem \ref{thm:pert_quad_envelop} for $\lambda_{\star} = 1$ and $g(x) = \frac{L}{2}\lVert x \rVert^2$. Hence the operator
\begin{align*}
    H(\lambda) = -\Delta + \lambda^2 f(x)
\end{align*}
has a spectral gap that is $\Omega(\lambda)$ for $\lambda \geq 1$. For $\lambda < 1$, we can apply Theorem \ref{thm:intrinsic-hypercontractivity} for the operators
\begin{align*}
&H_0(\lambda) = -\Delta + \lambda^2\frac{L}{2}\lVert x \rVert^2\\
&H_1(\lambda) = -\Delta + \lambda^2\frac{L}{2}\lVert x \rVert^2 + \lambda^2\Big(f(x) - \frac{L}{2}\lVert x \rVert^2\Big),
\end{align*}
where $\lambda^2(f(x) - \frac{L}{2}\lVert x \rVert^2)$ is the perturbation. Since the perturbation for $\lambda < 1$ is bounded by $\mathcal{O}(1)$, after applying Theorem \ref{thm:yau_strongly_convex} for $H_0$,  the gap can be lower bounded by $\Omega(\lambda + \textup{diam}(\Xcal)^{-2})$. Thus the gap can be lower bounded, $\forall \lambda \geq 0$, by
\begin{align*}
   c'\max(\lambda \cdot \textup{sgn}(\lambda - 1), \textup{diam}(\Xcal)^{-2})) \geq  c\max(\lambda \cdot \textup{sgn}(\lambda - 1), (d^{2}\ln(d))^{-1}).
 \end{align*}
For the adiabatic time, we need to upper bound the quantity
\begin{align*}
\theta = \delta^{-2}(1)\lambda_{\max}^2\Lambda + 12\int_0^1  \delta^{-3}(s) \lambda^4_{\max} \Lambda^2 \mathrm{d}s.
\end{align*}
From the above gap bound
\begin{align*}
\int_0^1 \lambda^4_{\max}\Lambda^2 \delta^{-3}(s)\mathrm{d}s &\leq  \tilde{c}\int_{0}^{(1/\lambda_{\max})^2} \lambda^4_{\max}\Lambda^2 \textup{diam}(\Xcal)^{6} \mathrm{d}s + \tilde{c}\int_{(1/\lambda_{\max})^2}^{1}\lambda^4_{\max} \Lambda^2 \left(\frac{1}{\lambda_{\max}^3 s^{3/2}}\right)\mathrm{d}s\\
& = \mathcal{O}(\lambda_{\max}^2 \Lambda^2 d^{6}\ln^3(d)).
\end{align*}
Hence $T_{\text{adiabatic}}  = \widetilde{\mathcal{O}}\left(\lambda_{\max}^2 \Lambda^2 d^{6}\right)$.

The analysis for determining $\lambda_{\max}$ to be $\mathcal{O}\left(d/\epsilon\right)$ proceeds exactly the same as in the proof of Theorem \ref{sec:quantum_runtime_separable}, with $L_{\star}, \mu_{\star}, \gamma = \mathcal{O}(1)$ from the definition of $f(x)$ in \eqref{eq:func-perturb-outside}.

\end{proof}

\subsection{Removing Large Convex Regions around Global Minimum}
In this section, we construct an example where the objective function can be written as a sum of strongly convex function and a nonconvex perturbation. Different from the previous section, the perturbation is supported everywhere on the domain so that one cannot simply use a structure-aware algorithm that tries to find the convex well inside the domain. Consider the following function $f:\mathcal{X}\to \mathbb{R}^d$ with a unique global minimum at $x^{\star}$. 
\begin{equation}
\label{eq:perturbation-strongly-convex}
    f(x) = h(x)+g(x)
\end{equation}
such that $h$ is $\mu^2$-strongly convex on $\mathcal{X}$, and exhibits a rapid growth, i.e. $f(x) \geq C_f \|x - x^{\star}\|^{2k}$ for $k \geq 1$
and the perturbation $g$ has slow growth, i.e., $g(x) \leq \frac{C_g(k+1)}{d}\|x - x^{\star}\|^{k+1}$.

We note that one can choose $\mathcal{X}$ to be a finite domain with possibly very large radius growing with dimension. Hence, it is hard find a region where the perturbation $g$ is small by an exhaustive search.

The following theorem, proven in Section \ref{sec:hypercontractivity}, provides a gap lower-bound for any perturbed function satisfying the aforementioned conditions.
\begin{restatable}{theorem}{perturbHyperContract}
\label{thm:instrin_hyper_perturb}Let $H = -\Delta + \lambda^2 h(x) + \lambda^2 g(x)$ be a Schrödinger operator on a bounded domain $\mathcal{X}$ and $x^{\star}$ is the unique minimizer of $f=h+g$. Assume the following conditions are met:
\begin{enumerate} 
\item The function $h$ is $\mu^2$-strongly convex on $\mathcal{X}$, and exhibits rapid growth, specifically it holds that $f(x) \geq C_f \|x - x^{\star}\|^{2k}$ for $k \geq 1$.
\item The function $g$ has slow growth, i.e, $g(x) \leq \frac{C_g(k+1)}{d}\|x - x^{\star}\|^{k+1}$. 
\end{enumerate} Under these conditions, the operator $H(\lambda)$ has a spectral gap $\delta^{(Q)}(\lambda)$ satisfying
\begin{align*}
    \delta^{(Q)}(\lambda) = \Omega\left(\lambda + \textup{diam}(\Xcal)^{-2}\right),
\end{align*}
for all $\lambda \geq 0$.
\end{restatable}

Using the above, we obtain the following theorem upper bounding the run-time of RsAA.
\begin{theorem}
There is a digital quantum algorithm that start from the discrete uniform superposition over $x_0 + [-2R, 2R]^d$, uses $\widetilde{\mathcal{O}}\left( d^{7}R^6 \Lambda^3/\epsilon^{4}\right)$ queries to an $\epsilon_f = \widetilde{\mathcal{O}}\left(\frac{\epsilon^{4}}{d^{7}R^6 \Lambda^3}\right)$ noisy oracle for $f$ in \eqref{eq:perturbation-strongly-convex} with $\Lambda \geq \lVert f \rVert_{\infty}$, $\Ocal\left(d^2\cdot \polylog(d, R, \frac{1}{\epsilon})\right)$ qubits and $\widetilde{\Ocal}\left(\poly(d, \frac{1}{\epsilon})\right)$ gates, and outputs a quantum state $|\Psi\rangle$ such that
    \begin{align*}
 \mathbb{P}_{|\Psi\rangle}[ f(X) - f(x^{\star})  \leq \epsilon] \geq \frac{3}{5}.
    \end{align*}
\end{theorem}
\begin{proof}
By Lemma~\ref{lem:hypercontractive-sufficient-lambda}, we only need to run the annealing algorithm up to $\lambda_{\textrm{max}} = \Ocal\left(\frac{d}{\epsilon}\right)$. On the other hand, by Theorem \ref{thm:instrin_hyper_perturb} the spectral gap satisfies $\delta^{Q}(\lambda) = \Omega\left(\lambda + \text{diam}(\Xcal)^{-2}\right), \forall \lambda \geq 0$. Hence we have a gap lower bound that has the same form as the one in Theorem \ref{thm:runtime_perturb_outside}. As in Theorem \ref{thm:runtime_perturb_outside}, we get $\theta = \mathcal{O}(\lambda_{\max}^2\Lambda^2\textup{diam}(\Xcal)^{6}) = \theta = \mathcal{O}(\Lambda^2 d^5R^6)$. The query, gate and qubit complexity in the theorem statement then follow from Theorem \ref{thm:adiabatic_simulation}.

\end{proof}
We again note that polynomial dependencies are overly pessimistic as we do not attempt to optimize parameters of adiabatic algorithm. However, the key significance of this finding is that the runtime of the algorithm does not depend on the shape or size of the barriers introduced due to the perturbation. Since $g$ is a growing function, for large $x$, the height of the barriers can possibly be large and hence trap the classical Langevin dynamics. In fact, $f$ can be highly oscillating. For example, for $k=1$ and $\|h(x)\|\leq C\|x\|$, one can take 
\[
f(x) = \|x\|^2 -\frac{h(x)}{d}\sin(d\|x\|^2)
\]
which is oscillating fast and has multiple local minima (See Figure~\ref{fig:strongly-convex-perturbation}).

We also note that the classical hypercontraction perturbation results~\cite{DEUSCHEL199030} have spectral gap bounds that fall as $\exp(-\|g\|_{\infty})$ which does not work for unbounded perturbations. On the other hand, one can always use the perturbation result for Schr\"odinger equations~\cite{gross2025invariance} by applying the perturbation analysis on Witten Laplacian. However, as the WKB potential can possibly contain multiple global minima due to the nonconvex perturbation, as can be seen in Figure ~\ref{fig:strongly-convex-perturbation}, %
one would still end up with an estimate that falls with $d$ exponentially as $\beta$ approaches $d$.

\begin{figure}
    \label{fig:strongly-convex-perturbation}
    \centering
    \includegraphics[width=1.0\linewidth]{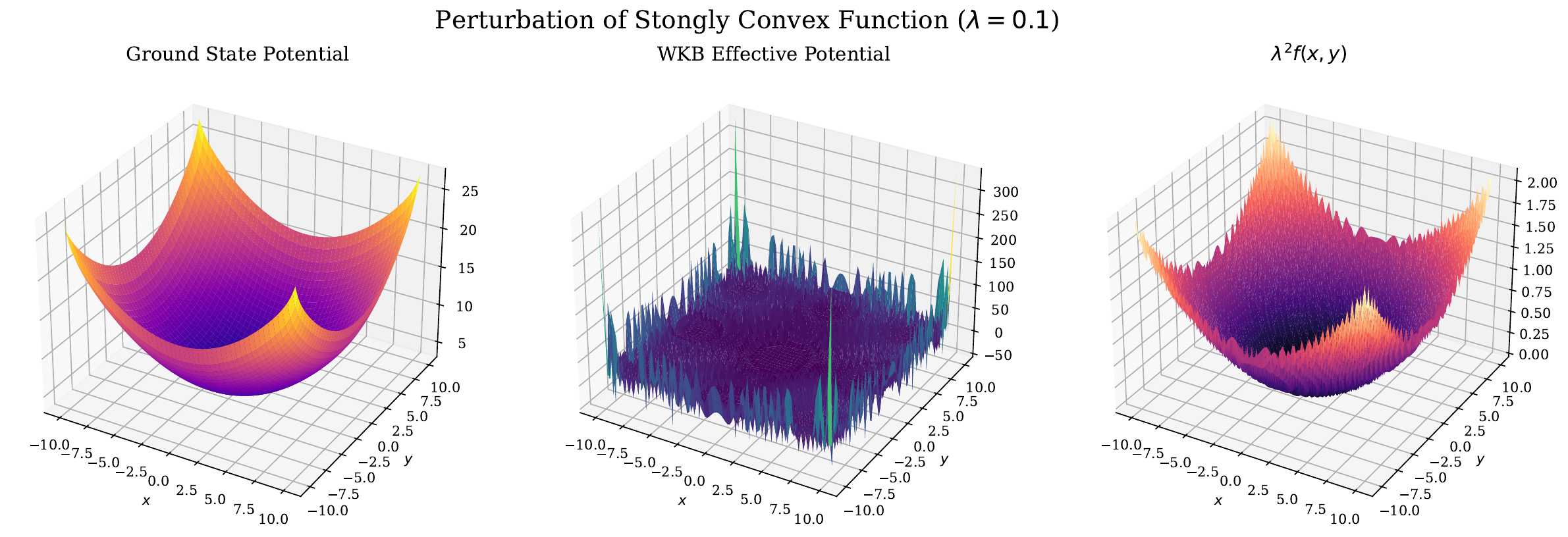}
    \includegraphics[width=1.0\linewidth]{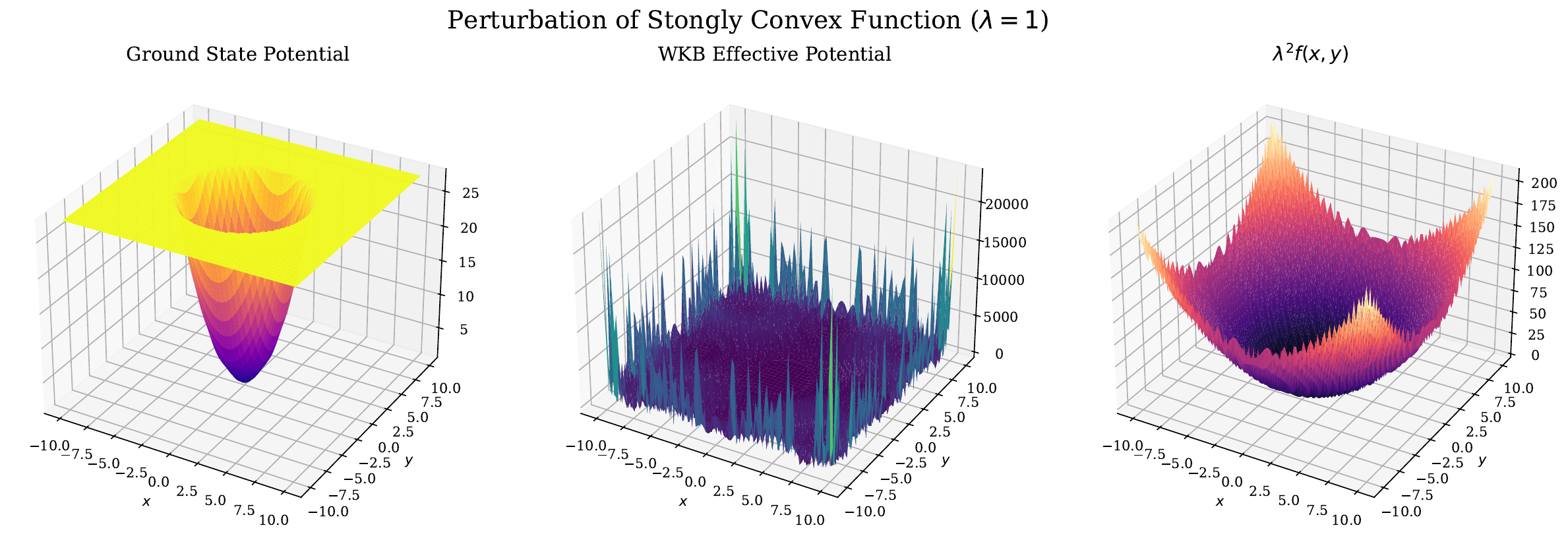}
    \includegraphics[width=1.0\linewidth]{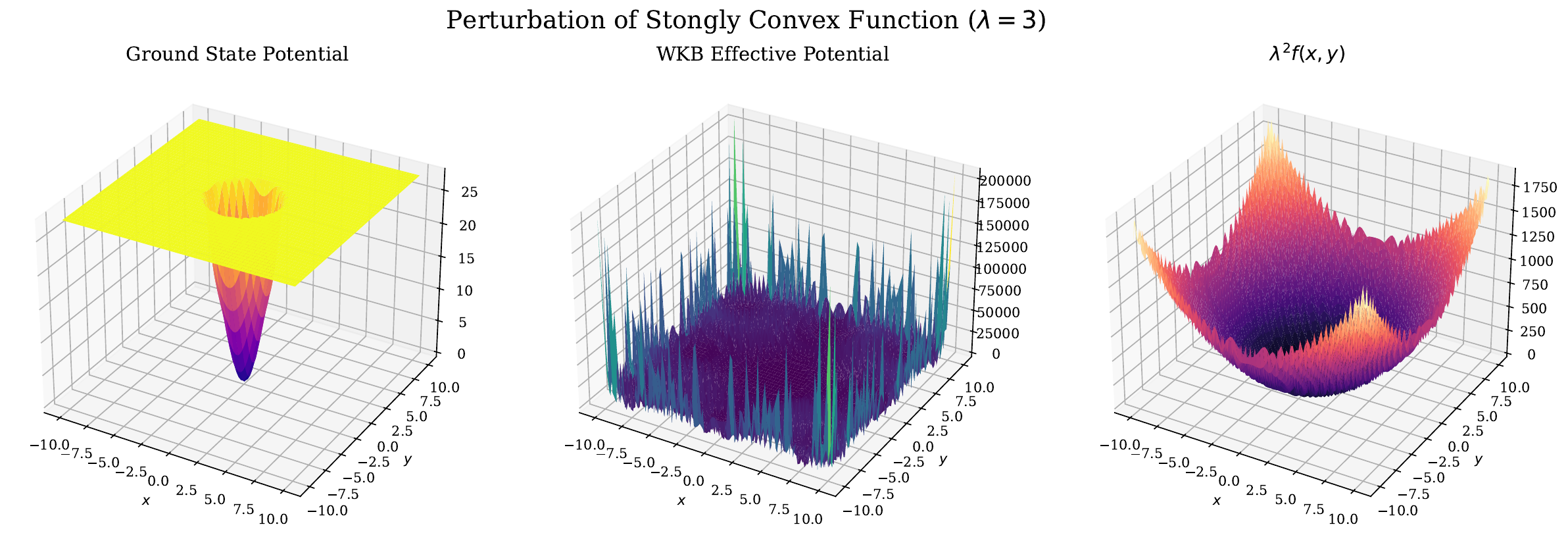}
    \caption{\textit{Visualization of the unique/multiple global minima separation for the strongly convex functions with perturbation:} $f(x, y) = x^2 + y^2 + \sqrt{x^2 + y^2} \cdot \sin(x^2 + y^2)$.}
    \label{fig:strongly-convex-3d}
\end{figure}

It is clear that neither function class constructed in this section needs to be (even approximately) block-separable as arbitrary perturbations satisfying the conditions can be added to the base function. In some ways, the construction of Theorem~\ref{thm:pert_quad_envelop} is more convenient as the new technical tools of hypercontractivity are not essential to the basic argument. However, this construction leaves a sufficiently large convex region around the global minimum to possibly permit algorithms that combine Langevin dynamics at $\beta = o(d)$ followed by gradient descent. The next construction removes this possibility and accommodates perturbations that have support everywhere, as long as they satisfy an appropriate growth condition. Considering a function such as $f = \|x\|^2 -\frac{h(x)}{d}\sin(d\|x\|^2)$ for some $|h|\leq x$, there are no large regions of local convexity around the global minimum and $\beta$ must be chosen to be a linear function in $d$. This requirement leads to an exponential cost for Langevin dynamics, and,  as discussed in Section \ref{sec:sqa}, even simulated quantum annealing. While we cannot theoretically analyze every algorithm, we note that these functions can exhibit all the typical properties of a highly nonconvex function, including an exponential number of sub-optimal local minima separated by possibly tall barriers. Hence, it is very unlikely that such functions can be numerically optimized by off-the-shelf algorithms.

\section{Analysis of the Real-space Adiabatic Algorithm}
\label{sec:algorithm_analysis}

In this section we prove the formal version of Theorem \ref{thm:adiabatic_simulation_inform}, resulting in Theorem \ref{thm:adiabatic_simulation}, which upper bounds the runtime of a digital version of the real-space adiabatic algorithm (Definition \ref{defn:schrodinger-anneal}) in terms of quantities that only depend on the potential and spectral gap. First we recall the standard definition of a noisy, quantum binary oracle for evaluating a function $f$.
\BinaryOracleDefn*
In all cases, we will assume access to an oracle of the above form and include bounds on the noise-tolerance $\epsilon_f$. However, we note that we do not make any effort to optimize the oracle noise tolerance.

To properly track the instantaneous ground state of the operator $H(\lambda) = -\Delta + \lambda^2f$, RsAA considers a time-rescaled version of the time-dependent Schr\"odinger evolution:
\begin{align}
\label{eqn:rescaled_evolution}
    \frac{1}{T}i\partial_s\Phi(x, s) = \Phi(x, s), \quad 0 \leq s \leq 1,
\end{align}
where $H(s)$ is 
\begin{align*}
    H(s) = -\Delta + (\lambda_{\max}^2\cdot s)f(x),
\end{align*}
and $T$ represents the evolution time scale. In $H(s)$, we have also fixed the adiabatic schedule to be a linear function of $s$, which is done for simplicity. It is definitely possible that more optimized schedules could lead to further polynomial reductions in the overall runtime.

The following is a precise version of the quantum adiabatic theorem  applicable to $H(s)$ that are unbounded operators.
\begin{theorem}[\cite{mozgunov2023quantum} Theorem 1]
\label{thm:adiabatic_thm}
Assume $\forall s \in [0, 1]$, there exists positive numbers $c_0, c_1$ such that the Hamiltonian $H(s)$ satisfies
\begin{align*}
    (H'(s))^2 \preceq c_0 + c_1H^2(s).
\end{align*}
Let $\Psi(s, x)$, $\delta(s)$ be the ground state and spectral gap of $H(s)$, respectively, and $U_{T}(s)$ the time-evolution operator of \eqref{eqn:rescaled_evolution}. Then
\begin{align*}
    \lVert \Psi(1, \cdot) - U_{T}(1)\Psi(0, \cdot)\rVert \leq \frac{\widetilde{\theta}}{T},
\end{align*}
where 
\begin{align*}
&\widetilde{\theta} = \tau^2(0)\lVert P_0H'(0)Q_0\rVert + \tau^2(1)\lVert P(1)H'(1) Q(1)\rVert \\
&+ \int_0^{1} \mathrm{d}s \left[\tau^3\left(5 \lVert PH'Q\rVert + 3 \lVert P H'P \rVert \right)\lVert P H' Q \rVert + \tau^2\lVert P H'' Q \rVert + 3\tau^3\sqrt{c_0 \lVert P H'Q \rVert^2 + c_1\lVert P H' HQ \rVert^2 }\right],
\end{align*}
\begin{align*}
    \tau(s) = \frac{1}{\delta(s)},
\end{align*}
and $P(s)$ is the orthogonal projector onto the instantaneous ground state space of $H(s)$, with $Q(s)$ the orthogonal complement.
\end{theorem}
The above roughly says that if the simulation time satisfies $T = \mathcal{O}\left(\frac{1}{\rho_{\text{adiabatic}} \delta_{\min}^3}\right)$, then $U_{T}$ produces an $\rho_{\text{adiabatic}}$-approximate (in trace distance) ground state of $H(1)$. The purpose of $\lambda_{\max}$ in RsAA is to control the localization of the ground state of $H(1)$ around the global minimum. When presenting the algorithm guarantees in this section, we leave $\lambda_{\max}$ as an input parameter. However, Lemma \ref{lem:hypercontractive-sufficient-lambda} provides a sufficient value of $\lambda_{\max}$ for producing an $\epsilon$-approximate optimizer with constant probability and under relatively weak assumptions.

We also make use of the following result on the digital simulation of time-dependent Schr\"odinger operators.
\begin{theorem}[{Adapted from \cite[Theorem 4.2]{chakrabarti2025speedups}}] 
\label{thm:schrodinger_sim}
   Suppose $f : \Xcal \rightarrow \mathbb{R}$ is $G$-Lipschitz when restricted to $x_0 + [-2R, 2R]^d \subseteq \Xcal$. Consider the Schr\"odinger equation
    \begin{align*}
    i \partial_t \Phi(x, t) &= [- a(t) \Delta + b(t)f(x)]\Phi(x, t),
    \end{align*}
    subject to initial data $\Phi(x,0) = \Phi_0(x)$, where $a, b$ are sufficiently smooth functions of time and $\Lambda \geq \norm{f}_\infty$, accessed through a $\epsilon_f$-noisy binary quantum oracle $O_f$.  If $\epsilon_f = \widetilde{\Ocal}(\epsilon/\lVert b \rVert_1)$, then there is a digital quantum algorithm that outputs a $\ket{\Psi_t}$ such that for any $\lVert h \rVert_{\textup{Lip}}$-Lipschitz $h$:
    \begin{align}
    \label{eqn:sim_guarantee}
        \bigg\lvert (2N)^{-d}\sum_{x_{\jmbf} \in \Gcal} {h}(2Rx_{\jmbf})\lvert \Psi_{T}(x_\jmbf)\rvert^2 - \int_{S}{h}(2Ry)\lvert \Phi(y, T)\rvert^2 dy \bigg\rvert \leq \lVert h \rVert_{\infty}\epsilon,
    \end{align}
    using $\Ocal\left(\Lambda \lVert b \rVert_1\right)$ queries to $O_f$, 
    $\Ocal\left(d^2\cdot \polylog(1/\epsilon, \lVert h \rVert_{\text{Lip}}, G, \lVert a \rVert_1, \lVert b \rVert_1)\right)$ qubits and $\widetilde{\Ocal}\left(\poly(d, \Lambda, \lVert b \rVert_1)\right)$ gates.
\end{theorem}

The above theorem shows that there is an algorithm that outputs a digital state $|\Psi\rangle$, where measuring observables on $|\Psi\rangle$ approximates measuring the same observable on the true continuous state $\Phi$ over $\Xcal$. Hence, we apply Algorithm 1 from \cite{chakrabarti2025speedups} as a subroutine to get RsAA (Algorithm \ref{alg:adiabatic_schro}).
\begin{algorithm}  
    \caption{Real-space Adiabatic Algorithm} \label{alg:adiabatic_schro}

    \begin{flushleft}
    \textbf{Input:} Input $\lambda_{\max} > 0$; $x_0 \in \Xcal$; precision parameters $\epsilon, \rho > 0$; dimension $d$; upper bound $\Lambda \geq \norm{f}_\infty$;  radius of simulation $R$; $\epsilon_f$-accurate quantum oracle $O_f$ for real function $f$ on $\Xcal$, with $\epsilon_f$ obeying the bound in Thm~\ref{thm:adiabatic_simulation}; $T_{\text{adiabatic}}$ obeying Thm~\ref{thm:adiabatic_simulation}. \vspace{-2mm} 
    \end{flushleft}
    \begin{flushleft}
    \textbf{Output:} A digital quantum state $|\Psi\rangle$ such that if $\Phi_{\lambda_{\max}}$ is the ground state of  $-\Delta + \lambda^2_{\max}f(x)$ then return $\tilde{x} \in x_0 + [-2R, 2R]^d$ such that
    \begin{align*}
        \mathbb{P}_{|\Psi\rangle}\left[ \lVert X - y\rVert < \epsilon\right] > \mathbb{P}_{\lvert \Phi_{\lambda_{\max}}\rvert^2}\left[ \lVert X - y\rVert < \frac{\epsilon}{2}\right] - \rho,
    \end{align*}
    for any $y \in \mathcal{X}$. \vspace{-2mm}
    \end{flushleft}
    \begin{flushleft}
    \textbf{Procedure:}
    \vspace{-2mm}
    \end{flushleft}
    \begin{algorithmic}[1]
        \State Let $U_0$ be a unitary that prepares the discrete uniform superposition over $[-\frac{1}{2}, \frac{1}{2}]^d$. 
        \State Set $a = t, b(t)= \lambda_{\max}^2 t$, $t \in [0, T_{\text{adiabatic}}]$.
        \State Run the digital Schr\"odinger simulation (Algorithm 1 \cite{chakrabarti2025speedups}) with the above inputs, including $U_0$, to prepare a digital state $|\Psi\rangle$.
        \State Measure $|\Psi\rangle$ in the computational basis obtaining outcome $x$, $x \rightarrow 2Rx = \tilde{x}$
        \State \textbf{Return} $\tilde{x}$ . 
    \end{algorithmic}
\end{algorithm}

 We can combine Theorem \ref{thm:schrodinger_sim}
  with Theorem \ref{thm:adiabatic_thm} to obtain a runtime bound on Algorithm \ref{alg:adiabatic_schro} given only a bound on the spectral gap and $\lambda_{\max}$. We defined the notation $\mathbb{P}_{|\Psi\rangle}$, where $|\Psi\rangle$ is a digital quantum state, in Section \ref{sec:quantum_runtime_separable}. For a continuous wave function $\Phi(x) : [-\frac{1}{2}, \frac{1}{2}]^d \rightarrow \mathbb{C}$, we analogously use $\mathbb{P}_{\lvert \Phi\rvert^2}[X \in A]$ to denote the corresponding measure obtained when measuring in the position basis.
\begin{theorem}[Digital Simulation of Adiabatic Schr\"odinger Operators]
\label{thm:adiabatic_simulation}
     Suppose $f : \Xcal \rightarrow \mathbb{R}$ is $G$-Lipschitz when restricted to $x_0 + [-2R, 2R]^d \subseteq \Xcal$. Consider the evolution
    \begin{align*}
       \frac{1}{T}i\partial_s\Phi(s, x) = (-\Delta + (\lambda_{\max}^2 \cdot s) f(x))\Phi(s, x), \quad s \in [0, 1], \quad x \in \left[-\frac{1}{2}, \frac{1}{2}\right]^{d}.
    \end{align*}
    Let $\Phi_{\lambda_{\max}}(x)$ be the ground state of $H(1)$,  $\delta(s)$ be the spectral gap of $H(s)$ so that
    
    \begin{align*}
        \theta = \lambda_{\max}^2\Lambda\delta^{-2}(1) + 12\int_0^1  \delta^{-3}(s) \lambda^4_{\max} \Lambda^2 \mathrm{d}s,
    \end{align*}
    and $T_{\text{adiabatic}} = \mathcal{O}\left(\frac{\lambda_{\max}^2\Lambda\theta}{ \rho_{\text{adiabatic}}}\right)$. Then there is a digital quantum algorithm that outputs a quantum state $|\Psi\rangle$ such that for any $y \in \Xcal, \mathcal{B}_2(y, \epsilon) \subseteq \mathcal{X}$: 
    \begin{align*}
         \mathbb{P}_{|\Psi\rangle}[ X \in \mathcal{B}_2(y, \epsilon)] > \mathbb{P}_{\lvert \Phi_{\lambda_{\max}}\rvert^2} \left[ X \in \mathcal{B}_2\left(y, \frac{\epsilon}{2}\right)\right] - \rho_{\text{sim}} - \rho_{\text{adiabatic}}.
    \end{align*}
    The algorithm starts from the discrete uniform superposition over $x_0 + [-2R, 2R]^d$, uses $\mathcal{O}\left(\frac{\lambda_{\max}^2\Lambda \theta}{ \rho_{\text{adiabatic}}} \right)$ queries to an $\epsilon_f = \widetilde{\Ocal}(\frac{\rho_{\text{adiabatic}}\rho_{\text{sim}}}{\lambda_{\max}^2 \Lambda \theta})$ accurate binary oracle, $\Ocal\left(d^2\cdot \polylog(1/\rho_{\text{sim}},1/\rho_{\text{adiabatic}}, 1/\epsilon, R, G, \lambda_{\max}, \Lambda, \theta\right)$ qubits, and $\widetilde{\Ocal}\left(\poly(d, \lambda_{\max},  \Lambda, 1/\rho_{\text{adiabatic}}, \theta)\right)$ gates.
\end{theorem}
\begin{proof}

We first start with the error from the approximate adibatic evolution, using Theorem \ref{thm:adiabatic_thm}. Note that the first condition of that theorem is  satisfied with $c_1=0$. Specifically,
\begin{align*}
    (\lambda_{\max}^2 f(x))^2 \preceq \lambda_{\max}^4 \Lambda^2.
\end{align*}
Note, since we are considering a bounded domain, $\Delta$ is the Dirichlet Laplacian. Hence, $s = 0$ is leads to a gapped operator, with the ground state being the uniform superposition over $\Xcal$.

For $H(s) = -\Delta + (\lambda_{\max}^2 \cdot s)f(x)$, we have
\begin{align*}
\widetilde{\theta} &= \tau^2(0)\lVert P_0H'(0)Q_0\rVert + \tau^2(1)\lVert P(1)H'(1) Q(1)\rVert \\
&+ \int_0^{1} \mathrm{d}s \left[\tau^3\left(5 \lVert PH'Q\rVert + 3 \lVert P H'P \rVert \right)\lVert P H' Q \rVert + \tau^2\lVert P H'' Q \rVert + 3\tau^3\sqrt{c_0 \lVert P H'Q \rVert^2 + c_1\lVert P H' HQ \rVert^2 }\right]\\
&\leq \lambda_{\max}^2 \Lambda \tau^2(1) + 12\int_0^1  \tau^3(s) \lambda^4_{\max} \Lambda^2 \mathrm{d}s =: \theta.
\end{align*}
This leads to a simulation time of $T_{\text{adiabatic}} = \mathcal{O}\left(\frac{\theta}{\rho_{\text{adiabatic}
}} \right)$. 

For the simulation error, we make use of Theorem \ref{thm:schrodinger_sim}. We will take $h$ in that theorem to be a radial, smooth, bump function $\phi$ such that $\phi(x) =1$ on $\mathcal{B}_2(y, r/2)$ and has support equal to $\mathcal{B}_2(y, r)$. The Lipschitz constant of $\phi$ is $\Theta(\frac{1}{r})$. Then \eqref{eqn:sim_guarantee} together with the adiabatic guarantee and triangle inequality imply
\begin{align*}
    \lvert \mathbb{P}_{|\Psi\rangle}[ X \in \text{supp}~\phi] - \mathbb{P}_{\lvert \Phi_{\lambda_{\max}}\rvert^2}[ X \in \text{supp}~\phi] \rvert < \rho_{\text{sim}} + \rho_{\text{adiabatic}} := \rho_{\text{tot}}.
\end{align*}
Then,
\begin{align*}
     \mathbb{P}_{|\Psi\rangle}[ X \in \mathcal{B}_2(y, r)] &\geq \mathbb{P}_{|\Psi\rangle}[ X \in \text{supp}~\phi] \\&> \mathbb{P}_{\lvert \Phi_{\lambda_{\max}}\rvert^2}[ X \in \text{supp}~\phi] - \rho_{\text{tot}}\\
     &> \mathbb{P}_{\lvert \Phi_{\lambda_{\max}}\rvert^2}[ X \in \mathcal{B}_2(y, r/2)] - \rho_{\text{tot}},
\end{align*}
as desired.
Lastly,  the total query complexity will be $\mathcal{O}\left(\Lambda \lambda_{\max}^2 T_{\text{adiabatic}}\right)$. The gate and qubit counts follow similarly, by plugging quantities into Theorem \ref{thm:schrodinger_sim}.
\end{proof}

\section{Spectral Analysis of Schr\"odinger Operators: Technical Results}
\subsection{Semiclassical Analysis}
\label{sec:semiclassical_analysis}

In this section, we derive the spectral analysis tools that were used in Sections \ref{sec:quantum_runtime_separable} and \ref{sec:removing_separability}. These are based on refined versions of techniques from semiclassical analysis.

Semiclassical analysis \cite{zworski2012semiclassical}  is a suite of techniques for analyzing Schr\"odinger operators:
\begin{align*}
    H(\lambda) = -\Delta + \lambda^2f(x),
\end{align*}
as $\lambda \rightarrow \infty$. In physical terms, this can be related to a vanishing Planck's constant, and hence corresponds to approximately analyzing Schr\"odinger operators in the classical limit. Specifically, in Section \ref{subsec:separation_mechanism}, we referred to the result of Simon \cite{simon1983semiclassical, cycon1987schrodinger}, which stated that as $\lambda \rightarrow \infty$, any Schr\"odinger operator tends to the direct sum of QHO's centered at the global minimizers. In the unique-global-minimizer case, this can be used to show that the operator remains gapped when the Hessian at the global minimum is non-degenerate. For the case of multiple global minimizers,  another result from Simon \cite{simon1983semiclassical} shows that the gap falls exponentially with the \emph{Agmon distance} (Section \ref{sec:tunneling_degeneratecase}/ Equation \eqref{eqn:agmon_metric}) between global minimizers.

One major issue with the result concerning potentials with a unique global minimum is that the theorem has not been shown to apply when $\lambda$ is at most polynomial in $d$. However, to be useful for algorithmic purposes, one would require a version of this result that allows for the dimension to be an asymptotic parameter and $\lambda$ a function of $d$. Focusing on the unique-global-minimizer case, we present refined versions of the existing semiclassical results, which showcase the asymptotic dependence on dimension and all corrections. Thus our contribution can be viewed as, to the best of our knowledge, the first step in transferring results from semiclassical analysis into results that can make computational statements.

The first of our results, and arguably the main result of this section, is a comparison theorem (Theorem \ref{thm:loc_spec_compare}) for relating the gaps of two Schr\"odinger operators whose potentials are locally similar around their unique global minimizer $x^{\star}$. This result is then used to derive a non-asymptotic version of \eqref{eqn:semiclass_nondegen}, but for the spectral gap (Theorem \ref{thm:semiclassical_local_taylor_version}). This reveals that, with certain assumptions, if $\lambda = \Theta(d^{3+\epsilon})$ for arbitrarily small but non-zero $\epsilon$, then the gap of $H(\lambda)$ is lower bounded by $\Omega\left(\lambda\sqrt{\mu}\right)$, where $\mu$ is the minimum eigenvalue of the Hessian at the global minimizer. Note that the coefficient $\mu$ is the same that is predicted by \eqref{eqn:semiclass_nondegen}, but shows the corrections and the scaling of $\lambda$ in $d$ required for the result to apply.

A direct consequence of Theorem \ref{thm:semiclassical_local_taylor_version} is that for block-separable potentials with unique global minimizer (Assumption \ref{assump:constant_block}), the spectral gap is larger than a dimension-independent constant for all $\lambda$ greater than some constant. Corollary \ref{cor:separable_gap} of Theorem \ref{thm:semiclassical_local_taylor_version}  shows that for a block-separable function there exists a constant $\lambda_{\star}$ after which the gap remains bounded below by $\Omega(\lambda)$. The fact that the gap is bounded below by a constant for all $\lambda$ then follows from a result due to Yau \cite{yau2009gapeigenvaluesschrodingeroperator}, Theorem \ref{thm:yau_nonconvex}, once combined with Corollary \ref{cor:separable_gap}. This leads to Theorem \ref{thm:gap_bound_for_sep}. Again, this is a setting where there is a significant divergence from the classical stochastic dynamics.  Langevin dynamics, even for a separable potential, has a gap that falls exponentially with $\beta$ (Theorem \ref{thm:lang_sgd_gap}). Additionally, under the assumption of quadratic growth, we can improve the $d^{3+\epsilon}$ dependence of Theorem \ref{thm:semiclassical_local_taylor_version} to $d$ via Theorem \ref{thm:quad_semiclassical_local_taylor_version}.

One of the challenges with applying Theorem \ref{thm:loc_spec_compare} is that the corrections depend on the localization of the first-excited and ground states of the operator we want to analyze. Additionally, the $\Omega(d)$ dependence in Theorem \ref{thm:quad_semiclassical_local_taylor_version} is still not sufficient for determining the performance of RsAA.
To get around these issues for non-separable potentials, we utilize Agmon's theorem \cite{agmon2014lectures}, in the form presented in  \cite{steinerberger2021effectiveboundsdecayschrodinger}, for quadratically-enveloped functions, to show that the QHO-gap lower bound still holds for such potentials. Additionally, these potentials can be nonconvex with many local minima. This leads to Theorem \ref{thm:pert_quad_envelop}, which requires a precise use of Agmon Localization derived in the appendix (Lemma \ref{lem:localization_bound_quadratic_pot}). The dependence of $\lambda$ on $d$ can be brought all the way down to a constant, and a result from \cite{gross2025invariance} can be used to handle the $\lambda = o_d(1)$ regime, as done in Section \ref{sec:removing_separability} for the proof of Theorem \ref{thm:runtime_perturb_outside}.

While it seems approaches based on semiclassical analysis struggle to imply bounds for all $\lambda$, without moving the perturbation far from the global minimizer, it appears that if $f$ is a perturbed version of a strongly convex function $g$, such that $f - g \geq 0$, we can uniformly bound the gap in $\lambda$. This leads to Corollary \ref{cor:positive_pert}, which will follow directly from Theorem \ref{thm:loc_spec_compare}. Specifically, the semiclassical analysis results can achieve something similar to the hypercontractivity-based results mentioned in Section \ref{sec:hypercontractivity}, but only for positive perturbations.

\subsubsection{Local Spectral Comparison of Schr\"odinger Operators}
\label{sec:local_spec_compare}

We first start with some notational conventions. Let $J(x)$ be a $C^{\infty}(\mathbb{R}^d)$ function with a compact support, such that $\lVert J \rVert_{\infty} = c$. Also define the complement  $J_0(x) = \sqrt{ 1 - J(x)^2}$, so that $J^2, J_0^2$ form a partition of unity for $\mathbb{R}^d$.
In most cases, we will just end up setting $c = \frac{1}{2}$, and considering
\begin{align}
\label{eqn:usual_bump}
    J(x) = ce^{-\frac{1}{1-\lVert x \rVert^2}}.
\end{align}
Then, to scale the domain to an $\ell_2$ ball with radius $r$, we simply consider $J(x/r)$.

When working with partitions of unity, we frequently use an identity, which is useful for approximately commuting elements of a partition of unity through Schr\"odinger operators. Specifically, for any twice-differentiable $\chi(x)$ in the domain of $H$, where $H$ is the  Schr\"odinger operator $-\Delta + h(x)$:
\begin{align}
\label{eqn:ims_identity}
    \chi H\chi = \frac{\chi^2 H}{2} + \frac{H\chi^2}{2} + \lVert \nabla \chi \rVert^2.
\end{align}
While this identity may seem simple, it is an extremely powerful tool \cite{cycon1987schrodinger}, as will become apparent when reading the below proofs. It will be used to commute certain operators used extensively in the below proofs.

One immediate consequence is the following \emph{IMS Localization formula} (due to Ismagilov, Morgan, and Simon) :
\begin{Lemma}[Theorem 3.2 \cite{simon1983semiclassical}, Adapted]
\label{lem:ims}If $\{\chi_i\}_{i=0}^{m}$ form a partition of unity for $\mathbb{R}^d$, and $H$ is a Schr\"odinger operator, then
\begin{align*}
    H = \sum_{i \in [m]} \left(\chi_i H_i \chi_i - \lVert \nabla \chi_i \rVert^2\right),
\end{align*}
where $\sum_{i \in [m]} \lVert \nabla \chi_i \rVert^2$ is called the Localization Error. 
\end{Lemma}
Due to the apparent, intimate connection between the identities, we sometimes also refer Equation \eqref{eqn:ims_identity} as the IMS Localization formula as well.

The following is our main tool inspired by techniques from semiclassical analysis for bounding the spectral gaps of Schr\"odinger operators. After, we will prove some consequences that are relevant for many applications, specifically used in Sections \ref{sec:quantum_runtime_separable} and \ref{sec:removing_separability}. 
\begin{theorem}[Local Spectral Comparison of Potentials]
\label{thm:loc_spec_compare}
Suppose $f, g : \mathcal{X} \rightarrow \mathbb{R}$ are functions with unique global minimum both at $x^{\star}$. Consider the Schr\"odinger operators
\begin{align*}
&H_f(\lambda) = -\Delta + \lambda^2f(x)\\
&H_g(\lambda) = -\Delta + \lambda^2g(x),
\end{align*}
with  $(\xi_1, E_1(\lambda)), (\xi_2, E_2(\lambda))$ and $(\psi_1, e_1(\lambda)), (\psi_2, e_2(\lambda))$ corresponding to the first and second eigen-pairs of $H_f$ and $H_g$ respectively. Then $\forall \lambda \in \mathbb{R}_+$ and $\forall \xi \in \text{span}\{ \xi_1, \xi_2\} \cap (\text{span}\{J\psi_1\})^{\perp}$:
\begin{align}
\label{eqn:compare_gap_bound}
    E_2(\lambda) - E_1(\lambda) &\geq [e_2(\lambda) - e_1(\lambda)] \nonumber \\ &\quad+\langle J_0\xi| H_f(\lambda) - e_2(\lambda)|J_0\xi\rangle - \frac{2-c^2}{1- c^2} \sup_{x \in \mathbb{R}^d} \lVert \nabla J(x)\rVert^2 \nonumber\\
    &\quad+\langle J\xi| \lambda^2(f-g)|J\xi\rangle - \frac{\langle J\psi_1| \lambda^2(f - g) | J\psi_1\rangle}{\lVert J\psi_1\rVert^2}.
\end{align}
\end{theorem}

Note as a special case if $\forall x \in \mathbb{R}^d, f(x) \geq g(x)$, we can take $J$ to be the identity and obtain a simpler expression
\begin{align}
\label{eqn:positive_case}
    E_2(\lambda) - E_1(\lambda) &\geq [e_2(\lambda) - e_1(\lambda)] - {\langle \psi_1| \lambda^2(f - g) | \psi_1\rangle}.
\end{align}
This is reminiscent of the energy change provided by first-order Rayleigh-Schr\"odinger perturbation theory. %

\begin{proof}

The proof is heavily inspired by the techniques of \cite{simon1983semiclassical}.
We can without loss of generality assume $g(x^{\star}) = f(x^{\star}) = x^{\star}=0$, as the shift does not impact the spectrum. Note that by H\"older's inequality for any $\psi$
\begin{align*}
    \frac{\langle J \psi | \lVert \nabla J \rVert^2 | J\psi\rangle}{\lVert J \psi\rVert^2} \leq \sup_{x \in \text{supp} J }\lVert \nabla J(x)\rVert^2 \leq \sup_{x \in \mathbb{R}^d}\lVert\nabla J(x)\rVert^2,
\end{align*}
and similarly for $J(H_f - H_g)J$.

By the above, the variational principle  and IMS:
\begin{align}
E_1(\lambda) &\leq \frac{\langle J\psi_1 | H_g |J\psi_1\rangle}{\lVert J\psi_1 \rVert^2} + \frac{\langle J\psi_1| \lambda^2(f - g) | J\psi_1\rangle}{\lVert J\psi_1\rVert^2}  \nonumber\\
& \leq e_1(\lambda) + \sup_{x \in \mathbb{R}^d}\lVert\nabla J(x)\rVert^2+ \frac{\langle J\psi_1| \lambda^2(f - g) | J\psi_1\rangle}{\lVert J\psi_1\rVert^2}  .
\label{eqn:e1_upper_bound}
\end{align}
From IMS Localization and the fact that $J^2, J_0^2$ form a partition of unity gives that
\begin{align}
 H_f(\lambda)   &=  JH_g(\lambda)J +J[H_f(\lambda) - H_g(\lambda)]J +  J_0H_f(\lambda)J_0 -\lVert\nabla J\rVert^2  - \lVert\nabla J_{0} \rVert^2 \nonumber  \\
 &=  JH_g(\lambda)J  -\lVert\nabla J\rVert^2  - \lVert\nabla J_{0} \rVert^2  + J\lambda^2(f-g)J + J_0H_f(\lambda)J_0
 \label{eqn:hf_lower_bound}
\end{align}

\begin{align*}
JH_g(\lambda)J &\succeq e_2(\lambda)J P_{n \geq 2} J + e_1(\lambda) J|\psi_1\rangle \langle \psi_1|J\\
&= e_2(\lambda ) J^2 - (e_2(\lambda )-e_1(\lambda ))|J\psi_1\rangle\langle J\psi_1|,
\end{align*}
where $P_{n\geq 2}$ projects on eigenstates of $H_g$ with index $n \geq 2$.

Let $\xi \in \text{span}(\{ \xi_1, \xi_2\}) \cap [J\psi_1]^{\perp}$ with $\lVert {\xi} \rVert =1$. Then by the variational principle and combining the above two estimates, we have
\begin{align}
\label{eqn:e2_lower_bound}
    E_2(\lambda) \geq \langle {\xi}| H_f(\lambda) | {\xi}\rangle &\geq e_2(\lambda)  + \langle \xi| J_0[H_f(\lambda) - e_2(\lambda)]J_0|\xi\rangle  \nonumber\\&- \sup_{x \in \mathbb{R}^d}\lVert\nabla J_{0}(x) \rVert^2 - \sup_{x \in \mathbb{R}^d}\lVert\nabla J(x) \rVert^2 + \langle J\xi| \lambda^2(f-g)|J\xi\rangle.
\end{align}

Combining Equations \eqref{eqn:e1_upper_bound} and \eqref{eqn:e2_lower_bound}
\begin{align*}
    E_2(\lambda) - E_1(\lambda) &\geq [e_2(\lambda) - e_1(\lambda)] \\ &+\langle \xi| J_0[H_f(\lambda) - e_2(\lambda)]J_0|\xi\rangle\\   &- 2\sup_{x \in \mathbb{R}^d} \lVert \nabla J(x)\rVert^2 - \sup_{x \in \mathbb{R}^d}\lVert\nabla J_{0}(x) \rVert^2 \\&+\langle J\xi| \lambda^2(f-g)|J\xi\rangle - \frac{\langle J\psi_1| \lambda^2(f - g) | J\psi_1\rangle}{\lVert J\psi_1\rVert^2} 
\end{align*}

However for the radial smooth bump function $J(x) = ce^{-\frac{1}{1-\lambda^{2\alpha}\lVert x \rVert^2}}$:
\begin{align*}
      \lVert \nabla J_0 \rVert^2 = \frac{J^2(x)}{1- J(x)^2} \cdot \lVert \nabla J(x)\rVert^2 \leq \frac{c^2}{1- c^2} \sup_{x \in \mathbb{R}^d} \lVert \nabla J(x)\rVert^2.
\end{align*}

Hence

\begin{align*}
    E_2(\lambda) - E_1(\lambda) &\geq [e_2(\lambda) - e_1(\lambda)] \\&+\langle \xi| J_0[H_f(\lambda) - e_2(\lambda)]J_0|\xi\rangle - \frac{2-c^2}{1- c^2} \sup_{x \in \mathbb{R}^d} \lVert \nabla J(x)\rVert^2\\
    \\&+\langle J\xi| \lambda^2(f-g)|J\xi\rangle - \frac{\langle J\psi_1| \lambda^2(f - g) | J\psi_1\rangle}{\lVert J\psi_1\rVert^2} .
\end{align*}

\end{proof}

Our first application is to the ``traditional semiclassical'' regime where $\lambda \rightarrow \infty$. However, here we replace this limiting condition with only that $\lambda$ grows asymptotically with $d$. This leads to a more refined version of \cite{simon1983semiclassical}, where all corrections involving $d$ are present. As a result, we obtain a quantity that is more useful for algorithmic applications.

In this setting, we suppose that $J$ has support on an $\ell_2$ ball of radius $\lambda^{-\alpha}$ and $c = \frac{1}{2}$. Then a simple computation applied to Equation \eqref{eqn:usual_bump} shows that $\lVert \nabla J \rVert^2 \leq \lambda^{2\alpha}$. Let $g$ be the second-order approximation of $f$ around the unique global minimizer $x^{\star}$, i.e.
\begin{align*}
    H_g = -\Delta + \frac{\lambda^2}{2}\langle (x-x^{\star}) , \nabla^2f(x^{\star}) (x-x^{\star})\rangle.
\end{align*}

When $H_g$ is a Dirichlet operator, i.e. a truncated QHO, we can use the following result.
\begin{theorem}[\cite{yau2009gapeigenvaluesschrodingeroperator}]
\label{thm:yau_strongly_convex}
    Suppose $\Xcal$ is convex and $f$ is $\mu$-strongly convex. Then the spectral gap $\delta^{(Q)}(\lambda)$ of the Dirichlet operator $H(\lambda) = -\Delta + \lambda^2 f$ satisfies
    \begin{align*}
        \delta^{(Q)} \geq  \sqrt{\mu}\lambda +  \frac{1}{4\textup{diam}(\Xcal)^2} .
    \end{align*}
\end{theorem}
\begin{proof}
    Follows from Lemma 1.1 combined with Theorems 1.1 and 1.2 of \cite{yau2009gapeigenvaluesschrodingeroperator} with $\beta = 1$.
\end{proof}

We now present and derive a refined version of \cite{simon1983semiclassical}, that only focusees on the eigengap but applies when $d$ is an asymptotically growing quantity. To best of our knowledge, this is the first ``semiclassical'' result on Schr\"odinger operators that presents the exact dependence on dimension. In addition, the exact corrections in terms of $\lambda$ are computed.

\begin{restatable}[Refined Semiclassical Eigengap Estimate for Unique Minimum]{theorem}{semiclasLocalTaylor}
\label{thm:semiclassical_local_taylor_version}
Suppose $f \in C^3(\mathcal{X})$, has a unique global minimizer $x^{\star}$ with
\begin{align*}
    0 \prec \mu_{\star} I \preceq \nabla^2f(x^{\star}) \preceq L_{\star} I,
\end{align*}
and at least one other local minimizer.
 Define $\gamma = \sup_{\mathcal{B}_2(x^{\star}, \lambda^{-\alpha})} \lVert \nabla^3 f(x)\rVert_{\text{op}}$, for some $\alpha >0$. Suppose $y^{\star}$ is the closest local minimizer to $x^{\star}$. Then for any $\alpha \in (1/3, 1/2)$ 
and 
\begin{align*}
    \lambda = \Omega\left(\max\left\{\left(\frac{\sqrt{L_{\star}}d}{\mu_{\star}}\right)^{\frac{1}{1-2\alpha}}, ~\gamma^{-\frac{1}{1-3\alpha}}, ~\frac{1}{\lVert y^{\star} - x^{\star}\rVert}\right\}\right)
\end{align*}
we have that the spectral gap of the operator
$$H = -\Delta + \lambda^2 f(x)$$
satisfies
\begin{align*}
\delta^{(Q)}(\lambda) 
        \geq \sqrt{\mu_{\star}}\lambda - \frac{7}{3}\lambda^{2\alpha} - \frac{\gamma}{3}\lambda^{2-3\alpha}.
\end{align*}
\end{restatable}

We note that the condition of at least one other local minimizer is not necessary. In the case of only one minimizer, we can simply drop the third condition on $\lambda$ in the $\max$. Also, as an example, if we take $\alpha = \frac{2}{5}$, then we get $\lambda$ needs to grow like $\Omega(d^5)$. However,  as will be apparent from the proof, if for sufficiently large $d$, we have 
\begin{align*}
 \sqrt{\mu} - \frac{\text{sup}_{\mathcal{B}(x^{\star}, d^{-1/3-\epsilon})}\lVert \nabla^3 f(x)\rVert_{\text{op}}}{3} = \Omega(1),   
\end{align*}
then we can take $\alpha = \frac{1}{3} + \epsilon$ for $\epsilon$ arbitrarily small to get an $\Omega(\lambda)$ gap lower bound on $H_f(\lambda)$, for $\lambda = \Omega(d^{3 +\epsilon})$.

\begin{proof}
To apply Theorem \ref{thm:loc_spec_compare}, we consider comparing:
\begin{align*}
&H_f = -\Delta + \lambda^2 f(x)\\
&H_g = -\Delta + \frac{\lambda^2}{2} \langle x, \nabla^2f(0), x \rangle,
\end{align*}
where we again assume without loss of generality $x^{\star} = f(x^{\star}) = 0$. We will consider $J$ to be a radial smooth bump function centered at $x^{\star}$ around an $\ell_2$ ball with radius $\lambda^{-\alpha}$ via Equation \eqref{eqn:usual_bump}.

We can uniformly bound the perturbation on the support of $J$ by Taylor's theorem:
\begin{align*}
\lVert J^2 f - J^2 g  \rVert_{\infty} \leq \frac{\text{sup}_{\text{supp} J}\lVert \nabla^3 f(x) \rVert_{\text{op}}}{6} \lambda^{-3\alpha} =: \zeta,
\end{align*}
 Hence
\begin{align*}
\lvert \langle J\xi| \lambda^2(f-g)|J\xi\rangle - \frac{\langle J\psi_1| \lambda^2(f - g) | J\psi_1\rangle}{\lVert J\psi_1\rVert^2}\rvert &\leq 2\lambda^2\lVert J(f-g)J\rVert_{\infty} \\
&\leq 2 \lambda^2\zeta.
\end{align*}

We now deal with the error term 
\begin{align*}
&\langle \xi | J_0 H_f J_0 | \xi\rangle - e_2(\lambda) \lVert J_0 \xi \rVert^2.
\end{align*}
Since $H_g$ is a harmonic oscillator, we know that $e_2(\lambda) \leq  \frac{3}{\sqrt{2}}\lambda d\sqrt{L}$. We also have  by the Taylor error that on the boundary of the support of $J$, $f$ must satisfy
\begin{align*}
    \lambda^2 f(x)  &\geq \frac{\lambda^2}{2}\langle x, \nabla^2 f(0), x\rangle - \lambda^2\zeta\\
    &\geq \frac{\mu_{\star}}{2} \lambda^{2-2\alpha} - \lambda^2\zeta.
\end{align*}
Let $y^{\star}$ be the location of the closest local minimizer to $x^{\star}=0$, then for all $\lambda =\Omega( \frac{1}{\lVert y^{\star} - x^{\star}\rVert})$, the above is a global lower bound on the support of $J_0$. Then for such a $\lambda$
\begin{align}
\label{eqn:positiveness_condition}
    J_0H_fJ_0 - e_2(\lambda) \geq
    \lambda\left(\frac{\mu_{\star}}{2} \lambda^{1-2\alpha} - \frac{\lambda^{1-3\alpha}}{6} \gamma - \frac{3}{\sqrt{2}} d\sqrt{L_{\star}}\right),
\end{align}
where due to continuity we know that $\gamma < \infty$.
If $\lambda = \Omega(\max((\frac{\sqrt{L_{\star}}d}{\mu_{\star}})^{\frac{1}{1-2\alpha}}, \gamma^{-\frac{1}{1-3\alpha}}))$, then the right-hand side will eventually, in $d$, be greater than some positive universal constant. Hence we can drop the error term when lower bounding the gap.

For $J$ being a radial bump function and $c = \frac{1}{2}$,
\begin{align*}
    \frac{2-c^2}{1- c^2} \sup_{x \in \mathbb{R}^d} \lVert \nabla J(x)\rVert^2 \leq \frac{7}{3}\lambda^{2\alpha}.
\end{align*}

Again from the fact that $H_g$ is a harmonic oscillator and Theorem \ref{thm:yau_strongly_convex}, we have that $e_2(\lambda) - e_1(\lambda) \geq \lambda\sqrt{\mu}$.
So we have
\begin{align*}
        E_2(\lambda) - E_1(\lambda)
        &\geq \sqrt{\mu}\lambda - \frac{7}{3}\lambda^{2\alpha} - \frac{\gamma}{3}\lambda^{2-3\alpha}.
\end{align*}
\end{proof}

As a consequence of the above result, we prove a dimension-dependence version of \cite[Theorem 2.3]{simon1983semiclassical}. For sufficiently large $\lambda$, we can ensure that the probability of observing $x \in \mathcal{B}_2(x^{\star}, \epsilon)$ when measuring the ground state of $-\Delta + \lambda^2f$  is close to that of a Gaussian centered at $x^{\star}$ with covariance $(\lambda \sqrt{\nabla^2f(x^{\star})})^{-1}$.

\begin{restatable}[Sub-Gaussian-like Localization of the Ground State]{corollary}{approxGaussianCor}
\label{cor:approx_gaussian}
   Suppose $f \in C^3(\mathcal{X})$, $x^{\star}$ is the unique global minimizer, and
\begin{align*}
    \mu_{\star} I \preceq \nabla^2f(x^{\star}) \preceq L_{\star} I.
\end{align*}
 Define $\gamma = \sup_{\mathcal{B}_2(x^{\star}, \lambda^{-2/5})} \lVert \nabla^3 f(x)\rVert_{\text{op}}$, for some $\alpha >0$. Suppose $y^{\star}$ is the closest local minimizer to $x^{\star}$. Suppose $J$ is a smooth bump radial function centered at $x^{\star}$ and support being an $\ell_2$ ball of radius $\lambda^{-2/5}$. Let $\psi_1$ be a Gaussian wave function centered $x^{\star}$ with covariance matrix $(\lambda\sqrt{\nabla^2f(x^{\star})})^{-1}$, and $\phi_\lambda$ the ground state of the operator
 \begin{align*}
     H(\lambda) = -\Delta + \lambda^2 f(x),
 \end{align*}
 Then for
    \begin{align*}
    \lambda = \Omega\left(\max\left(\left(\frac{\sqrt{L_{\star}}d}{\mu_{\star}}\right)^{5}, \gamma^{6}, \frac{1}{\lVert y^{\star} - x^{\star}\rVert}\right)\right),
\end{align*}
we have
\begin{align*}
    \mathbb{P}_{\lvert\phi_{\lambda}\rvert^2}[\lVert X - x^{\star} \rVert < \epsilon] >  \mathbb{P}_{\lvert \psi_1\rvert^2}[\lVert X - x^{\star} \rVert < \epsilon] - \frac{2}{5}
\end{align*}
\end{restatable}
The proof is in Appendix Section \ref{sec:proof_of_cor_approx_gaussian}.

One issue with the above is that it requires $\lambda$ to be a large polynomial in $d$. We can obtain a significantly lower dependence on $d$, when only requiring closeness in function value. This result uses techniques similar to Theorem \ref{thm:loc_spec_compare} but bypasses explicitly lower bounding the gap. The key proof technique is again IMS Localization (Lemma \ref{lem:ims}). The below result is the one that we use for determining a sufficient $\lambda_{\max}$ for bounding all quantum runtime.
\begin{Lemma}[Concentration of Potential Energy]
\label{lem:hypercontractive-sufficient-lambda}
Let $H(\lambda) = -\Delta + \lambda^2 f$ be a Schrödinger operator on $\mathcal{X}\subseteq \mathbb{R}^d$ with ground state eigen-pair $(\psi_{1,\lambda}, E_{1}(\lambda))$. Assume that $f$ is a twice differentiable function with a locally bounded Hessian around its unique global minimum $x^{\star} \in \mathcal{X}$, i.e., $\|\nabla^2f(x)\|\leq L$ for $\|x-x^{\star}\|\leq r$. Then, for all $\lambda \geq \lambda_{\text{max}} = \mathcal{O}\left(\max\left(\frac{d\sqrt{L}}{\epsilon}, \frac{1}{\sqrt{\epsilon}r}\right)\right)$, satisfies
\[
\mathbb{P}_{|\psi_\lambda|^2}\left[f(X)-f(x^{\star})\geq \epsilon\right]\leq \frac{1}{5}.
\]
\end{Lemma}
\begin{proof}
By conjugating the Schr\"odinger operator with $\psi_\lambda$, 
$$
E_1({\lambda}) = \langle \psi_{1,\lambda}, H({\lambda}) \psi_{1,\lambda} \rangle = -\langle \psi_{1,\lambda}, \Delta \psi_{1,\lambda} \rangle + \lambda^2 \int_{\mathcal{X}} |\psi_{1,\lambda}(x)|^2 f(x) \, \mathrm{d}x,
$$
Since the Laplacian term is non-negative ($\langle \psi_{1,\lambda}, -\Delta \psi_{1,\lambda} \rangle = \int_{\mathcal{X}}\mathrm{d}x\|\nabla \psi_{1,\lambda}(x)\|^2 \geq 0$), we have
$$
\frac{E_1({\lambda})}{\lambda^2} \geq \mathbb{E}_{m_\lambda}[f(X)],
$$
where $m_{\lambda} = \psi_{1,\lambda}^2$.
Without loss of generality, set $f(x^{\star}) = 0$. Then,
$$
\frac{E_1({\lambda})}{\lambda^2} \geq \mathbb{E}_{m_\lambda}[f(X) - f(x^{\star})].
$$

To upper bound the ground state energy, consider a test function $\widetilde{\phi}_{\lambda} =J\phi_\lambda$, where $J$ is a smooth bump function (see the definition in Section~\ref{sec:local_spec_compare}
) supported on a ball of radius $r$ centered at $x^{\star}$, and $\phi_\lambda$ is the ground state of the following QHO Hamiltonian
$$
H_L(\lambda) = -\Delta + \frac{\lambda^2 L}{2} \|x-x^{\star}\|^2.
$$
By the variational principle and IMS localization formula~\eqref{eqn:ims_identity}, in a similar manner to Equation \eqref{eqn:e1_upper_bound}:
\begin{align*}
 E_1({\lambda}) &\leq \frac{\langle \widetilde{\phi}_\lambda| H_L| \widetilde{\phi}_\lambda \rangle}{\|\widetilde{\phi}_\lambda\|^2} + \frac{\langle \widetilde{\phi}_\lambda| (H-H_L)|\widetilde{\phi}_\lambda \rangle}{\|\widetilde{\phi}_\lambda\|^2}
 \leq \frac{\langle \widetilde{\phi}_\lambda| H_L| \widetilde{\phi}_\lambda \rangle}{\|\widetilde{\phi}_\lambda\|^2} \\
 &\leq  \frac{\langle \phi_\lambda|J^2 H_L|\phi_\lambda\rangle}{2\|\widetilde{\phi}_\lambda\|^2}+ \frac{\langle \phi|H_LJ^2|\phi_\lambda \rangle}{2\|\widetilde{\phi}_\lambda\|^2}+ \|\nabla J\|^2\\
 &\leq  \sqrt{2}\lambda dL^{1/2}+\|\nabla J\|^2
\end{align*}
where the second inequality follows from the local Hessian bound on the support of $J$. 
\begin{align*}
E_1({\lambda}) &\leq  \langle \widetilde{\phi}_\lambda| H_L| \widetilde{\phi}_\lambda \rangle \leq \sqrt{2}\lambda dL^{1/2}+\|\nabla J\|^2
\end{align*}
Therefore,
$$
\mathbb{E}_{m_\lambda}[f(x) - f(x^{\star})] \leq \frac{\sqrt{2}dL^{1/2} }{\lambda} +\frac{1}{\lambda^2 r^2}
$$
Choosing $\lambda \geq \max\left(\frac{8dL^{1/2}}{\epsilon}, \frac{4 }{\sqrt{\epsilon}r}\right) $, we obtain
$
\mathbb{E}_{m_\lambda}[f(x) - f(x^{\star})] \leq  \epsilon / 5
$. Then, by Markov's inequality,
$$
\mathbb{P}_{m_\lambda}\left[f(X)-f(x^{\star}) \geq \epsilon\right]  \leq \frac{1}{5}.
$$
This completes the proof.

\end{proof}

\subsubsection{Block-Separable Potentials}
\label{sec:semiclassical_block_separable}

In this Section we derive the spectral analysis results used in Section \ref{sec:quantum_runtime_separable} for block-separable functions. We start with a result that lower bounds the spectral gap for all $\lambda$ above some constant.
\corgapblocksep*
\begin{proof}

Due to the separable nature of the potential and the invariance of the Laplacian under rotations, it is apparent that the spectral gap, $\delta^{(Q)}(\lambda) $, of 
\begin{align*}
    H_f = -\Delta + \lambda^2 f(x)
\end{align*}
is lower bounded by the minimum of the gaps
\begin{align*}
    H = -\Delta+ \lambda^2 g_i(\hat{x}),
\end{align*}
over wave functions on $\mathbb{R}^{d_i}$.
We simply note the corresponding $g_i(\widehat{x})$ with minimum gap as $g(x) : \mathbb{R}^{d_i} \rightarrow \mathbb{R}$.

We can simply apply Theorem \ref{thm:semiclassical_local_taylor_version}  to the  $H$ with $\alpha = 2/5$, and note that the $d_i = \mathcal{O}_d(1)$. Hence, there is no need for $\lambda_{\star}$ to depend on $d$. The first condition in the Corollary statement comes from Equation \eqref{eqn:positiveness_condition} in the proof of Theorem \ref{thm:semiclassical_local_taylor_version}.
\end{proof}

To show that the gap is large for $\lambda < \lambda_{\star}$, we make use of the following Theorem.
\begin{theorem}[Theorem 3.2 \cite{yau2009gapeigenvaluesschrodingeroperator}]
\label{thm:yau_nonconvex}
Let $\Xcal$ be a convex domain so  the Hessian of $\ground{f}$ is greater than $-a$. Then the spectral gap $\delta^{(Q)}$ of the Dirichlet operator $H(\lambda) = -\Delta + \lambda^2 f$ satisfies:
\begin{align*}
    \delta^{(Q)} \geq 2\textup{diam}(\Xcal)^{-2}\exp\left(-a\cdot\textup{diam}(\Xcal)^2\right).
\end{align*}
\end{theorem}
When the above is combined with Theorem \ref{thm:yau_nonconvex},
we obtain a gap lower bound for all $\lambda$.
\thmgapblocksep*
\begin{proof}

First we start by applying Corollary \ref{cor:separable_gap}, which given the rest of the Corollary assumptions there does clearly exist a nonzero, constant $\lambda_{\star}$ such that all of the inequalities are satisfied. Hence the gap is $c'\lambda$ for $\lambda > \lambda_{\star}$ and some constant $c' > 0$. 

For $\lambda \leq \lambda_{\star}$, we again look at the gap of the operator
\begin{align*}
    H_i = -\Delta + \lambda^2g_i(\hat{x}),
\end{align*}
for $g_i$ as in Corollary \ref{cor:separable_gap}.
Note that by assumption the domain of the functions on which $H_i$ acts has dimension $d_i = \mathcal{O}_d(1)$, and so the diameter is $\mathcal{O}(1)$. Again, by assumption the minimum eigenvalue of $\ground{g_i}$ for $H_i$ must also be independent of $d$. Hence, Theorem \ref{thm:yau_nonconvex} implies that for $\lambda \leq \lambda_{\star} = \mathcal{O}_d(1)$, $\delta^{(Q)}(\lambda) = \Omega_d(1)$.
Hence, the result.
\end{proof}

\subsubsection{Potentials with Quadratic Growth}

\label{sec:semiclassical_quad_env}

In this section, we investigate what improvements on the dependence of $d$ can be made in Theorem \ref{thm:semiclassical_local_taylor_version} when $f$ satisfies quadratic growth-like assumptions. The first of these results (Theorem \ref{thm:quad_semiclassical_local_taylor_version}) also makes use of a local quadratic upper bound and reduces the $\lambda$ dependence from $d^5$ to just $d$. 

\begin{theorem}
\label{thm:quad_semiclassical_local_taylor_version}
Suppose $f \in C^3(\mathcal{X})$ satisfies
\begin{align*}
&\frac{c_f}{2} \lVert x - x^{\star} \rVert^2 \leq f(x) - f(x^{\star}), \forall x \in \Xcal\\
&f(x) - f(x^{\star}) \leq \frac{C_f}{2}\lVert x - x^{\star} \rVert^2, \forall x \in \mathcal{B}\left(x^{\star}, \frac{\sqrt{ C_f d}}{\sqrt{c_f\lambda}}\right),
\end{align*}
where $x^{\star}$ is the global minimizer.
If \begin{align*}
 \lambda = \Omega\left(\frac{C_f d}{c_f\sqrt{c_f}}\right),
\end{align*}
then we have that the spectral gap of the operator
$$H = -\Delta + \lambda^2 f(x)$$
satisfies
\begin{align*}
\delta^{(Q)}(\lambda) = \Omega\left(\sqrt{c_f}\lambda\right).
\end{align*}
\end{theorem}

\begin{proof}
Like usual, wlog, we can assume $x^{\star} = f(x^{\star}) = 0$.
The proof follows that of Theorem \ref{thm:semiclassical_local_taylor_version} closely. We consider comparing the following two operators and apply Theorem \ref{thm:loc_spec_compare}:
\begin{align*}
&H_g = -\Delta + \frac{c_f}{2}\lVert x\rVert^2\\
&H_f = -\Delta + \lambda^2 f(x),
\end{align*}
with $J$ being a radial smooth bump function with support on $\mathcal{B}\left(0, \frac{\sqrt{\beta d}} {\sqrt{\lambda}}\right)$,  $c = \frac{1}{2}$, and some $\beta > 0$.
From the poof of Theorem \ref{thm:semiclassical_local_taylor_version}, we have
\begin{align*}
    e_2(\lambda) \leq \frac{3}{\sqrt{2}}\lambda d\sqrt{C_f},
\end{align*}
and thus applying the supposed quadratic lower bound:
\begin{align*}
    J_0H_fJ_0 - e_2(\lambda) \geq  \frac{c_f}{2} \beta \lambda d - \frac{3}{\sqrt{2}}\lambda d\sqrt{C_f}.
\end{align*}

Note that by Theorem \ref{thm:yau_strongly_convex},
\begin{align*}
    e_2(\lambda) - e_1(\lambda) \geq  \sqrt{c_f}\lambda + \frac{1}{4\textup{diam}(\Xcal)^2}.
\end{align*}

From  $f -g \leq f \leq C_f \lVert x \rVert^2$, Theorem \ref{thm:loc_spec_compare}, and the above
\begin{align*}
    E_2(\lambda) - E_1(\lambda) \geq  \sqrt{c_f}\lambda + \frac{c_f}{2} \beta \lambda d - \frac{3}{\sqrt{2}}\lambda d\sqrt{C_f} - \frac{7}{3}\frac{\lambda}{\beta d} - \frac{\langle J \psi_1 | \lambda^2 C_f \lVert x - x^{\star}\rVert^2 | J\psi_1\rangle}{\lVert J \psi_1 \rVert^2}
\end{align*}

\begin{align*}
    \langle J \psi_1 | \lambda^2 C_f \lVert x - x^{\star}\rVert^2 | J\psi_1\rangle \leq \frac{1}{4}\lambda d C_f.
\end{align*}
Since $\psi_1$ is a Gaussian centered around $x^{\star}$ with covariance $(\lambda \sqrt{c_f})^{-1} I$. For $\lambda =\Omega(d/\sqrt{c_f})$, we can ensure that $\lVert J \psi_1 \rVert^2 \geq 1/2$.

Hence
\begin{align*}
E_2(\lambda) - E_1(\lambda) \geq \sqrt{c_f}\lambda + \lambda d\left( \frac{c_f}{2} \beta  - \frac{3}{\sqrt{2}}\sqrt{C_f} - \frac{7}{3}\frac{1}{\beta d^{2}} -  \frac{1}{2}\ C_f \right).
\end{align*}
If $\beta  = \Theta\left(\frac{C_f}{c_f}\right)$, then we can make the second term positive. Hence 
\begin{align*}
    E_2(\lambda) - E_1(\lambda) = \Omega\left(\sqrt{c_f}\lambda + \textup{diam}(\Xcal)^{-2}\right).
\end{align*}
\end{proof}
One can also use the above to obtain an improved version of Corollary \ref{cor:approx_gaussian} for quadratically enveloped functions.

However, as our goal is to determine the runtime of RsAA, the above is still not sufficient to cover the whole range of $\lambda$. To get around this issue, we derive an alternative result that instead varies the radius at which the perturbation is applied. This bound applies for  $\lambda$ all the way down to a constant. The tradeoff is now that $f$ and $g$ only differ outside of a large ball of radius $\Omega(\sqrt{d})$. Additionally, the quadratic upper bound on $f$ is now global.

\pertQuatEnv*

The proof makes use of Agmon's Theorem and is left to the appendix (Section \ref{sec:appendix_semi_classical}). The above result was used in Section \ref{sec:removing_separability}.

Lastly, to see if we can push Theorem \ref{thm:loc_spec_compare} all the way to a bound that holds uniformly in $\lambda$, we include the following result. Specifically, if the perturbation turns out to be positive, then we can drop the localization error component and remove  the constraints on $\lambda$. While we do not use this result, the purpose of it to show that the extended semiclassical approaches of this section can achieve similar results to the next section, when the perturbation is positive. The proof of the result follows simply from Equation \eqref{eqn:positive_case}.
\begin{corollary}
\label{cor:positive_pert}
Suppose $f \in C^3(\Xcal)$ has a unique global minimizer and that for a $c_f$-strongly convex function $g$:
\begin{align*}
   0 \leq  f(x) -g(x)  < \frac{\sqrt{c_f}}{d}, \quad \forall x \in \Xcal
\end{align*} 
then the gap of 
\begin{align*}
   H_f(\lambda) = -\Delta + \lambda^2 f,
\end{align*}
is $\Omega(1)$, uniformly in $\lambda$.
\end{corollary}

\subsection{Intrinsic Hypercontractivity}
\label{sec:hypercontractivity}
To establish the spectral gap of certain class of Schr\"odinger operators and further regularity behavior of quantum ground states, we analyze a property of the associated semi-group called \emph{hypercontractivity}. In this section, we give a brief overview of hypercontractivity and other related functional inequalities of Markov semi-groups, which will be utilized for characterizing the minimum spectral gap along the trajectory of the quantum adiabatic algorithm. For a broader exposition of Markov semi-groups and hypercontractivity, we refer readers to \cite{vonhandel}.
\begin{definition}
    A semi-group $\mathcal{M} = (\mathcal{L}(P), \mu)$ is hypercontractive if for all $f \in L^2(\mu)$
\begin{equation}
\label{eq:hypercontractivity}
    \|P_tf\|_{L^q(\mu)}\leq \|f\|_{L^p(\mu)}
\end{equation}
for some $q>p$. 
\end{definition}

The notion of hypercontractivity for Markov semi-groups implies a spectral gap of the generator and in fact it provides an even stronger link between functional inequalities and the long-time behavior of stochastic processes. %
For instance, although the contraction inequality~\eqref{eq:hypercontractivity} holds for every semi-group when $p=q$ due to Jensen's inequality, the hypercontractivity condition $q>p$ does not hold for every semi-group. In fact, hypercontractivity implies additional regularity and concentration properties of the stationary measure such as sub-Gaussian tails for the stationary measure. The following theorem is due to Gross~\cite{gross_1975} and relates hypercontractivity to the well-known Log-Sobolev inequality, which is used widely to understand the mixing properties of stochastic processes.

\begin{theorem}
    \label{thm:hypercontractivity-lsi}
    Let $\mathcal{M}=(\mathcal{L}(P),\mu)$ be a reversible Markov semi-group with stationary measure $\mu$. Then the hypercontractivity condition is equivalent to the following. For every $f$, there exists $c>0$ such that
    \begin{equation}
        \label{eq:LSI}
        \mathrm{Ent}_\mu[f]\leq \frac{1}{c}\mathcal{E}(f,f)
    \end{equation}
    where 
    $
      \text{Ent}_{\mu}[f] = \mathbb{E}_\mu[f\log f] -\mathbb{E}_\mu[f]\log \mathbb{E}_\mu[f].
    $
\end{theorem}
\begin{definition}
   Log-Sobolev (LS) constant $\omega$ of a semi-group  is the largest $c$ that satisfies inequality~\eqref{eq:LSI} for every $f$.
\end{definition}

In the setting of Schr\"odinger operators, the semi-group that generates the associated Langevin diffusion equation exhibits hypercontractivity under suitable assumptions on the potential $V$. In this case, the Schr\"odinger operator is called \emph{intrinsically hypercontractive} by Davies and Simon~\cite{DAVIES1984335}, as the unitary quantum evolution itself cannot be contractive. Specifically, a Schr\"odinger operator is intrinsically hypercontractive if can be transformed, via an isometry (e.g. Theorem \ref{thm:ground-state-transform}), into a hypercontractive diffusion process.
This property reflects strong regularization of the wavefunctions and the geometry of the low-lying eigenvalues. In fact, this property has been widely used to understand mixing behavior of Langevin diffusion for sampling tasks \cite{chewi2025logconcave} in the machine learning context. More generally, as a result of Theorem \ref{thm:ground-state-transform}, the spectral results in the classical literature on diffusion processes hold for Schr\"odinger operators provided that the underlying assumptions hold for the ground-state potential $\log\psi_1^2$. 

For our purpose, it is important to link the spectral gap between the low lying eigenvalues and the hypercontractive. The following well-known result shows that hypercontractivity implies a spectral gap lower bound.
\begin{theorem}
\label{thm:sobolev-vs-poincare}
    Assume that $\mu$ has Log-Sobolev constant $\omega$. Then $\mu$ satisfies the Poincar\'e inequality for all $f$:
    \begin{equation}
    \label{eq:Poincare}
        \mathrm{Var}_\mu[f]\leq \frac{2}{c} \mathcal{E}(f,f)
    \end{equation}
    with constant $c\geq \omega$.
\end{theorem}
\begin{definition}
   The Poincar\'e constant of a semi-group $\delta$ is the largest $c$ that satisfies inequality~\eqref{eq:Poincare} for every $f$.
\end{definition}

For reversible semi-groups, the Poincar\'e constant is equal to the spectral gap. Hence, the spectral gap of a Schr\"odinger operator can be lower bounded by the Log-Sobolev constant (or the hypercontractivity constant by Theorem ~\ref{thm:hypercontractivity-lsi}) of the associated semi-group.

Although it is harder to directly analyze the properties of $\psi_1$, it is possible to show that the Langevin diffusion with the drift $\nabla \log\psi_1^2$ is hypercontractive based on the properties of the potential.
In this section, we show that the quantum ground state of some of the Schr\"odinger operators presented in this paper are in fact hypercontractive, and hence the spectral gap in theorem can be lower bounded. To this end, we use the recent result on the intrinsic hypercontractivity of Schr\"odinger operators by Gross~\cite{gross2025invariance}.

\begin{theorem}[Adapted from Theorem 2.2 in~\cite{gross2025invariance}]
\label{thm:intrinsic-hypercontractivity}
    Let $H = -\Delta + V$ be a self-adjoint operator with ground state $\psi_1$. Assume that $m = \psi_1^2$ satisfies log-Sobolev inequality, i.e.,
    \[
    \text{Ent}_m[f^2]\leq \frac{1}{\omega}\langle \nabla u, \nabla u \rangle_m 
    \]
    and the following holds,
    \[
    \|e^{W}\|_{\kappa}<\infty \quad \text{and}\quad  \|e^{-W}\|_{\nu}<\infty \text{ for some } \kappa>0\quad \text{and}\quad \nu>\frac{1}{\omega}
    \]
    where $\|\cdot\|_p$ is the $L^p$ norm under measure $m$. Then, the Schr\"odinger operator 
    \[
    \tilde{H} = -\Delta + V+W
    \]
    satisfies log-Sobolev inequality with constant $\tilde{\omega}\geq \alpha \omega M$ where 
    \[
    M =  (\|e^{W}\|_{\kappa}^{\kappa} \cdot \|e^{-W}\|_{\nu}^{\nu})^{\beta}
    \]
    and $\alpha, \beta$ are dimension independent constants. 
\end{theorem}
\begin{remark}
    We note that the form of the Theorem~\ref{thm:intrinsic-hypercontractivity} is slightly different than the version presented in \cite{gross2025invariance}. In the original work the constants $\alpha,\beta$ can depend on $\frac{1}{\omega}$ which can have dimension dependence and $M$ is given in terms of $\|e^W\|_\kappa  \|e^{-W}\|_\nu$. However, we show in Section~\ref{sec:constants-appendix} that the $\omega$ dependency of these constants can be pulled out and we can express the final result in terms of $\|e^W\|_\kappa^{\kappa} \|e^{-W}\|_\nu^{\nu}$ as long as $\nu\geq \frac{1}{\omega}$ and redefine $\alpha,\beta$ to be dimensionless even if $\omega$ has dimension dependency.
\end{remark}

\begin{corollary}
Let $H = -\Delta + V$ be a self-adjoint operator with ground state $\psi_1$. If the ground state measure $m = \psi_1^2$ satisfies log-Sobolev inequality with respect to $m$ then the operator $\tilde{H} = -\Delta + V+W$ has spectral gap 
\[
\delta \geq \alpha\omega M
\]
\end{corollary}
\begin{proof}
    This directly follows from~\ref{thm:intrinsic-hypercontractivity} and~\ref{thm:sobolev-vs-poincare} and the fact that Poincar\'e constant is equal to spectral gap for reversible processes.
\end{proof}
The next corollary shows that hypercontractivity implies sub-Gaussian tail properties. 
\begin{corollary}
\label{cor:lsi-subgaussian}
    Let $H = -\Delta + V$ be a self-adjoint operator with ground state $\psi_1$. If the ground state measure $m = \psi_1^2$ satisfies log-Sobolev inequality with constant $\omega$ with respect to $m$, then for any $1$-Lipschitz function $g$,
    \[
    \mathbb{P}\left[g(X)-\mathbb{E}[g]\geq t\right]\leq \exp\left(-\frac{\omega 
    t^2}{2}\right).
    \]
\end{corollary}
The proof of Corollary~\ref{cor:lsi-subgaussian} directly follows from the Herbst's argument (Lemma 3.13 in \cite{vonhandel}) which we do not repeat here. Although we don't directly used Corollary~\ref{cor:lsi-subgaussian}, such tail properties in general might be important in general to analyze the performance of RsAA. For example, to characterize the $\lambda_{\mathrm{max}}$, the authors in~\cite{leng2023quantum} had to show that the ground state has sub-gaussian tails using Agmon's theorem under certain growth conditions. Since the authors only considered coordinate separable potentials, their analysis is conducted in single dimension. However, for more general $d$ dimensional potentials, it might be challenging to obtain such tail properties.

Hypercontractivity on the other hand directly implies both spectral gap and sub-Gaussian tails. Therefore, we believe hypercontractivity might be an important tool for establishing the runtime of the adiabatic algorithm and can provide more insight for the theoretical analysis of these algorithms than the perturbation theory alone. In fact, in the discrete case, the works by~\cite{Hastings2018shortpathquantum, 10.1145/3564246.3585203, chakrabarti2024generalizedshortpathalgorithms} showed that the ground state of the adiabatic Hamiltonian can be prepared efficiently as long as the mixer satisfies some entropic inequalities such as log-Sobolev inequality and spectral density conditions. These assumptions are actually very similar to the conditions that the entire Hamiltonian is intrinsically hypercontractive, although these works have not made explicit connections to intrinsic hypercontractivity of the associated semi-group and it is not immediately clear what the corresponding semi-group corresponds to in these cases.

\subsubsection{Perturbation of Strongly Convex Potentials}

We start with a simple result regarding the intrinsic hypercontractivity of Schr\"odinger operators with strongly-convex potentials. This follows from Theorem \ref{thm:yau_strongly_convex} and a few other known results. It is proven in the Appendix (Section \ref{sec:proof-of-dirichlet-log-sobolev-strongly-convex}).
\begin{theorem}
\label{thm:dirichlet-log-sobolev-strongly-convex}
    Suppose $f$ is $\mu$-strongly convex over a bounded convex domain $\Xcal$, then the operator $H(\lambda) = -\Delta + \lambda^2f$ is intrinsically hypercontractive with log-Sobolev constant $\omega$ satisfying
    \begin{align*}
        \omega = \Omega\left(\sqrt{2\mu}\lambda  + \textup{diam}(\Xcal)^{-2}\right).
    \end{align*}
\end{theorem}
The above will be used to derive the following perturbation result.
\perturbHyperContract*
\begin{proof}
    Without loss of generality, let $x^{\star}=0$. Instead of $\mathcal{X}$, we can bound the spectral gap on $\mathbb{R}^d$ as imposing Dirichlet condition only increases the spectral gap. We treat $\lambda^2 h$ as the base potential with the ground state $\psi_1$ and treat $\lambda^2 g$ as a perturbation. As $\lambda^2f$ is $\mu^2 \lambda^2$ strongly convex, the ground state $\psi_1$ satisfies log-Sobolev inequality with constant at least $\omega_0 =  \Omega\left(\sqrt{2\mu}\lambda  + \textup{diam}(\Xcal)^{-2}\right)$  by Theorem \ref{thm:dirichlet-log-sobolev-strongly-convex} . To be able to prove the spectral gap, we need to show that $M =  \|e^{-\lambda^2 g}\|_\nu^{\nu} \|e^{\lambda^2 g}\|_\kappa^{\kappa}$ from the previous section is lower bounded by a constant so that the ground state of $H$ satisfies an log-Sobolev inequality with a dimension independent constant $\omega$. The first term in the product is
\begin{align*}
    \|e^{-\lambda^2 g}\|_\nu^\nu &= \int_{x\in \mathbb{R}^d} \psi_1^2(x) \exp(-\nu \lambda^2 g(x))\mathrm{d}x.
\end{align*}
We can take $\nu = \frac{3}{\lambda \mu} > \frac{2}{\omega_0}$,
\begin{align*}
     \|e^{-\lambda^2 g}\|_\nu^\nu &=\int_{x\in \mathbb{R}^d} \psi_1^2(x) \exp(-3\lambda \mu^{-1} g(x))\mathrm{d}x
\end{align*}
Let $\tilde{H} = - \Delta + \lambda^2 C_f\|x\|^{2k}-\sqrt{C_f}\lambda k\|x\|^{k-1}$. Let $\tilde{\psi}_1$ be the ground state of $\tilde{H}$. By WKB equation, we can exactly verify that $\tilde{\psi}_1 = A \exp(-\sqrt{C_f}\lambda  \|x\|^{k+1})$ by where $A$ is the normalization constant which can be computed exactly,
\begin{align*}
    A = \left(\frac{2\pi^{d/2}}{\Gamma(d/2)}\frac{1}{k+1}(\sqrt{C_f}\lambda)^{-d/(k+1)}\Gamma(d/(k+1))\right)^{-1}
\end{align*}

We can argue that $\|e^{-\lambda^2 g}\|_\nu^\nu$ with respect to $\psi_1^2$ is smaller than $\tilde{\psi}_1^2$ and compute the expectation with respect to the measure $\tilde{\psi}_1^2$ instead.  We have

\begin{align*}
         \|e^{-\lambda^2 g}\|_\nu^\nu &\leq  A\int\exp\left(-\lambda \sqrt{C_f} \|x\|^{k+1}+\frac{3\lambda C_g(k+1)  \|x\|^{k+1}}{d \mu}\right)\mathrm{d}x\\
         & \leq AS_{d-1}\int_{\mathbb{R}^d} r^{d-1}\exp \left(-\lambda \sqrt{C_f} r^{k+1}\left(1+\frac{3C_g(k+1)}{d \sqrt{C_f}}\right) \right)\mathrm{d}r\\
        &\leq \exp \left(\frac{3C_g}{C_f^{1/2}}\right)
\end{align*}
where we use change of variables in the last step. We can similarly bound $\|e^{\lambda^2g}\|_{\kappa}^{\kappa}$ by setting $\kappa = \nu$ and prove the inequality above holds for $\|e^{\lambda^2g}\|_{\kappa}^{\kappa}$ in this case. Then, from Theorem~\ref{thm:intrinsic-hypercontractivity}, the log-Sobolev constant of $H$ (with Dirichlet conditions imposed) is $\Omega(\lambda + \text{diam}(\Xcal)^{-2})$. 
Since spectral gap is larger than the log-Sobolev constant, this concludes the proof.

\end{proof}

\section{Additional Analysis of Classical Algorithms}
\label{sec:classical-algorithms}
\subsection{Lower Bound for Langevin-based Algorithms}
\label{sec:sgd_lower_bound}

Here we aim to construct a family of separable objective functions where it takes at least $e^{\Omega (d )}$ iterations for algorithms based on Langevin diffusion, e.g. SGD, to find a point close to the global minimum with high probability. In this section, we consider a Langevin-based algorithm with a fixed learning rate; decaying learning rate without any structural awareness should not change the asymptotics of the convergence time~\cite{shi2023learning}. In the next sections, we will consider adaptive version of the diffusion process where the learning rate and the behavior of the algorithm takes the structure into account.
To establish the lower bound, we use the continuous-time dynamics ~\eqref{eq:learning-rate-sde} corresponding to an actual discretized algorithm, e.g. SGD given as~\eqref{eq:sgld-update}. We note that since actual algorithms are obtained by discretizations of the continuous dynamics, the convergence time of the continuous dynamics is actually larger than the convergence time of the algorithm. Therefore, considering the continuous dynamics is legitimate as discretizations only degrades the convergence. 

\begin{Lemma}
    \label{lem:beta-lower-bound}
    Let $f:\mathbb{R}^d\to \mathbb{R}$ be a continuous function with a unique global minimum $x^{\star}$. Assume that $f(x) - f(x^{\star}) \leq C_f\|x-x^{\star}\|^b$ . Then, the following is satisfied:
    \[
    \mathbb{P}_{\mu_{\beta}}\left[\|X-x^{\star}\|\leq r\right] \leq 2^{-d}\exp(\beta C_f (2r)^b),
    \]
    where $\mu_{\beta}$ is the following distribution
    \[
\mu_{\beta}(x) = \frac{\exp(-\beta f(x))}{Z}.
\]
\end{Lemma}
\begin{proof}
  
For simplicity, assume that $x^{\star} = f(x^{\star})=0$. Then we can lower bound the probability simply as follows.
\begin{align*}
    \mu_{\beta}(\|X\|\leq r) &= \frac{\int_{\|x\|<r} \exp(-\beta f(x)) \mathrm{d}x}{\int_{\mathbb{R}^d} \exp(-\beta f(x)) \mathrm{d}x}\\
    &= \frac{\int_{\|x\|<r} \exp(-\beta f(x)) \mathrm{d}x}{\int_{\|x\|\leq 2r} \exp(-\beta f(x)) \mathrm{d}x +\int_{\|x\|> 2r} \exp(-\beta f(x)) \mathrm{d}x}\\
    &\leq \frac{\exp(-\beta f(x^{\star}))\int_{\|x\|<r}  \mathrm{d}x}{\exp(-\beta L (2r)^b) +\int_{\|x\|\leq 2r} \exp(-\beta f(x)) \mathrm{d}x}\\
    &\leq \frac{r^d \exp(\beta C_f (2r)^b )}{ (2r)^d }
\end{align*}
which concludes the proof. 
\end{proof}
The simple lemma above implies that one needs to set $\beta = \Omega\left(\frac{d}{\epsilon^{b}}\right)$ to be able to sample a point $\epsilon$ close the minimizer $x^{\star}$ from the distribution $\mu_{\beta}$ with at least constant probability. Since the function is also upper bounded by a polynomial with degree $b$, one need to set $\beta = \Omega(d/\epsilon)$ to sample a point $x$ so that $f(x)-f(x^{\star})<\epsilon$ with high probability. 

To obtain a lower bound on the mixing time of the Langevin diffusion that converges to the stationary distribution $\mu_{\beta}$, we use Morse theory to characterize the spectral gap of the diffusion process. Suppose that $f$ is an infinitely differentiable function. Here we only give minimal definitions related to Morse theory, and we refer the reader to~\cite{nier:hal-00002744,Michel_2019,shi2023learning} fore more details on Morse theory.

A function $f$ is called a Morse function if for any critical point of $f$, all eigenvalues of the Hessian are non-zero. For Morse functions, it is possible to show that spectral gap of the Witten Laplacian associated with the Langevin diffusion decays exponentially as the barrier height between the local maxima and local minima increases. To characterize the barrier height, we first need to define a proper labeling of each critical point of $f$. 

Let $\{x^\bullet_i\}$ be the set of all local minima of $f$. The labeling process for the critical points of $f$ can be obtained with the procedure below.
\begin{enumerate}
    \item Let $x^{\star} = x_0$ be the global minimum of $f$. We will consider the level sets of $f$ defined as $f_\nu=\{x\in \mathcal{X}|  f(x)<\nu\}$. We initially set $\nu=\infty$ and consider the level set $f_\infty$. Since $f$ is a continuous function, the set $f_\infty$ forms a connected component. 
    \item Start decreasing $\nu$ such that the number of connected components in the level set $f_v$ increases.
    \item The values of $\nu$ where the number of connected components increase must intersect with the critical points of $f$. If the global minimum of the new connected component is $x^{\bullet}_{i}$, denote the critical point that meets the level set $f_\nu$ by $x_i^{\circ}$. If there are more than one new connected components, each critical point that meets the level set can be matched to the global minimum of each new connected component similarly.
\end{enumerate}
The procedure above assumes that each connected component has a unique global minimum. Finally, we map each index $i$ to a new index $k$ such that 
\[
x^\circ_1-x^\bullet_1\geq x^\circ_2-x^\bullet_2\geq \cdots.
\]
Let $H_f =x^\circ_1-x^\bullet_1$. Then, the following theorem characterizes the eigenvalue of the Witten Laplacian.
\begin{proposition}[Adapted from~\cite{nier:hal-00002744}]
\label{prop:shi_eigenvalue}
Under the assumptions (a) and (b) of Theorem 3.1 of~\cite{nier:hal-00002744}, the spectral gap of the Witten Laplacian associated with the diffusion equation~\eqref{eq:learning-rate-sde} is
   \[
    \delta_s = \frac{s |\hat{\delta}(x_1^{\circ})|}{\pi} \left( \sqrt{\frac{\textrm{det}({\nabla^2f(x^\circ_1))}}{\textrm{det}(\nabla^2 f(x^\bullet_1))}} + c(s)\right) \exp\left(- \frac{2}{s} \left(f(x^\circ_1) - f_1(x^\bullet_1)\right)\right),
\]
where $\hat{\delta}(x_1^{\circ})$ is the unique negative eigenvalue of $\nabla^2 f(x_1^{\circ})$ and $c(s) = \Ocal(s^{1/2}\log(s^{-1}))$.
\end{proposition}

\begin{figure}[!ht]
    \centering
    \includegraphics[width=0.8\linewidth]{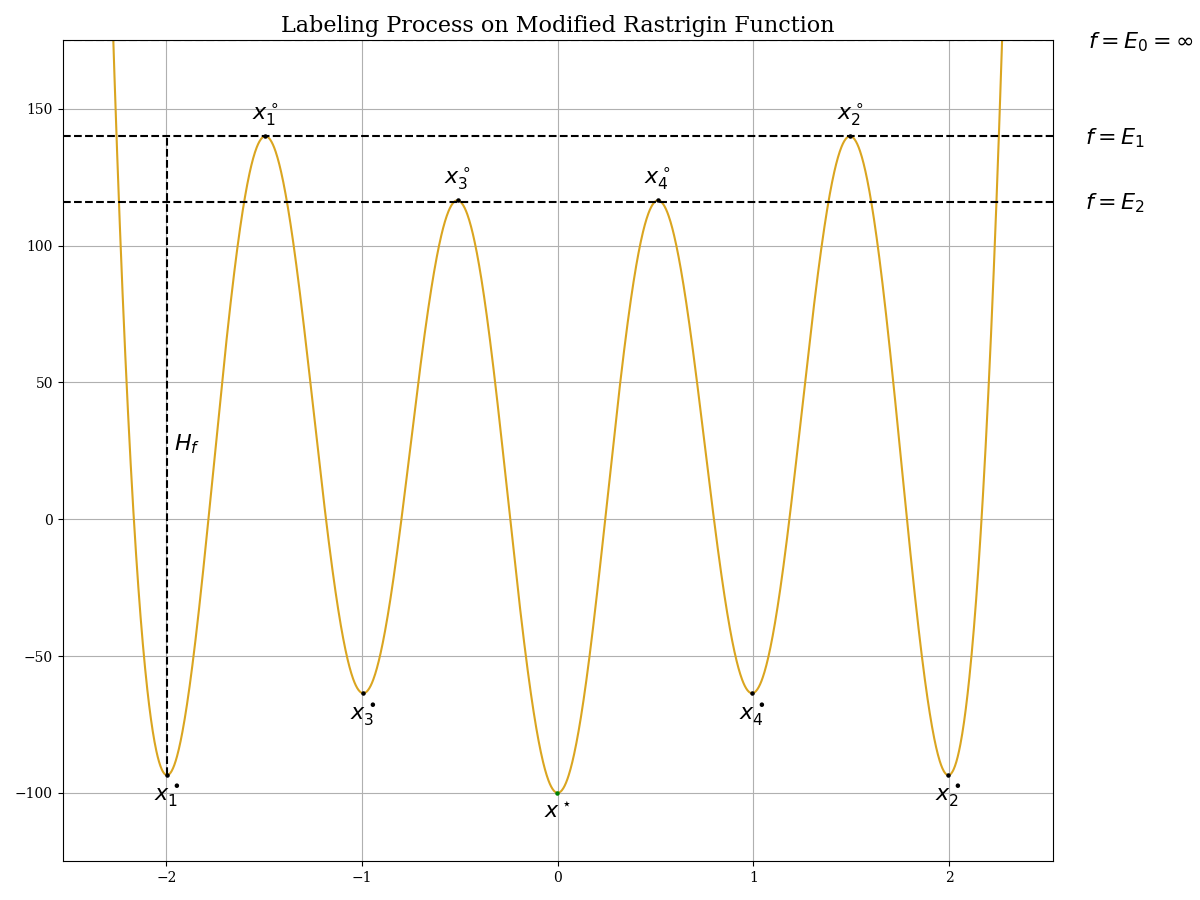}
    \caption{The critical points and level sets of 1 dimensional function $f_i(x_i) = x_i^2((x_i^2 - 1)^2 + 4)((x_i^2 - 4)^2 + 1/8) - 100 \cos(2\pi x_i)$.} 
    \label{fig:1d-sgd}
\end{figure}

Let $f_i(x_i) = x_i^2((x_i^2 - 1)^2 + 4)((x_i^2 - 4)^2 + 1/8) - 100 \cos(2\pi x_i)$. The function is illustrated in Figure~\ref{fig:1d-sgd} with the critical points labeled by the procedure mentioned above. Consider a $d$- dimensional construction,
\[
f(x) = \sum_{i=1}^d f_i(x_i).
\]
Then this function is smaller than a polynomial with degree $10$. Therefore, one needs to set $s=\frac{1}{\beta} =\Ocal(\epsilon/d)$ to reach the global minimum by Lemma~\ref{lem:beta-lower-bound}. By the Proposition~\ref{prop:shi_eigenvalue}, the spectral gap of the SDE with $s = \frac{1}{\beta}$ is then at most
\begin{align*}
    \delta_s =\Ocal\left(\text{poly}(\epsilon/d)\exp\left(\frac{-2dH_f}{\epsilon}\right)\right).
\end{align*}
Since $H_f$ is at least a constant, the total run time is at least
\begin{align*}
    \Omega(\text{poly}(d) \exp(d/\epsilon))
\end{align*}
because the diffusion process is reversible and the mixing time is tightly related to the inverse spectral gap.

\subsection{Optimizing Block-separable Functions With Structure-aware Algorithms}
\label{sec:classical-separable}

In this section we prove Theorems \ref{thm:convex_hone_thm} and \ref{thm:hessian_algo}, which upper bound the runtimes of two structure-aware classical algorithms for optimizing rotated block-separable functions satisfying Assumption \ref{assump:constant_block}. These algorithms both run in time that is polynomial in $d$ and $1/\epsilon$, but the degree of the polynomial in $d$ is $f$ dependent.

\subsubsection{Convexity-Honing Algorithm}
\label{sec:convex_honing}
Here, we analyze the convexity-honing algorithm of Section \ref{sec:classical_algs_separable} that optimizes block-separable functions by first utilizing Langevin dynamics to locate a region of convexity, and then runs the standard gradient descent. The algorithm only runs Langevin dynamics to $\beta = \mathcal{O}(\log(d))$, implying inverse-poly spectral gap. However, the mixing time can actually be an arbitrarily large polynomial. This is because the spectral gap for this case is $\Theta\left(d^{-H_f}\right)$ from Theorem \ref{thm:lang_sgd_gap}, where $H_f$ is a constant depending on $f$.

\convexHoning*
\begin{proof}
    Without loss of generality, consider a block-separable function $f(x) = \sum_{i=1}^{k}g(\hat{x}_i)$, over a bounded region $\Xcal$, where $d_i = \mathcal{O}(1)$ for all $i$. Let $x^{\star}$ be the unique global minimizer, where without loss of generality $x^{\star} = f(x^{\star})  = 0$. Since $d_i = \mathcal{O}(1)$ and $\nabla^2g(0) \succ 0$, there exists a constant $c'_i$ such that $\forall \hat{x}_i \in \mathcal{B}_{\infty}(0, c'_i)$, we have $\nabla^2g(\hat{x}_i) \succ 0$.   Hence by separability, there exists a constant $c$ such that $f$ restricted to $\mathcal{B}_{\infty}(0, c')$ is strongly convex. 

Also since $d_i = \mathcal{O}(1)$, we can find another constant radius $\ell_{\infty}$ ball $\mathcal{B}_{\infty}(0, c''_i)$ such that $g$ on $\mathcal{B}_{\infty}(0, c''_i)$ is strictly smaller than $g$ outside. Define $A_i$ to be the smaller of the above $\ell_{\infty}$ balls.

Again by separability, the Gibbs measure with respect to $f$ decomposes as a product measure over the components $i=1, \dots, k$. To ensure that $x \in \prod_{i=1}^{k} A_i$, we need all $\hat{x}_i \notin A_i$ with at most $1/k  = \Omega(1/d)$ probability. 

Consider a small $\ell_{\infty}$ ball with radius $\gamma$, inside $A_i$ with $\hat{y}_i$ maximizing $g_i$ over this ball. Let $\hat{z}_i$ be a minimizer of $g_i$ in $(A_i)^{c}$. Suppose the domain of $g_{i}$, $\Xcal_i$ lies inside an $\ell_\infty$ ball of radius $\alpha$. Then since $g_i(\hat{y}_i) < g_i(\hat{z}_i)$:
\begin{align*}
    p := \frac{\int_{A_i} e^{-\beta g_i(\hat{x}_i)} d\hat{x}_i}{\int_{(A_i)^c} e^{-\beta g_i(\hat{x}_i)} d\hat{x}_i} \geq \frac{(2\gamma)^{d_i} e^{-\beta g_i(\hat{y}_i)}}{(2\alpha)^{d_i}e^{-\beta g_i(\hat{z}_i)}} = \Omega_d\left(e^{-c\beta}\right),
\end{align*}
for some positive constant $c$. Hence with $\beta = \mathcal{O}(\log(d))$,  by
 $\mathbb{P}[\hat{x}_i \notin A_i] = \frac{1}{1 + p}$, and union bound the probability the Gibbs measure puts outside $\prod_{i=1}^{k} A_i$ is upper bounded by a constant $< 1$. 

The mixing time of Langevin dynamics is bounded by 
\begin{align*}
   \mathcal{O}\left(\frac{\ln\left(\chi(\mu_{\beta}, \pi)/\epsilon\right)}{\delta^{(C)}(\beta)}\right),
\end{align*}
where $\chi$ is the chi-square divergence, $\mu_{\beta}$ is the Gibbs measure at inverse-temperature $\beta$, $\pi$ is the initial distribution, and $\epsilon$ is the total variation distance from $\mu_{\beta}$. From Theorem \ref{thm:lang_sgd_gap} 
\begin{align*}
    \delta^{(C)}(\beta) = \Theta\left(e^{-\beta H_f}\right),
\end{align*}
where due to separability $H_f = \mathcal{O}(1)$. Also due to the product nature of the Gibbs measure over the block of $f$ and Assumption \ref{assump:constant_block}, we have $\chi(\mu_{\beta}, \pi) = \mathcal{O}\left(e^{d}\right)$. Hence running Langevin dynamics for $\mathcal{O}\left(d^{H_f+1}\ln(1/\epsilon)\right)$ time suffices to reach a distribution $\epsilon$ close to the stationary distribution in total variation distance. Note that the above does not account for discretization, which only adds a polynomial overhead \cite[Theorem 6]{vempala2022rapidconvergenceunadjustedlangevin}.

Thus if we restart the Langevin dynamics $\mathcal{O}(\log(1/\delta))$ times, we will find a point in the convex region $\prod_{i=1}^{k} A_i$ with $1- \delta$ probability. Once in the convex region, we run gradient descent from the sampled point. Since the region is strongly convex, gradient descent will converge the global minimizer in $\mathcal{O}(\kappa \log(1/\epsilon))$ queries, where $\kappa$ is the condition number of the Hessian in the convex region. By repeating this whole process and outputting the minimum we can ensure that we find an $\epsilon$-approximate minimizer to $f$ with probability $1- \delta$.
\end{proof}

\subsubsection{Hessian Algorithm}

\label{sec:hessian_algorithm}

In this subsection, we analyze the Hessian algorithm mentioned in Section \ref{sec:classical_algs_separable}. If $f$ is a rotated block-separable function, then there exists $U \in \text{SO}(d)$ such that $f(Ux) = g(x) = \sum_{i=1}^{k}g(\hat{x}_i)$ is block-separable. It can be verified by a simple calculation that the Hessian of $f$ at ${x}$ then takes the form
\begin{align*}
    \nabla^2 f({x}) = U \nabla^2 \left(g(U^{\mathsf{T}}x)\right) U^{\mathsf{T}}.
\end{align*}
Given first-order access to $f$, we can compute the Hessian at a given point with $\mathcal{O}\left(d\right)$ queries to the first-order oracle, and $\mathcal{O}\left(d^2\right)$ total work. Hence, if we compute the Hessian of $f$ at two random points $x_1, x_2$, then we can uncover $U$ by perform simultaneous block-diagonalization of $\nabla^2f(x_1), \nabla^2f(x_2)$, which uses $\mathcal{O}(d^3)$ work. 

We note that there is a possibility that the algorithm fails to recover $U$; such a failure occurs if $\nabla^2f(x_1), \nabla^2f(x_2)$ both commute with another matrix that causes the  blocks $\nabla^2g(\hat{x}_i)$ to decompose further. This would then hide the true block structure of $\nabla^2g$.  This is of course not an issue in the completely separable case, where querying the Hessian at one point suffices. 
The random choice of $x_{1}, x_{2}$ reduces the failure probability. However, for analyzing the Hessian algorithm in the strictly block-separable case, we decide to be generous to the classical algorithm, so we assume that the number of points that the algorithm evaluates the Hessian to recover $U$ is $\mathcal{O}(1)$. Realistically, after a constant number of Hessian evaluations it should be very unlikely that they all commute with another operator that is not $U$.

After recovering $U$, we can now globally optimize $f$ by optimizing the blocks of $f(Ux)$ separately, which are of constant size by Assumption \ref{assump:constant_block}. The complexity of finding the global minimum depends on the size of the largest block. In the worst-case, since $g_i$ can be an arbitrary nonconvex function of dimension $d_i$, we can only perform grid searches over a $d_{i}$-dimensional grid, which have complexity $\Ocal((B/\epsilon)^{\max_{i} d_{i}})$ for grid search over a box of length $B$ with accuracy $\epsilon$ \cite{nesterov2018lectures}. We can assume $B = \mathcal{O}(1)$ wlog by rescaling the function. If $\hat{x}^\star_i$ is the global minimizer of $g_i$, then the grid search outputs $\hat{y}_i$, such that $\lvert g(\hat{y}_i) - g(\hat{x}^\star_i) \rvert < \epsilon$. Hence for an overall error $\epsilon$, we need to take $\epsilon \rightarrow \epsilon/d$, in the worst case. The total amount of work is then $\mathcal{O}\Big(d^2 + d^3 + d\left(d/\epsilon\right)^{\max_{i} d_{i}}\Big)$.
Hence, the proof of the following theorem is self-evident.

\hessianAlgo*

\subsection{Numerical Benchmarking of Off-the-Shelf Global Optimizers}
\label{sec:numerical-benchmarking}

\begin{figure}
    \centering
    \includegraphics[width=1.0\linewidth]{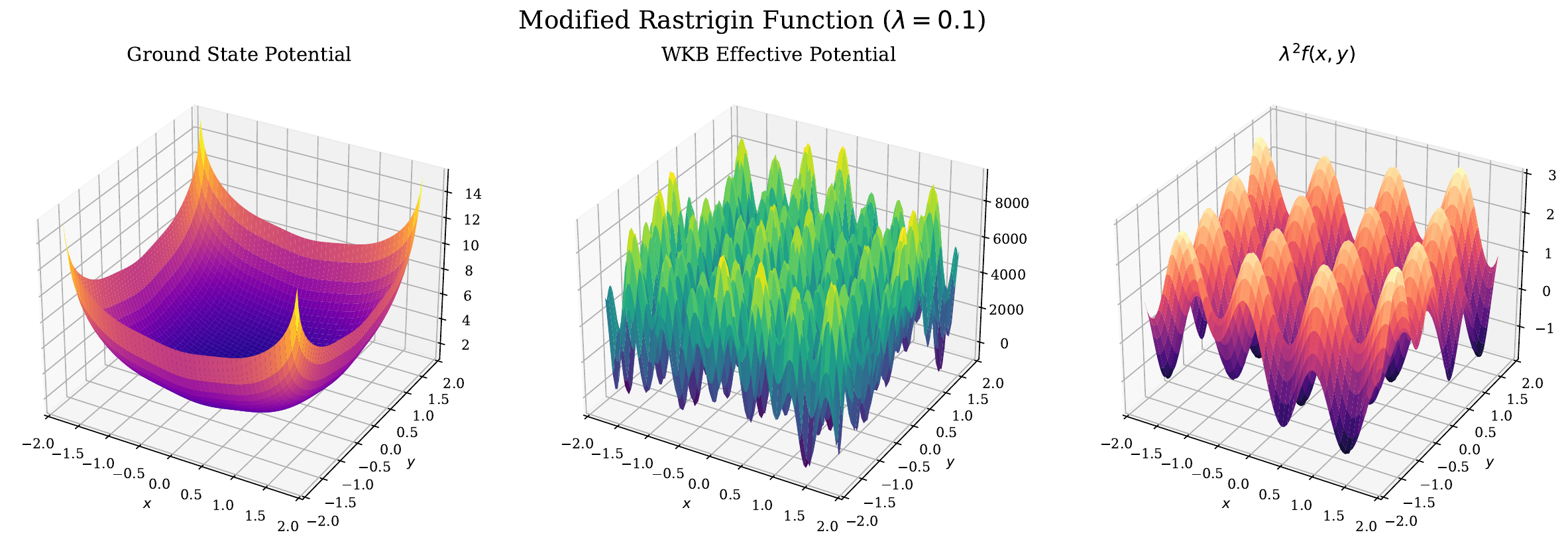}
    \includegraphics[width=1.0\linewidth]{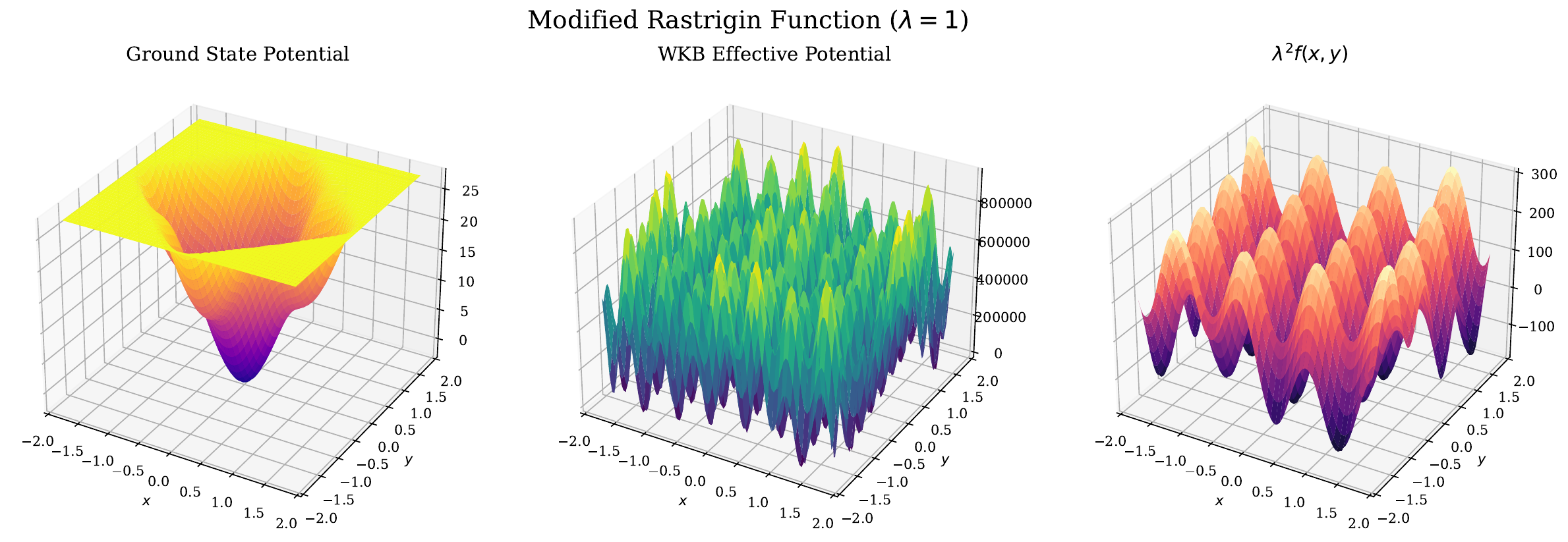}
    \includegraphics[width=1.0\linewidth]{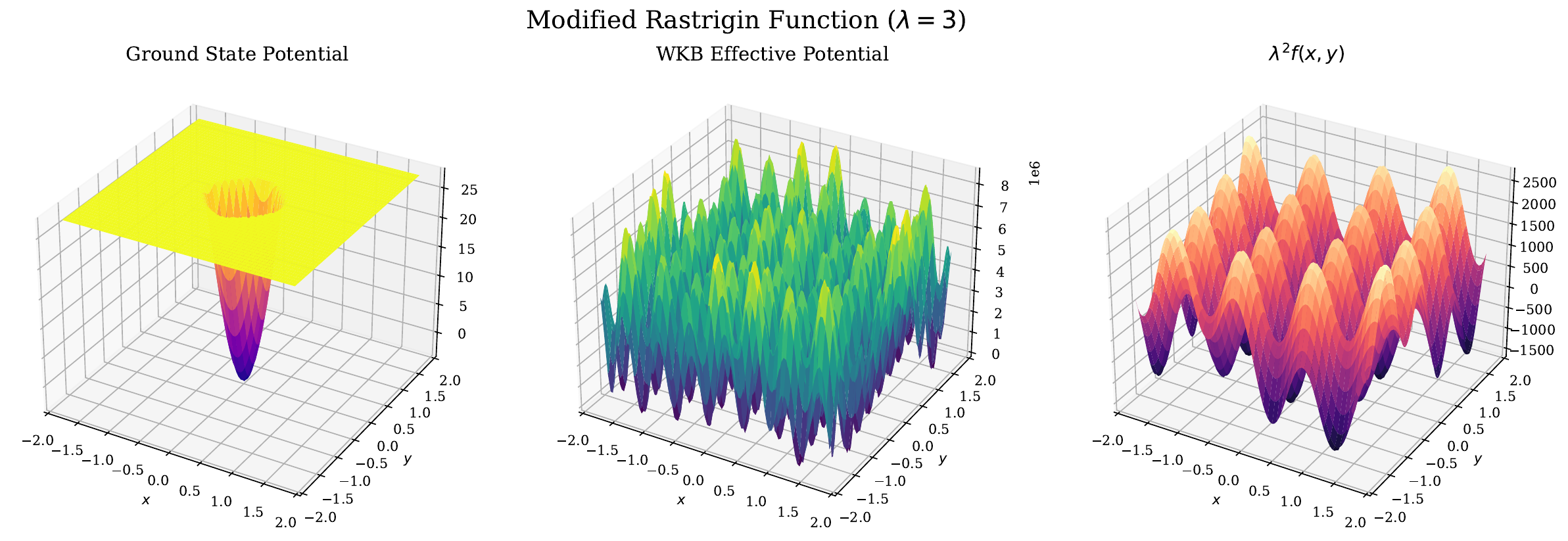}
    \caption{\textit{Visualization of the unique/non-unique minima separation for the modified, rotated version of Rastrigin function in \eqref{eq:modified-rastrigin}.}
    }
    \label{fig:modified-rastrigin}
\end{figure}

We now consider algorithms that do not directly fit into the Langevin diffusion framework. In particular, we demonstrate numerically evaluate three representative off-the-shelf global optimization methods (via their implementation in the popular Python library SciPy \cite{virtanen2020scipy}). These algorithms are Basin-hopping~\cite{wales1997global}, Dual Annealing~\cite{xiang1997generalized}, and Differential Evolution~\cite{storn1997differential}. Each of these algorithms is well-studied and collectively represent the three main classes of global optimization algorithms: hybrid local-global methods, simulated annealing, and genetic algorithms, respectively.

As a running example, we consider optimizing the following modified Rastrigin function satisfying:\footnote{The original Rastrigin function is of the form $f(x) = 10 d + \sum_{i=1}^d [ x_i^2 - 10 \cos(2\pi x_i) ] $.}
\begin{equation} \label{eq:modified-rastrigin}
    f(x) = 10d + \sum_{i=1}^d \left[
    x_i^2 \left( (x_i^2 - 1)^2 + 4 \right) \left( (x_i^2 - 4)^2 + \frac{1}{8} \right)
    - 100 \cos(2\pi x_i)\right].
\end{equation}
Then to hide the separable structure, we choose a rotation matrix $U$ at random and consider $f(Ux)$ at the target function to optimize. Besides being rotated, the differences between $f$ and the original Rastrigin function are: (1) a quartic term replaces the originally quadratic term within the summation, and (2) the cosine ``perturbation'' term has a higher weight of 100, compared to the original weight of 10. A one-dimensional plot of this function is illustrated in Figure~\ref{fig:1d-sgd}, and two-dimensional plots---including the WKB effective and the ground state potentials---of \eqref{eq:modified-rastrigin} are illustrated in Figure~\ref{fig:modified-rastrigin}. Below, we detail each method considered and motivate our reasoning in choosing \eqref{eq:modified-rastrigin} as a running example.

\subsubsection{Basin-hopping}
\label{sec:basin-hopping-benchmarking}

Basin-hopping~\cite{wales1997global,wales1999global,li1987monte} is a two-phase global optimization algorithm that combines a local minimization algorithm with a random perturbation mechanism to iteratively identify increasingly optimal local minima. It is particularly effective for functions whose local minima exhibit a ``funnel-like" structure and has been used to find minimum energy molecular structures.

Specifically, basin-hopping iteratively applies the following steps:
\begin{enumerate}
    \item From a current point $x_{\mathrm{old}}$ propose a new point $y$ sampled uniformly from $x_t + [-\delta,\delta]^d$ where $\delta$ is a parameter specifying the step-size.
    \item Run a local search algorithm such as L-BFGS-B\cite{byrd1995limited,zhu1997algorithm} initialized at $y$ to find a new local minima $x_{\mathrm{new}}$.
    \item Accept/reject $x_{\mathrm{new}}$ according to the Metropolis-Hastings rule at a specified temperature $T$; that is, always accept if $f(x_{\mathrm{new}}) < f(x_{\mathrm{old}})$ and otherwise accept with a probability of $\exp\left(-(f(x_{\mathrm{new}}) - f(x_{\mathrm{old}}))/T\right)$.
\end{enumerate}

\begin{figure}[H]
    \centering
    \includegraphics[width=1.0\linewidth]{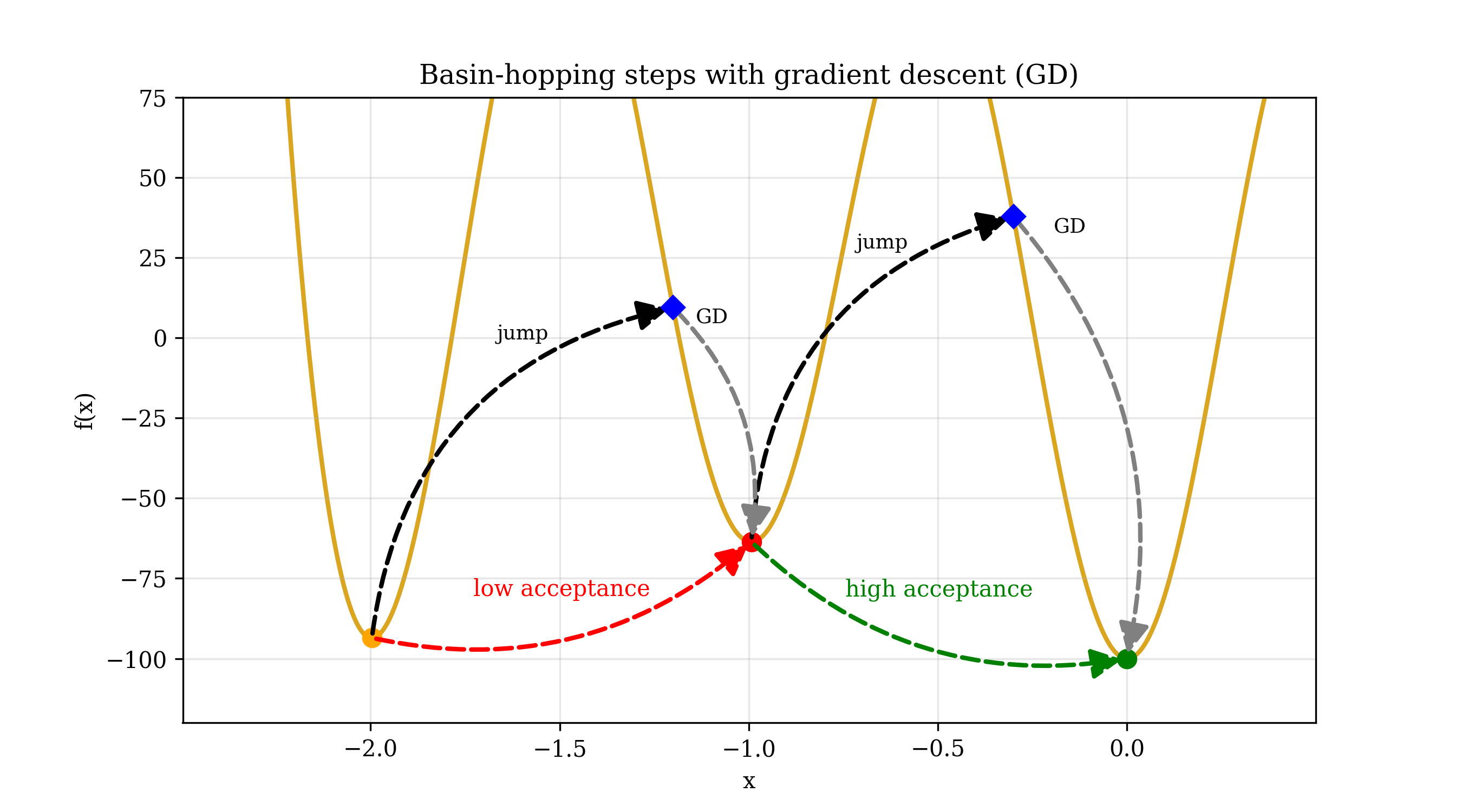}
    \caption{Basin-hopping on the modified Rastrigin function. The suboptimal local minimum at $x=-2.0$ is a trap for the hopping procedure. } 
    \label{fig:hopping-steps}
\end{figure}

Basin-hopping is a powerful algorithm as it can utilize efficient local minimization algorithms, including quasi-newton methods such as L-BFGS-B~\cite{byrd1995limited,zhu1997algorithm}, that may additionally employ techniques like line-search. Another advantage over the annealing and genetic algorithms we will discuss later is that the nature of the ``walls" between local minima are irrelevant to basin-hopping as it uses large random perturbations. The algorithm can therefore be modeled as (a potentially irreversible) random walk on the local minima of the function, whose properties are determined only by the step-size, the distance between the minima, and the function value at those minima.

In fact, these properties motivates our consideration of a modified function~\eqref{eq:modified-rastrigin} as standard benchmarks such as the original Rastrigin functions can in fact be optimized efficiently by basin-hopping with properly chosen parameters due to the landscape of their local minima. As a motivating example, we consider the primary function investigated in~\cite{leng2023quantum}, which is a completely-separable function (Definition \ref{def:separable-func}) with each $g_i = f$ given by
\begin{align} \label{eq:biquartic}
    f(x) = x^4 - (x - 1/32)^2 + c,
\end{align}
with $c$ chosen to make the function value at the global minimum equal to 0. From inspection, it can be seen that the local minima lie on the vertices of a hypercube and with appropriately chosen step-size, basin-hopping performs a Metropolis-Hastings walk at temperature $T$ and potential $f$, on the vertices of this hypercube. If the temperature is chosen to be low enough, for example $T \sim 1/d$, we can ensure that with high probability only decreasing moves are accepted. However from the function definition, it can be seen that each local minima is adjacent to at least one other that has a lower function value. As a consequence, there are no traps on the landscape of local minima and the algorithm is funneled in polynomial time to the global minimum. 

This theoretical explanation is confirmed by numerical benchmarking. In Figure~\ref{fig:optimize-biquadratic}, we illustrate that the function can be optimized by the basin-hopping algorithm with temperature equal to $1/d$ with the time to find a global minimum growing roughly linearly. As a consequence, we find that even a 200-dimensional version of the function can be optimized in roughly 5 minutes. In contrast the reported time-to-solution (TTS) in~\cite{leng2023quantum} is $\sim 10^4$ seconds for a 14 dimensional function (which in our experiments, can be optimize in less than 10 seconds). We offer a couple possible reasons for this discrepancy: firstly, unlike~\cite{leng2023quantum} our implementation explicitly checks if the global minimum is found after every iteration. Secondly, the temperature parameter $T$ significantly affects performance in this setting (as shown in Figure~\ref{fig:optimize-biquadratic}, optimization is much less effective with $T=1$.)

The above arguments broadly show that Basin-hopping is likely to be successful for functions where the landscape of the local minima itself does not have any traps. For this reason, benchmarks such as the Rastrigin function are unlikely to be exponentially hard for Basin-hopping. This motivates us to construct a new separable function, the so called ``modified Rastrigin function" that exhibits traps even when we only consider the local minima. We illustrate this property in Figure~\ref{fig:hopping-steps}. For this function, we observe (as detailed in Figure~\ref{fig:basin-hopping-modified-rastrigin}) that the runtime scales exponentially with the dimension $d$. In this setting, we observe no major qualitative difference from choosing the temperature to be $1/d$ instead of the default of $1$.

\begin{figure}[H]
    \centering
    \includegraphics[width=0.45\linewidth]{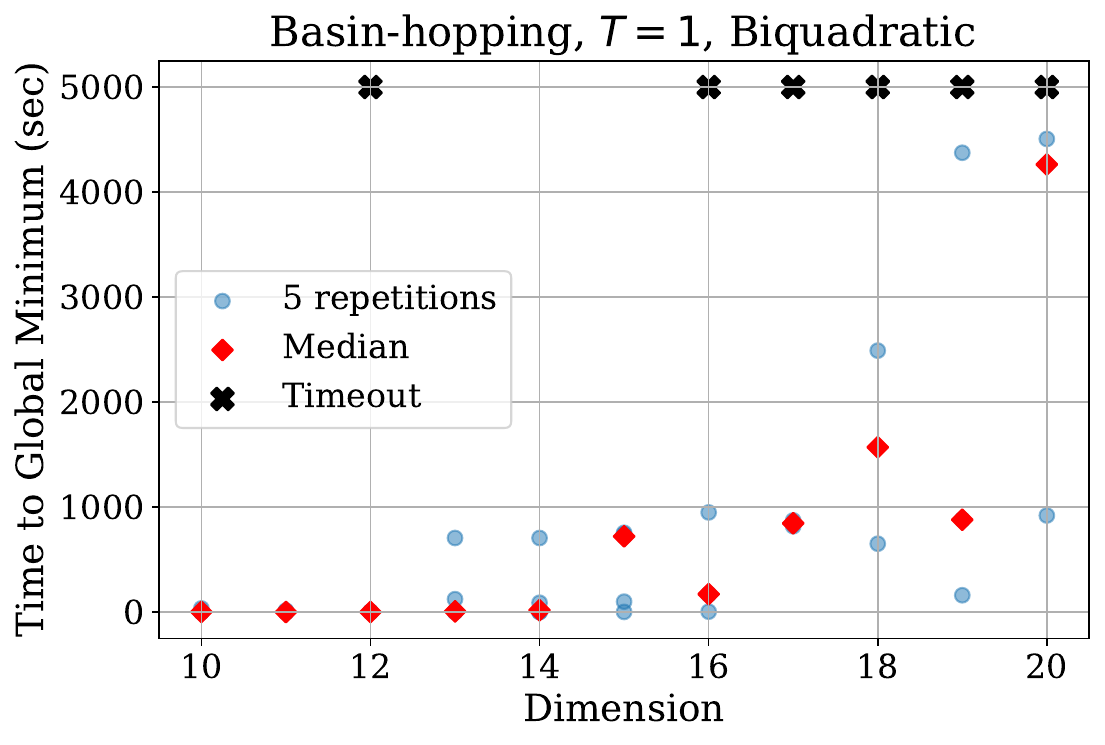}
    \includegraphics[width=0.45\linewidth]{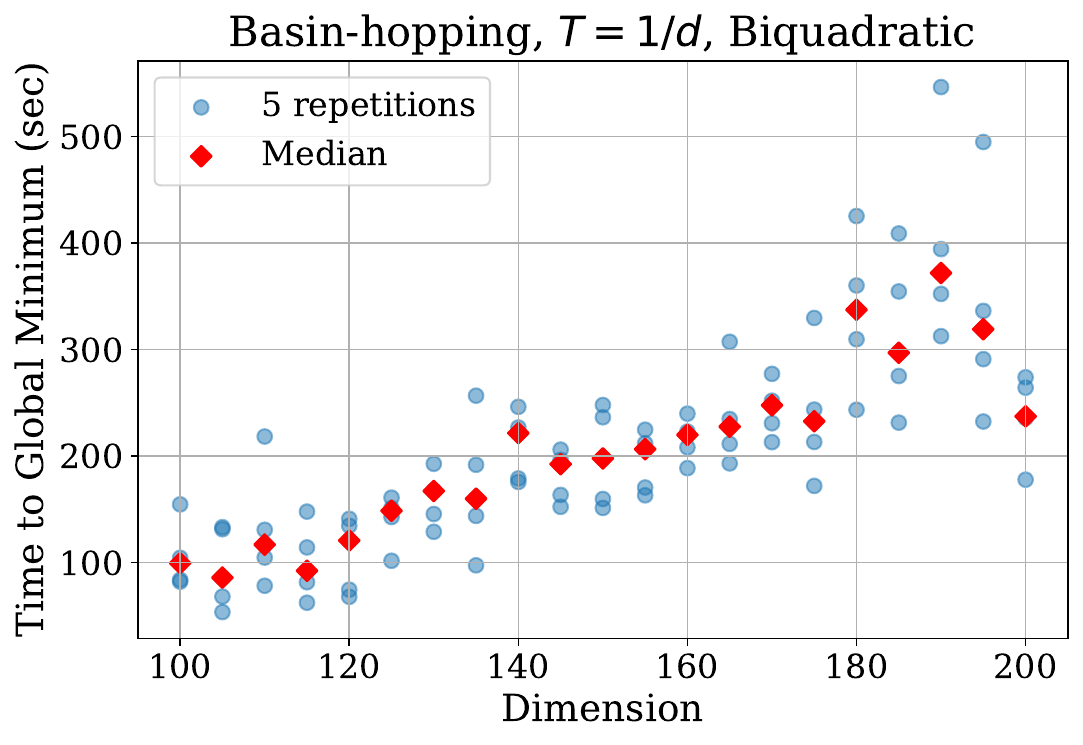}
    \caption{\textit{Time to reach global minimum of \eqref{eq:biquartic} using Basin-hopping with $T=1$ and $T=1/d$.}
    The hard function in~\cite{leng2023quantum} is easily optimized by Basin-hopping, if the temperature $T \sim 1/d$. Each method is repeated for 5 times per dimension, with the median indicated with red $\diamond$. Black $\times$ marks indicate the trial could not reach global optimum within the corresponding time.
    }
    \label{fig:optimize-biquadratic}
\end{figure}

\begin{figure}[H]
    \centering
    \includegraphics[width=0.45\linewidth]{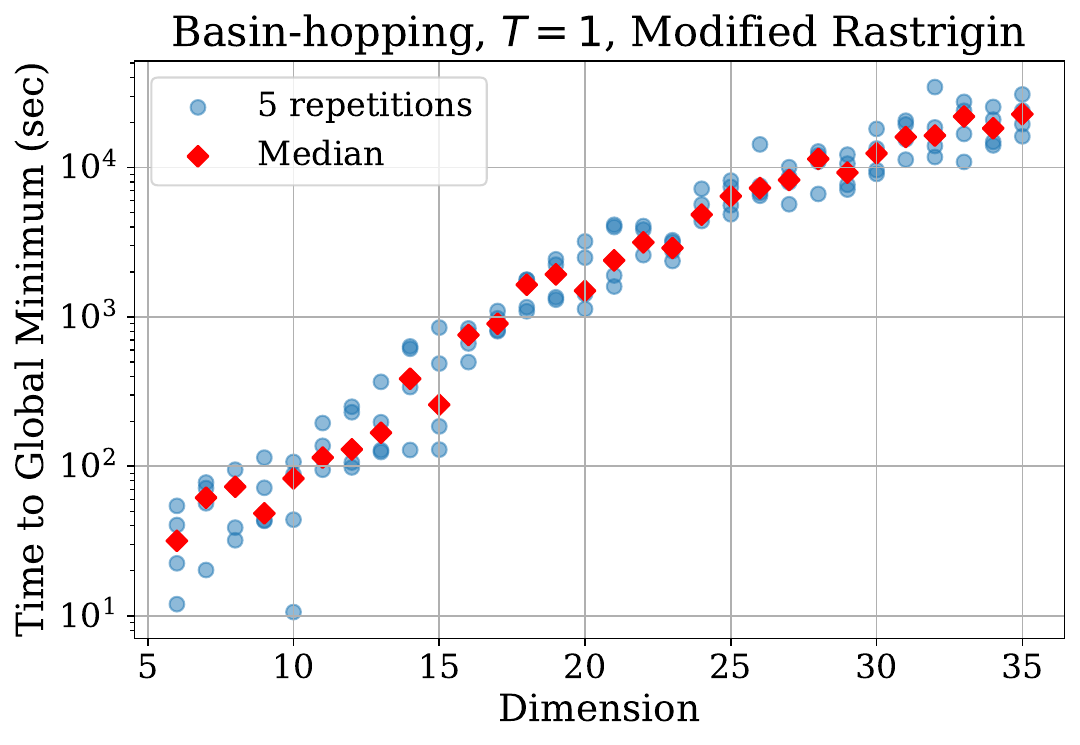}
    \includegraphics[width=0.45\linewidth]{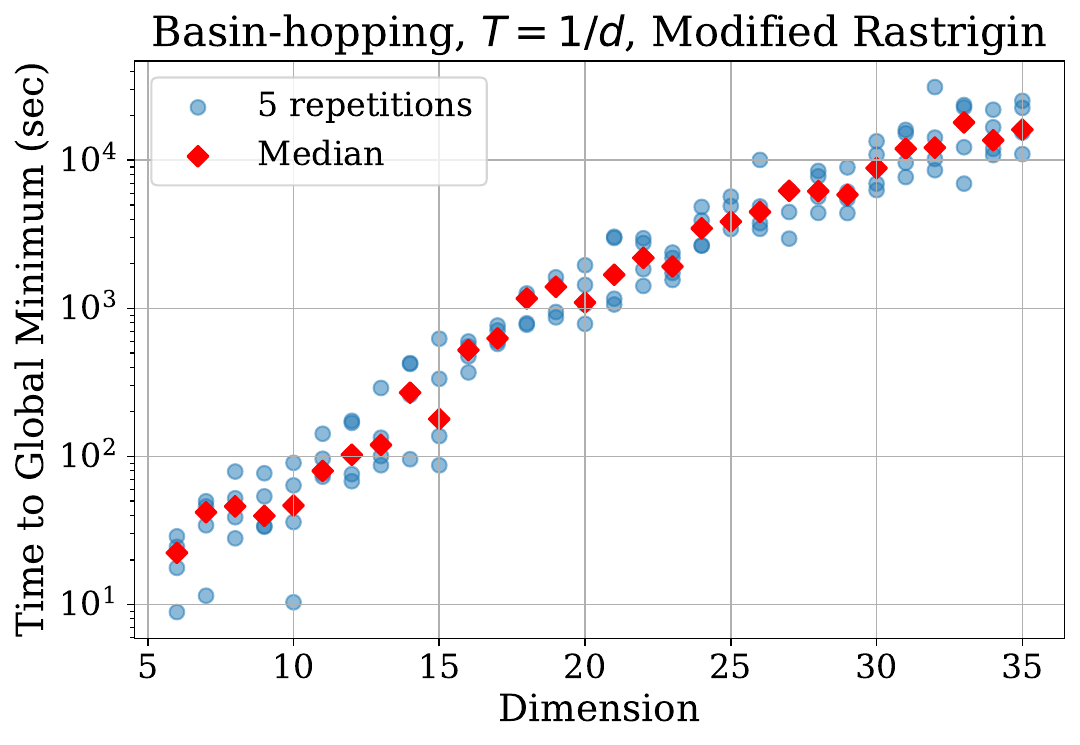}
    \caption{\textit{Time to reach global minimum of \eqref{eq:modified-rastrigin} using Basin-hopping with $T=1$ and $T=1/d$.} Each method is repeated 5 times per dimension, with the median indicated with red $\diamond$.}
    \label{fig:basin-hopping-modified-rastrigin}
\end{figure}

\subsubsection{Dual Annealing (Simulated Annealing)}
\label{sec:dual-annealing-benchmark}
\begin{figure}[H]
    \centering
    \includegraphics[width=0.45\linewidth]{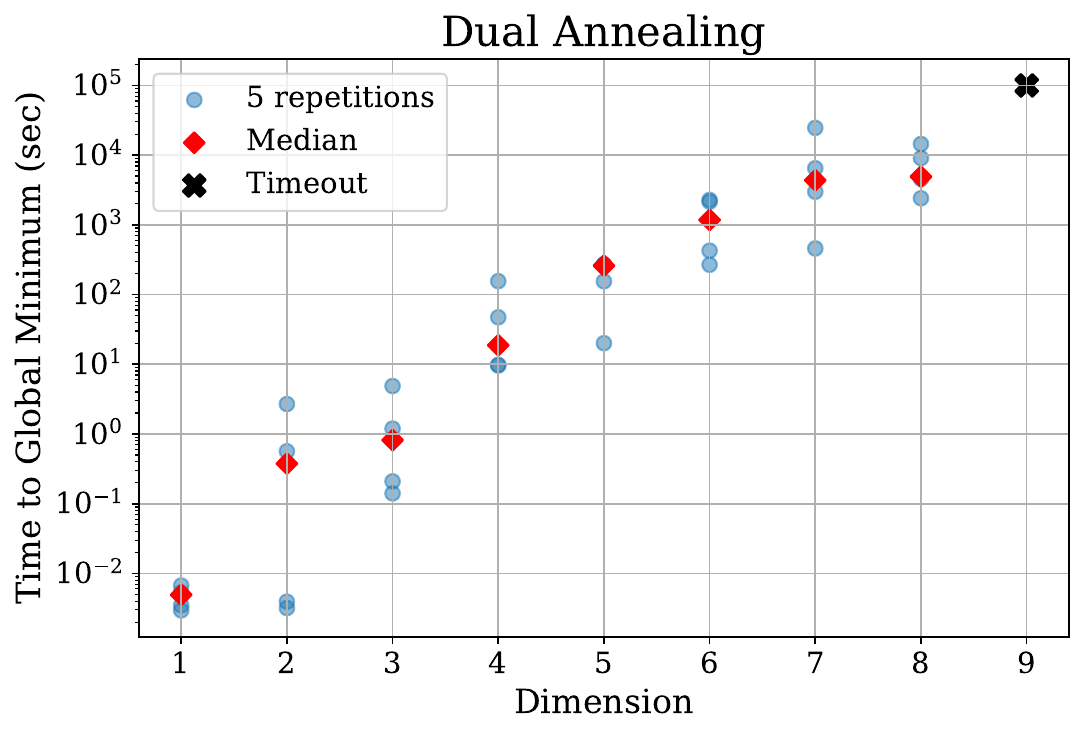}
    \includegraphics[width=0.45\linewidth]{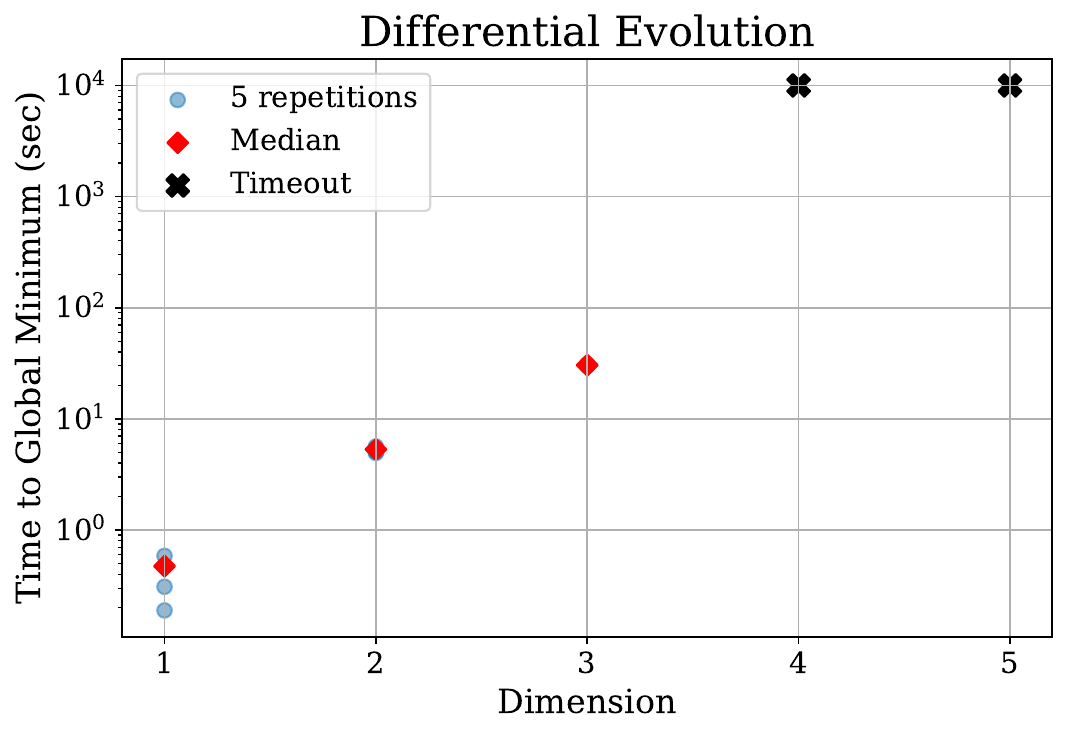}
    \caption{\textit{Time to reach global minimum of \eqref{eq:modified-rastrigin} using Dual Annealing and Differential Evolution.} Each method is repeated for 5 times per dimension, with the median indicated with red $\diamond$. Black $\times$ marks indicate the method could not reach global optimum within the corresponding time.
    }
    \label{fig:dual-annealing-diff-evolution}
\end{figure}

Dual Annealing (DA)~\cite{xiang1997generalized} is a global optimization algorithm inspired by (simulated) annealing, where a system is gradually cooled to reach a low-energy state. DA is inherently stochastic in that it explores the search space by accepting both better and worse solutions with some probability, guided by a temperature parameter that decreases over time. 
DA sophisticates the vanilla simulated annealing by combining a stochastic global search with deterministic local optimization steps to refine promising solutions. 

The performance of DA for optimizing the function in \eqref{eq:modified-rastrigin} is illustrated in Figure~\ref{fig:dual-annealing-diff-evolution} (left panel). We use the default setting implemented in Scipy, except for ensuring that the maximum number of iterations and function evaluations are large enough. While it can optimize up to $d=8$ in a reasonable amount of time, dimensions above $d=9$ could not be optimized within $10^5$ (sec) $\approx$ 2.8 (hr). This makes sense given the exponentially increasing search space from which DA needs to explore, on top of many local minima as illustrated in Figure~\ref{fig:modified-rastrigin}.

\subsubsection{Differential Evolution (Genetic Algorithm)}
\label{sec:differential-evolution-benchmark}
Differential Evolution (DE)~\cite{storn1997differential} is a population-based global optimization algorithm that searches for the minimum by ``evolving'' a set of candidate solutions, successively. At each step, DE creates new candidate solutions (population) by combining existing ones through mutation and crossover operations, and then selects the best candidates based on their objective values. 

Naturally, the population size plays a crucial role: a larger population provides more diverse candidate solutions, increasing the likelihood of exploring the search space thoroughly and finding the global optimum; however, a large population size necessarily increases the time needed to find a solution. As the problem dimension increases, the search space grows exponentially, making it more challenging for the algorithm to cover the space adequately. 

To reflect this, we increase the population size with the dimension, specifically as $1000 \cdot d$; we use the default hyperparameters for the rest, again except for ensuring large enough maximum of iterations. The results are illustrated in Figure~\ref{fig:dual-annealing-diff-evolution} (right panel). As can be in seen in the figure, despite scaling the population size with the dimension differential evolution fails to reach global optimum even for $d=4$, within the time window of $10^4$ seconds. Similarly to the DA case, DE is trapped in a local minima, and is unable to make progress.

\subsubsection{Other Algorithms}

\paragraph{Local Optimizers:} The majority of results in optimization theory for nonconvex optimization have focused on the problem of finding local minima rather than global optimization. This research has resulted in a large number of highly sophisticated algorithms for local optimization. These include first-order algorithms such as perturbed/stochastic gradient descent~\cite{jin2017escapesaddlepointsefficiently,jin2018accelerated}, Adam~\cite{kingma2014adam} and RMSProp~\cite{tieleman2012lecture}; derivative-free algorithms based on trust region optimization such as COBYLA/BOBYQA~\cite{powell2007view}; quadratic programming methods such as Sequential Quadratic Programming~\cite[Chapter 18]{nocedal2006numerical}; quasi-Newton methods such as BFGS~\cite{broyden1970convergence,fletcher1970new,goldfarb1970family,shanno1970conditioning}, L-BFGS~\cite{liu1989limited} and L-BFGS-B~\cite{byrd1995limited,zhu1997algorithm}; and interior-point methods such as IPOpt~\cite{wachter2006implementation}.

Most of these algorithms do not have an inbuilt mechanism to escape a local minima once found and terminate once such a point is found. In order to use them for local optimization they must be used as part of an outer procedure that performs the search over local minima. The most natural choice is to simply try many random initializations of the local algorithm and report the best classical algorithm. Block-coordinate separable functions can be constructed to have an exponentially large number of local minima distributed uniformly through the domain of optimization. A random-restart method necessarily requires an exponentially large number of restarts to globally optimize such functions. An alternative is to use a more sophisticated outer scheme such as Basin-hopping, which has already been discussed and benchmarked in Section~\ref{sec:basin-hopping-benchmarking}. Finally, algorithms such as SGD that use stochasticity to escape local minima are theoretically analyzed in Section~\ref{sec:sgd_lower_bound}.

\paragraph{Branch-and-Bound methods:} One of the most popular global optimization methods is the use of exact Mixed Integer Programming solvers. These solvers include the highly optimized commercial solvers Gurobi~\cite{gurobi} and CPLEX~\cite{cplex2009v12} and are based on an underlying Branch-and-Bound procedure. Branch-and-Bound algorithms solve complex problems by solving a series of simpler sub-problems or relaxations, which are then iteratively refined in tree like manner to discover solutions for the original problem. The algorithm typically takes the shape of a search over a tree of relaxations (each of which is referred to as a node). A Branch-and-Bound mechanism requires a relaxation scheme and the specification of two further procedures: a bounding procedure that computes at each node a valid lower bound on the optimal solution of the original problem, and a branching procedure that constructs new relaxations based on the bound at a particular node. There are no universal specifications for these procedures, and have been studied primarily for problems such as Mixed Integer Linear Programs or Quadratic Programs, or Second Order Cone Programs. For general classes of continuous nonconvex functions of the type considered in this paper. Finally, we note that branch-and-bound algorithms are ultimately clever exhaustive search algorithms that use the bounding procedure to prune significant parts of the search space. For this reason, they are usually observed to exhibit exponential scaling with problem size in general, even for problems such as mixed-integer linear programming for which they are highly optimized~\cite{dey2023lower}. Therefore, even if an appropriate bounding procedure is found, we should expect exponential scaling for the function classes considered in this paper unless the bounds can be specifically optimized for the particular function. We leave the investigation of this possibility to future work.

\subsection{Simulated Quantum Annealing}
\label{sec:sqa}

In this section, we reason about the performance of simulated quantum annealing (SQA) based on path-integral quantum Monte Carlo for nonconvex optimization. We provide reasoning for why SQA is likely to be obstructed from globally optimizing such functions, due to mechanisms similar to  those discussed in Section \ref{sec:quat_clas_block_sep} and Section \ref{sec:sgd_lower_bound}. 

Given the potential computational speedups attainable with the quantum adiabatic algorithm (QAA), which have mostly come from quantum tunneling \cite{farhi2002quantum, reichardt2004quantum}, it is reasonable to ask if quantum Monte Carlo (QMC) methods, which can sometimes mimic quantum tunneling rates, can also match QAA's performance. The prime candidate for producing a separation between QAA and classical stochastic dynamics, in discrete space, has been the Hamming-weight-spike potential \cite{reichardt2004quantum}. This potential introduces a tall, thin barrier that obstructs classical simulated annealing but not quantum annealing. Interestingly, it was shown that this does not lead to a separation against more sophisticated classical algorithms, since SQA also globally optimizes the Hamming-weight spike in polynomial time \cite{crosson2016simulated}. On the contrary, there are  settings where it appears challenging for SQA to efficiently mix, even when QAA is efficient, due to certain topological obstructions \cite{hastings2013obstructions}, e.g. many tunneling paths for path-integral QMC, and even leading to sub-exponential black-box separations for a search-like problem \cite{gilyen2021sub}. The authors of \cite{leng2025sub} translated the setting of \cite{gilyen2021sub} into the setting of continuous, nonconvex optimization. Unfortunately, the resulting problem is very unnatural for traditional nonconvex optimization and ends-up appearing like an unstructured search problem. 
Hence, the quantum algorithm requires some additional advice to lead to a super quadratic separation.

SQA is a classical algorithm that attempts to find the fixed-point of the imaginary-time evolution of a quantum system with $d$ spatial dimensions by simulating a classical stochastic system of $d$-dimensional trajectories, i.e. an additional imaginary-time parameter.
The continuous-variable version of SQA samples from the Gibbs distribution $\mu[x; \lambda] \propto \exp(- S[x; \lambda])$, under the classical action in imaginary time:
\begin{align}
    S[x; \lambda] = \int_0^T \mathrm{d}\tau \, \left[ \frac{1}{2} \left(  \frac{\mathrm{d} x (\tau)}{\mathrm{d} \tau} \right)^2 + \lambda^2 f(x(\tau)) \right]
\label{eqn:action_c}
\end{align}
where $\tau \in [0, T]$ is the imaginary time, $x: [0, T] \rightarrow \mathbb{R}^d$ is called a worldline (path), $f$ is the nonconvex function to be optimized, and $\lambda$ scales the potential $f$ everywhere. As shown in Lemma \ref{lem:beta-lower-bound}, even in relatively ideal settings, $\lambda = \omega(\log(d))$ for the Gibbs distribution to sufficiently concentrate around the global minimizer. Hence, if the gap depends exponentially on some function that grows faster than a $\log$, then SQA will have a super-polynomial runtime.

One can also define an infinite-dimensional analog of the Witten Laplacian \cite{brooks2019sharptunnelingestimatesdoublewell}, which leads to a connection between Schr\"odinger operators and Langevin dynamics for SQA in a similar vein to Section \ref{sec:main-mechanism}. This leads to a Schr\"odinger-like operator with a potential that depends on the variations of the action functional, similar to how the WKB potential depends on the derivatives of the potential. Like the WKB potential, the infinite-dimensional analog can have additional global minima that are not present in $S$. Hence, this also puts SQA in the tunneling regime, which, as discussed earlier, is expected to lead to exponentially falling gaps.

Let $X_t$ be a function-valued stochastic process, and $\widehat{B}_t$ is a function-valued Brownian motion. The value of $X_t$ at fixed $t$ corresponds to a random space-imaginary-time path $x(\tau)$. In continuous-time and space, one can model path-integral SQA as the following functional Langevin diffusion
\cite{damgaard1987stochastic}:
\begin{align}
 \mathrm{d}X_t = -\frac{\delta S[X_t; \lambda]}{\delta X_t} \mathrm{d}t+ \mathrm{d}\widehat{B}_t,
\end{align}
where $t$ is the Langevin diffusion time.
The above formally looks like \eqref{eqn:classical_lang_poten_dyn} but is now an infinite-dimensional SDE, with the drift replaced by the functional derivative of the action. By analogy, one can also see that the stationary distribution is the Gibbs measure under the action functional. It has been shown that the SQA gap depends on the action at critical points around the Morse saddle barrier in the original potential $f$ \cite{isakov2016understanding, brooks2019sharptunnelingestimatesdoublewell}. The authors of \cite{digesù2015smallnoisespectralgap} showed that for a discretized action corresponding to a double well, the gap falls exponentially in $\lambda$, which can be translated to the continuous-time case \cite{brooks2019sharptunnelingestimatesdoublewell}. While this was for a symmetric double well, it is reasonable to believe that for an asymmetric double well, the infinite-dimensional WKB potential corresponding to its action functional can have multiple global minima and hence suffer from an exponentially decaying gap as well. We leave it to future work to show this result rigorously for the asymmetric double well.

\section*{Acknowledgments}
The authors thank Rob Otter and Shaohan Hu for their support and valuable feedback on this project. We also acknowledge our colleagues at the Global Technology Applied Research Center of JPMorganChase, especially Brandon Augustino, Jacob Watkins, Enrico Fontana, Abhishek Som, Akshay Seshadri, and Niraj Kumar, for many helpful technical discussions.

\section*{Disclaimer}
This paper was prepared for informational purposes by the Global Technology Applied Research center of JPMorgan Chase \& Co. This paper is not a product of the Research Department of JPMorgan Chase \& Co. or its affiliates. Neither JPMorgan Chase \& Co. nor any of its affiliates makes any explicit or implied representation or warranty and none of them accept any liability in connection with this paper, including, without limitation, with respect to the completeness, accuracy, or reliability of the information contained herein and the potential legal, compliance, tax, or accounting effects thereof. This document is not intended as investment research or investment advice, or as a recommendation, offer, or solicitation for the purchase or sale of any security, financial instrument, financial product or service, or to be used in any way for evaluating the merits of participating in any transaction.

\printbibliography

\appendix

\section{Additional Details on Hypercontractivity Perturbation Theorem}
\label{sec:constants-appendix}
Here we show that the constants $\alpha, \beta$ in Theorem~\ref{thm:intrinsic-hypercontractivity} do not have $d$ dependency even if $\omega$ scales with $d$. The analysis only follows by tracking down the constants in the original paper~\cite{gross2025invariance}. We first note that the original result is presented as
\[
\tilde{\omega}^{-1} \leq \alpha (\|e^W\|_\kappa \|e^{-W}\|_\nu)^{-\beta}.
\]
First we consider the definitions below.
\begin{align*}
    \beta_1 &= \frac{\log 3}{(2c)^{-1}-\nu^{-1}}(4b_{\kappa}-1)+2\alpha_1\\
    \beta_2 &= 12 \left(\frac{2(4b_{\kappa}-1)}{a_\nu}-1\right)\\
    \beta_3 &= 4b_\kappa-2\\
    \beta_4 &= (24 e^{12/e})^{\left(\frac{2(4b_\kappa-1)}{a_\nu}-1\right)}2^{2-b_\kappa}\\
    \beta_5 &= \beta_1+2\beta_2(2a+c\log_3/b_\kappa)
\end{align*}
where $c = 1/\omega_0, a_\nu = \sqrt{1-2c/\nu}$, $b_\kappa = \sqrt{1+2c/\kappa}$, $c_\nu =\frac{c}{1-2c/\nu},a = 2c_\nu b_\kappa $ and $\alpha_1 = a+(c\log3)/b_\kappa $. Furthermore,
\begin{align*}
    e_1 &= \beta_5 + \kappa \beta_3\\
    d_1 &= 2a + 8c (1+2c/\kappa)\beta_4
\end{align*}
Then by the definition of these constants $\alpha, \beta$ in \cite{gross2025invariance}, we can write
\begin{align*}
    \alpha &= a+d_1\\
    \beta &= e_1 + (2a+ c\log 3/b_\kappa)
\end{align*}
Assume that we choose $\nu=\kappa = 2/\omega$. Then, $\beta_1 = \Theta(1/\omega), \beta_2 = \Theta(1), \beta_3 = \Theta(1), \beta_4 = \Theta(1), \beta_5 = \Theta(1/\omega) $. This implies that $e_1 = \Theta(1/\omega), d_1 = \Theta(1/\omega)$. Finally we have $\alpha = \Theta(1/\omega)$ and $\beta = \Theta(1/\omega)$. Hence, alternatively, we can alternatively present the result as 
\[
\tilde{\omega}\geq \omega \alpha M^{\beta}
\]
by redefining the constants $\alpha, \beta$ so that they are independent of $\omega$ and we write the $\omega$ dependence explicitly.

\subsection{Proof of Theorem \ref{thm:dirichlet-log-sobolev-strongly-convex}}
\label{sec:proof-of-dirichlet-log-sobolev-strongly-convex}

We recall the spectral gap comparison result of \cite{andrews2011prooffundamentalgapconjecture} based on the modulus of convexity.
\begin{theorem}[\cite{andrews2011prooffundamentalgapconjecture}, Spectral Gap Comparison Theorem]
\label{thm:modulusofconvec_compare}Consider the Dirichlet eigenvalue problem for the operator $-\Delta + f$ on a compact convex domain $\Omega$. If $\omega$ is such that 
\begin{align*}
    \langle \nabla f(x) - \nabla f(y), \frac{x-y}{\lvert x -y\rvert}\rangle \geq \omega'(\frac{\lvert x - y \rvert}{2}),
\end{align*}
for all $x ,y \in \Omega$.
Then the spectral gap is lower bounded by the spectral gap for the Dirichlet problem associated with the following one-dimensional operator:
\begin{align}
    -\frac{\mathrm{d}^2}{\mathrm{d}x^2} + \omega(x), x \in [-D/2, D/2],
\end{align}
where $D = \textup{diam}(\Omega)$.
\end{theorem}

\begin{proof}[Proof of Theorem \ref{thm:dirichlet-log-sobolev-strongly-convex}]
    Since $\lambda^2 f$ is $\lambda^2\mu$ strongly convex, Theorem \ref{thm:modulusofconvec_compare} implies that we can look at the gap of the one-dimensional operator
    \begin{align*}
       \tilde{H}(\lambda) =  -\frac{\mathrm{d}^2}{\mathrm{d}x^2} + \frac{\lambda^2 \mu x^2}{2}, x \in  [-\textup{diam}(\Xcal)/2, \textup{diam}(\Xcal)/2].
        \end{align*}
    Also Theorem 2.7 \cite{gong2015probabilistic} implies that if $\delta$ is gap of $\tilde{H}$, then $H(\lambda) = -\Delta + \lambda^2 f$ is intrinsically hypercontractive with LSI constant at least $\delta$.
    By Theorem \ref{thm:dirichlet-log-sobolev-strongly-convex}, the gap of $\tilde{H}$ is $\Omega\left(\sqrt{2\mu}\lambda + \textup{diam}(\Xcal)^{-2}\right)$. Hence the result follows.
\end{proof}

\section{Additional Details on Semiclassical Analysis Results}

\subsection{Proof of Theorem \ref{thm:pert_quad_envelop}}
\label{sec:appendix_semi_classical}

In this section we focus on proving Theorem \ref{thm:pert_quad_envelop}. Before proceeding to the proof, we derive an eigenstate localization result that makes use of the Agmon metric.

We recall the Agmon metric \cite{simon1984semiclassicaltunneling} associated with a Schr\"odinger operator and  eigen-energy $E$,
\begin{align}
\label{eqn:agmon_metric}
\rho_{E}(y, x) = \inf_{\gamma \in C^1([0,1])} \left\{ \int_0^1 \sqrt{(V(\gamma(s)) - E)_+} \lVert \dot{\gamma}(s)\rVert \mathrm{d}s | \gamma(0)= y, \gamma(1) = x \right\},
\end{align}
and define $F_{E} = \{ y \in \mathbb{R}^d | V(y) - E \geq 0\}$ to be the ``classically-allowed region.'' The Agmon distance from the allowed region is $\rho_{E}(x) := \inf_{y \in F_{E}} \rho(y, x)$. 

However, one will note that the integrand only increases outside of the forbidden region. Hence, we can without loss of generality assume that the paths we are minimizing over do not leave $F_E$ once they enter it. We will denote this set $\tilde{C}([0, 1])$. Thus, when minimizing over $\tilde{C}([0, 1])$, we can drop the $( \cdot )_+$, as $V(\gamma(s)) - E \geq 0$ for the paths considered. In this case, \cite{hislop2012introduction} showed that the Agmon metric is equivalent to
\begin{align*}
\rho_{E}(y, x) = \inf_{\gamma \in \tilde{C}^1([0,1])} \left\{ \frac{1}{2}\int_0^1  \left(\lVert \dot{\gamma}(s)\rVert^2 +  V(\gamma(s))\right) \mathrm{d}s - \frac{E}{2}| \gamma(0)= y, \gamma(1) = x \right\},
\end{align*}
which easier to work with analytically.

Agmon's Theorem (in the form of \cite{steinerberger2021effectiveboundsdecayschrodinger}) states that for an eigenstate $\psi$ of a Schr\"odinger operator with energy $E$:
\begin{align}
\label{eqn:pointwise-localization}
    \lvert \psi(x) \rvert^2 \leq m_{\epsilon} \sup_{y \in F_E} e^{-2(1-\epsilon)\rho_E(y, x)} 
\end{align}

\begin{Lemma}[Localization under Quadratic Enveloped Potentials]
\label{lem:localization_bound_quadratic_pot}
Suppose that $f : \mathcal{X} \rightarrow \mathbb{R}$ satisfies
\begin{align*}
    c_f \lVert x - x^{\star}\rVert^2 \leq f(x) - f(x^{\star}) \leq C_f \lVert x - x^{\star} \rVert^2, \forall x \in \mathcal{X},
\end{align*}
where $c_f = \Omega_d(1)$ and $\text{diam}(\mathcal{X}) = \Omega(d)$.  Let $r= \Omega(\sqrt{d\max\left(\ln(d/\lambda), \ln(d)\right)})$.
If $\xi_1, \xi_2$ are the ground and first-excited states of $-\Delta + \lambda^2 f$, then  for $i \in \{1, 2\}$
\begin{align*}
    \mathbb{P}_{\lvert \xi_i\rvert^2}[\lVert x - x^{\star}\rVert \geq r] = o_d(1).
\end{align*}
\end{Lemma}
\begin{proof}
Without loss of generality, we take $x^{\star} = f(x^{\star}) = 0$. Hence $f(x) \leq C_f\lVert x\rVert^2$.

Let $A = -\Delta + \lambda^2 C_f\lVert x \rVert^2$, with ground and first excited states $\psi_1, \psi_2$, respectively. From the minimax principle
\begin{align*}
    E_1(\lambda) \leq E_2(\lambda) = \sup_{\zeta} \inf \{ \langle \psi | H | \psi\rangle , \lVert \psi \rVert =1 , \psi \perp \zeta\}.
\end{align*}

Fix $\epsilon > 0$ arbitrary and $\zeta$ such that
\begin{align*}
    E_2(\lambda) \leq \inf \{ \langle \psi | H | \psi\rangle , \lVert \psi \rVert =1 , \psi \perp \zeta\} + \epsilon.
\end{align*}

Hence pick normalized $\chi \in \text{span}(\{ \psi_1, \psi_2\})$ such that $\chi \perp \zeta$. 
\begin{align*}
E_2(\lambda) &\leq \langle \chi | A | \chi \rangle  \leq \lambda\frac{3}{\sqrt{2}}d\sqrt{C_f}
\end{align*}
Since $\epsilon$ was arbitrary and $H(\lambda)$ has at least two bound states
\begin{align}
\label{eqn:E_2_upper_bound}
    E_2(\lambda) \leq \frac{3}{\sqrt{2}}d\sqrt{C_f}\lambda.
\end{align}

Let $y_x$ be the closest point to $x$ in the $\rho_E$ metric. Then
\begin{align*}
    \inf_{y\in F_E}\rho_E(y,x) \geq  \frac{1}{2}\inf_{\gamma \in \tilde{C}^1([0,1])} \left\{ \int_0^1 \lVert \dot{\gamma}(s)\rVert^2 + c_f\lambda^2\lVert \gamma(s)\rVert^2\mathrm{d}s  -  \frac{3}{\sqrt{2}}d\sqrt{C_f}\lambda: \gamma(0)= y_x, \gamma(1) = x \right\}.
\end{align*}

For simplicity, let $y_x$ be $y$ below. Solving the Euler Lagrange equations, we get that the minimizing curves (due to convexity of  the Lagrangian) are of the form
\begin{align*}
\gamma(s) = Ae^{-\lambda\sqrt{c_f}s} + Be^{\lambda\sqrt{c_f}s},
\end{align*}
where
\begin{align*}
&A = \frac{e^{\lambda\sqrt{c_f}}y - x}{e^{\lambda\sqrt{c_f}} - e^{-\lambda\sqrt{c_f}}}\\
&B = \frac{x-e^{-\lambda\sqrt{c_f}}y}{e^{\lambda\sqrt{c_f}} - e^{-\lambda\sqrt{c_f}}}.
\end{align*}

By plugging the curve back into the Lagrangian (without the constant shift) and integrating we get
\begin{align*}
&\frac{\lambda}{2}\sqrt{c_f}[\lVert A \rVert^2(1-e^{-2\lambda\sqrt{c_f}} ) + \lVert B\rVert^2(e^{2\sqrt{c_f\lambda}} - 1)]\\
&=\frac{\lambda}{2}\sqrt{c_f}\left(\frac{1-e^{-2\lambda\sqrt{c_f}}}{(e^{\lambda\sqrt{c_f}} -e^{-\lambda\sqrt{c_f}})^2}\right) \lVert e^{\lambda\sqrt{c_f}}y - x\rVert^2 + \lambda\sqrt{c_f}\left(\frac{e^{2\lambda\sqrt{c_f}} -1}{(e^{\lambda\sqrt{c_f}} -e^{-\lambda\sqrt{c_f}})^2}\right) \lVert e^{-\lambda\sqrt{c_f}}y - x\rVert^2 \\
&\geq\frac{\lambda}{2}\sqrt{c_f}\left(\frac{1-e^{-2\lambda\sqrt{c_f}}}{(e^{\lambda\sqrt{c_f}} -e^{-\lambda\sqrt{c_f}})^2}\right) ( e^{\lambda\sqrt{c_f}}\lVert y \rVert - \lVert x\rVert)^2 + \lambda\sqrt{c_f}\left(\frac{e^{2\lambda\sqrt{c_f}} -1}{(e^{\lambda\sqrt{c_f}} -e^{-\lambda\sqrt{c_f}})^2}\right) (e^{-\lambda\sqrt{c_f}}\lVert y \rVert - \lVert x\rVert)^2
\end{align*}

Also $\lVert y \rVert^2 \leq \frac{3\sqrt{C_f}}{\sqrt{2}\lambda^2 c_f}$,
so the above is bounded by
\begin{align*}
&\frac{\lambda}{2}\sqrt{c_f}\left(\frac{1-e^{-2\lambda\sqrt{c_f}}}{(e^{\lambda\sqrt{c_f}} -e^{-\lambda\sqrt{c_f}})^2}\right) \left(\lVert x \rVert^2 - e^{\lambda\sqrt{ c_f}}\left(\frac{3d\sqrt{C_f}}{\sqrt{2}\lambda^2 c_f}\right)^{1/2}\lVert x \rVert\right)\\
&+ \sqrt{c_f\lambda}\left(\frac{e^{2\lambda\sqrt{c_f}} -1}{(e^{\lambda\sqrt{c_f}} -e^{-\lambda\sqrt{c_f}})^2}\right) \left(\lVert x \rVert^2 - e^{-\lambda\sqrt{ c_f}}\left(\frac{3d\sqrt{C_f}}{\sqrt{2}\lambda^2 c_f}\right)^{1/2}\lVert x \rVert\right) \\
&\geq \frac{\lambda}{2}\sqrt{c_f}\coth(\lambda\sqrt{c_f})\lVert x \rVert^2 - \frac{\sqrt{3d\sqrt{C_f}}}{2\sqrt{2}\sinh(\lambda\sqrt{ c_f})}\lVert x \rVert.
\end{align*}

From Equation \eqref{eqn:pointwise-localization} 
\begin{align*}
    \lvert \xi_{i}(x)\rvert^2 \leq m_{1/2} \exp\left(-\frac{1}{4}\lambda\sqrt{c_f}\coth(\lambda\sqrt{c_f})\lVert x - x^{\star} \rVert^2 + f(x^{\star}) + \frac{\sqrt{3\sqrt{C_f}d}}{4\sqrt{2}\sinh(\lambda\sqrt{ c_f})}\lVert x - x^{\star} \rVert + \frac{3}{2\sqrt{2}}\sqrt{C_f}\lambda d\right),
\end{align*}
for $i \in \{1, 2\}$.

Suppose $\lVert x \rVert^2 \geq r_0 :=\frac{\sqrt{C_f} d\ln(\max(d, d/\lambda))}{\min(\coth(\sqrt{c_f\lambda})\sqrt{c_f}, 1)}$, then for sufficiently large $d$
\begin{align*}
    \lvert u(x) \rvert^2 \leq m  \exp\left(-\frac{\lambda\sqrt{c_f}\coth(\lambda\sqrt{c_f})}{8}\lVert x \rVert^2\right),
\end{align*}
as all other terms decay strictly slower in $d$ than the first grows.

Since
\begin{align*}
     \lVert J_0{\xi} \rVert^2 \leq 2(\lVert J_0 \xi_1\rVert^2 + \lVert J_0 \xi_2\rVert^2),
\end{align*}
we just need to bound the localization of $\xi_1, \xi_2$.

Then

\begin{align*}
    \lVert J_0 \xi_{1/2} \rVert \leq m\Omega_{d}\int_{r > \tilde{r}_0} r^{d-1}\exp\left(-\frac{\lambda\sqrt{c_f}\coth(\lambda\sqrt{c_f})}{8}r^2\right) dr, 
\end{align*}
where $\text{supp} J_0$ is outside a $\ell_2$ ball of radius $\tilde{r}_0$ around $x^{\star}$, and $\tilde{r}_0 \geq r_0$.

Let $k = \frac{d}{\lambda\sqrt{\mu}\coth(\lambda\sqrt{\mu})}$.  Then take $\tilde{r}_0^2 \geq k + 2\sqrt{ckd\ln(d/\lambda)} + 2cd\ln(d/\lambda)$. Then for sufficiently large $d$
\begin{align*}
    \lVert J_0 \xi_{1/2} \rVert &\leq 
    m\frac{\Gamma(\frac{d}{2})}{2(4\lambda \sqrt{c_F}\coth(\lambda\sqrt{c_f}))^{d/2}}(d/\lambda)^{-cd}\\
    &= \mathcal{O}\left(m \sqrt{d}\left((d/\lambda)^{-d(c- 1/2)}(\sqrt{c_f}\coth(\lambda\sqrt{c_f}))^{-d/2}\right)\right).
\end{align*}

Suppose that $\mu = \Omega(1)$, $\lambda = o(d)$, then if $m \leq e^{\mathcal{O}(d)}$, the the above is $o_d(1)$ for some constant $c > 1/2$. Alternatively, if $\lambda = \Omega(d)$, we take replace $\ln(d/\lambda) \rightarrow \ln(d)$, and we get $o_d(1)$ again.

Taking the final radius to be the max of $\tilde{r}_0$ and $r_0$ gives the result.

\end{proof}

\begin{proof}[Proof of Theorem \ref{thm:pert_quad_envelop}]

Since  $J^2(x)$ is a radial mollifier that vanishes outside of a ball of radius $\Omega(\sqrt{d})$ one can show that $\lVert \nabla J(x)\rVert^2 = \mathcal{O}(d^{-1})$. 

The theorem statement implies 
\begin{align}
\label{eqn:closeness}
    f(x)J^2(x) - g(x)J^2(x) = 0,
\end{align}
hence the error terms in \eqref{eqn:compare_gap_bound} involving $f-g$ vanish. For the error term 
\begin{align*}
    \langle J_0\xi| H_f(\lambda) - e_2(\lambda)|J_0\xi\rangle,
\end{align*}
without loss of generality, we can make $H_f$ PSD and hence can just focus on the $e_2(\lambda)\lVert J_0 \xi\rVert^2$ term. Note that  to $e_2(\lambda) = \mathcal{O}(\lambda d)$, which follows from effectively the same argument as for Equation \eqref{eqn:E_2_upper_bound}. Then applying Lemma \ref{lem:localization_bound_quadratic_pot} gives that the above error term is $o_d(1)$.
Then Theorem \ref{thm:yau_strongly_convex} implies
\begin{align*}
    e_2(\lambda) - e_1(\lambda) \geq \lambda \sqrt{c_f}.
\end{align*}
Hence the result follows.
\end{proof}

\subsection{Proof of Corollary \ref{cor:approx_gaussian}}

\label{sec:proof_of_cor_approx_gaussian}

\begin{proof}
Define
\begin{align*}
H(\lambda) = -\Delta + \lambda^2 f = H_g(\lambda) + \lambda^2(f-g),
\end{align*}
where $g(x) = \frac{1}{2}\langle (x-x^{\star}), \nabla^2 f(x^{\star}) (x-x^{\star})\rangle$, so that $H_g(\lambda)$ is a QHO. Hence the ground state energy, $e_1(\lambda)$, of $H_g(\lambda)$ is $ \lambda \text{Tr}\sqrt{\nabla^2 f(x^{\star})}$.

We now proceed similarly as in the proof of Theorem \ref{thm:semiclassical_local_taylor_version}, but we will just be focusing on $E_1(\lambda)$.

For $E_1(\lambda)$, we have (from IMS):
\begin{align*}
    E_1(\lambda)  &\leq e_1(\lambda) + \sup_{x \in \mathbb{R}^d}\lVert\nabla J(x)\rVert^2+ \frac{\langle J\psi_1| \lambda^2(f - g) | J\psi_1\rangle}{\lVert J\psi_1\rVert^2} \\
    &\leq e_1(\lambda) + \lambda^{2\alpha} + \frac{\zeta}{6}\lambda^{2-3\alpha}\\
    & \leq \lambda \text{Tr}\sqrt{\nabla^2 f(x^{\star})}+ \lambda^{2\alpha}+\frac{\zeta}{6}\lambda^{2-3\alpha} \\
    &=   \lambda \text{Tr}\sqrt{\nabla^2 f(x^{\star})}+ \left(1+\frac{\zeta}{6}\right)\lambda^{4/5}
\end{align*}

also by effectively the same argument
\begin{align*}
   e_1(\lambda) \leq E_1(\lambda) + \sup_{x \in \mathbb{R}^d}\lVert\nabla J(x)\rVert^2+ \frac{\langle J\xi_1| \lambda^2(g-f) | J\xi_1\rangle}{\lVert J\xi_1\rVert^2}
\end{align*}
where $\xi_1$ is the ground state of $H(\lambda)$. So
\begin{align*}
    \lambda \text{Tr}\sqrt{\nabla^2 f(x^{\star})} \leq E_1(\lambda)+ \lambda^{4/5}+\left(1+\frac{\gamma}{6}\right)\lambda^{4/5},
\end{align*}
and hence
\begin{align}
\label{eqn:dist_harm_ground}
    \lvert E_1(\lambda) -   \lambda \text{Tr}\sqrt{\nabla^2 f(x^{\star})}\rvert \leq (1+\frac{\gamma}{6})\lambda^{4/5}
\end{align}

From Theorem \ref{thm:semiclassical_local_taylor_version}, for $\lambda$ as specified in that theorem
\begin{align*}
    E_2(\lambda) - E_1(\lambda) \geq \lambda  \sqrt{\mu_{\star}} - o_d(\lambda) > \frac{1}{2}\sqrt{\mu_{\star}} \lambda,
\end{align*}
for sufficiently large $d$.
Consider a circle of radius $r$ in the complex plane:
\begin{align*}
    \mathcal{C} = \{ z \in \mathbb{C} : \lvert z - E_1(\lambda) | = \frac{1}{2}\sqrt{\mu_{\star}} \lambda =: r \},
\end{align*}
then if $\xi_1$ is the unique ground state of $H(\lambda)$:
\begin{align*}
    P_{\lambda} := |\xi_1\rangle\langle\xi_1| = (2\pi i)^{-1}\oint_{\mathcal{C}} dz(z-H(\lambda))^{-1}.
\end{align*}
By Equation \eqref{eqn:dist_harm_ground}, we can take $\lambda$ (specifically $d$) large enough, such that  $\lambda \text{Tr}\sqrt{\nabla^2 f(x^{\star})} \in \mathcal{C}$, giving
\begin{align*}
    |\psi_1\rangle\langle\psi_1| = (2\pi i)^{-1}\oint_{\mathcal{C}} dz(z-H_g(\lambda))^{-1},
\end{align*}
since the gap of $H_g$ is exactly $\sqrt{2\mu}\lambda$. 

Formally,
\begin{align*}
    (z-H(\lambda))^{-1} = \sum_{k=0}^{\infty}(-1)^{k}(z-H_g(\lambda))^{-1}\lambda^2(f-g)(z - H_{a}(\lambda))^{-1}.
\end{align*}

Then for $\alpha = 2/5$,
\begin{align*}
\sum_{k=0}^{M}(-1)^{k} \lvert \langle J\psi_1|(z-H_g(\lambda))^{-1}[\lambda^2(f-g))(z-H_g(\lambda))^{-1}]^{k}|J\psi_1\rangle \rvert &\leq \sum_{k=0}^{M}\left(\frac{\lambda^2\lVert J(f -g)J \rVert_{\infty}}{r^2}\right)^k \\
&\leq \sum_{k=0}^{M} \left(\frac{\gamma}{6 \mu_{\star} \lambda^{6/5}}\right)^{k},
\end{align*}
which of course, being geometric, converges for $M\rightarrow \infty$, for sufficiently large $d$. In fact by dominated convergence, we can move the contour integral under the sum. This gives
\begin{align*}
 &\lvert \langle J\psi_1| (2\pi i)^{-1}\oint_{C} dz(z - H(\lambda))^{-1}|J\psi_1\rangle\rvert  \\&= \lvert (2\pi i)^{-1} \sum_{k=0}^{\infty} \oint_{\mathcal{C}}dz \langle J \psi_1| (z - H_g(\lambda))^{-1} \lambda^2(f-g)(z - H_g(\lambda))^{-1} |J \psi_1\rangle\rvert\\
 &\geq  \lvert (2\pi i)^{-1} \langle J\psi_1|\oint_{\mathcal{C}}dz (z- H_{a}(\lambda))|J\psi_1\rangle\rvert - \lvert (2\pi i)^{-1} \sum_{k=1}^{\infty} \oint_{\mathcal{C}}dz \langle J \psi_1| (z - H_g(\lambda))^{-1} \lambda^2(f-g)(z - H_g(\lambda))^{-1} |J \psi_1\rangle\rvert\\
 &\geq \lvert \langle J\psi_1 | \psi_1\rangle\rvert^2 - \sum_{k=1}^{\infty}\left(\frac{\zeta}{6\sqrt{\mu_{\star}}\lambda^{1/5}}\right)^k\\
 &\geq \lvert \langle J\psi_1 | \psi_1\rangle\rvert^2 - \frac{\gamma}{(6\sqrt{\mu_{\star}}\lambda^{1/5})(1-\frac{\zeta}{6\sqrt{\mu_{\star}}\lambda^{1/5}})}.
\end{align*}

Hence for sufficiently large $d$,
\begin{align*}
    \lvert\langle J \psi_1 | \xi_1\rangle\rvert^2>\lvert \langle J\psi_1 | \psi_1\rangle\rvert^2 - \frac{\gamma}{3\sqrt{\mu_{\star}}\lambda^{1/5}}.
\end{align*}

Note that $\psi_1$ is a Gaussian centered at $x^{\star}$ (which we can take to be $0$ wlog). Hence, using that $\lVert J \rVert_{\infty} = \frac{1}{2}$, 
\begin{align*}
   \lvert \langle J\psi_1 | \psi_1\rangle\rvert^2 = \frac{1}{2} - \frac{1}{2}\mathbb{P}[ \lVert x \rVert \geq \lambda^{-1/5}],
\end{align*}
where from standard Gaussian concentration and using $\nabla^2 f(x^{\star}) \succeq \mu I$,
\begin{align*}
    \mathbb{P}[ \lVert x \rVert \geq \lambda^{-1/5}] &\leq \exp\left(-\frac{\lambda^{-2/5}}{4(\lambda^{-1} d\mu_{\star}^{-1} +  \lambda^{-1}\mu_{\star}^{-1/4}\lambda^{-1/5})}\right)\\
    &\leq\exp\left(-\Omega(\mu_{\star} \lambda^{1/5}/d)\right),
\end{align*}
so our $\lambda$ is already large enough to ensure that $ \mathbb{P}[ \lVert x \rVert \geq \lambda^{-1/5}] < \frac{1}{3}$. The same can be said for  $\frac{\zeta}{3\sqrt{\mu_{\star}}\lambda^{1/5}}$, which we take to be $<\frac{1}{3}$. Giving that
\begin{align*}
    \frac{\lvert \langle J\psi_1 | \xi_1\rangle\rvert^2}{\lVert J\psi_1\rVert^2} > \frac{1}{3(1 - \frac{1}{3})} = \frac{1}{2}.
\end{align*}
Hence, the total variation distance (TVD)
\begin{align*}
    \text{TVD}\left(\frac{(J\psi_1)^2}{\lVert J\psi_1\rVert^2} , \xi_1^2\right) < \frac{1}{2}\sqrt{2 - 2\sqrt{1/2}} < \frac{2}{5}
\end{align*}
This gives that we can ensure that
\begin{align*}
    \mathbb{P}_{\xi_1^2}[\lVert x - x^{\star} \rVert < \epsilon] >  \mathbb{P}_{\psi_1^2}[\lVert x - x^{\star} \rVert < \epsilon] - \frac{2}{5}.
\end{align*}

\end{proof}

\subsection{On Separations in Optimization via Tunneling}
\label{sec:tunneling_degeneratecase}

Another consequence of the quantum and Langevin dynamics connection is a potential advantage in the case of functions with non-unique global minima, which manifests as enhanced tunneling. It is well-known that quantum states tend to tunnel well through ``tall, thin barriers.'' Specifically, \cite{Liu_2023} showed that a certain continuous quantum walk can jump from one global minimum to another with a faster relaxation time than classical Langevin dynamics. However, it is generally unclear in what sense this provides an advantage for optimization problems. We discuss how there can be an advantage in this so-called ``tunneling'' regime, which is again a result of the quantum/Langevin connection discussed in Section~\ref{sec:main-mechanism}.

In the multiple-global-minima case, both the classical Gibbs measure and quantum ground state tunnel between the global minima, however, quantum has a larger tunneling probability. Here, we do not make use of \eqref{eq:langevin-diffusion} with the ground state potential, but instead look at the imaginary-time evolution operator. The connection between the imaginary-time evolution operator and \eqref{eq:langevin-diffusion} is again via the ground-state transformation. 

To illustrate the above idea, we will focus on  multi-dimensional double-well potentials, and connect to the results of Simon~\cite{simon1984semiclassicaltunneling}. First, as a result of the connection in Section~\ref{sec:main-mechanism}, the quantum spectral gap of
\begin{align*}
    H(\lambda) = -\Delta + \lambda^2 f(x)
\end{align*}
can be expressed as the gap of a Dirichlet form
\begin{align}
\label{eqn:dirichlet_form_gap}
    E_2(\lambda) - E_1(\lambda) = \int_{\mathbb{R}^d} \left\lVert \nabla\left(\frac{\xi_2}{\xi_1}\right) (x) \right\rVert^2 \xi_1^2(x) \mathrm{d}x := \int_{\mathbb{R}^d} g_\lambda (x)  \xi_1^2(x) \mathrm{d}x 
\end{align}
where $\xi_2, \xi_1$ are the first-excited and ground state of $H(\lambda)$, respectively, and we let $g_{\lambda}(x) = \lVert \nabla\left(\xi_2/\xi_1\right) (x)\rVert^2$ for brevity.

Even in the non-unique global minima case, the connection between QHO's and general Schr\"odinger operators, from semiclassical analysis, still holds, but results in a trivial gap lower bound of zero. This was discussed in Section \ref{sec:quat_clas_block_sep}. However, it still ends up being useful for approximating the shape of the ground states in the double well case. 

Suppose $f(x) > 0$ is confining and has only two wells centered at $a$ and $b$ that correspond to the two global minima with function value zero. Simon \cite{simon1984semiclassicaltunneling} shows the following.
\begin{Lemma}[Informal version of {\cite[Lemma 2.5]{simon1984semiclassicaltunneling}}]
    There exists constants $C> 0$, and $\Lambda_0 > 0$ such that for $\lambda > \Lambda_0,$ and $\lVert x - a \rVert \leq \lambda^{-1/2}$, $g_{\lambda}(x) > C$ and for $\lVert x- b \rVert \leq \lambda^{-1/2}$, $g_{\lambda}(x) < - C$.
\end{Lemma}
The above is a result non-uniqueness of the minima causing the first-excited and ground states to have large support in both wells, yet be out of phase due to the orthogonality constraint. One can lower bound \eqref{eqn:dirichlet_form_gap} by integrating over a tube connecting $a, b$ and passing through the barrier. Due to the previous Lemma, $g_{\lambda}$ will change substantially across this tube. As Simon \cite{simon1984semiclassicaltunneling} shows, this results in  \eqref{eqn:dirichlet_form_gap} across this tube being lower bounded by effectively the integral of the second factor, i.e. the probability that the ground state $\xi_1$ puts on the barrier.

As is well-known, and also apparent from Agmon's theorem, i.e. Equation \eqref{eqn:pointwise-localization}, the ground state must have mass that is decaying within the barrier region, which is intuitive as this is the ``classically-forbidden'' region. It will turn out that the imaginary-time evolution operator, specifically the relationship
\begin{align*}
    \xi_1(x) = e^{sE_1(\lambda)}\int_{\mathbb{R}^d} e^{-sH(\lambda)}(x, y)\xi_1(y)dy,
\end{align*}
can be used to show that the mass decays with an Agmon-distance-like quantity
\begin{align*}
\rho(y, x) = \inf_{\gamma \in C^1([0,T])} \left\{ \int_0^T \sqrt{f(\gamma(s))} \lVert \dot{\gamma}(s)\rVert \mathrm{d}s | \gamma(0)= y, \gamma(T) = x \right\},
\end{align*}
as $T \rightarrow \infty$. Simon  \cite[Theorem 1.5]{simon1984semiclassicaltunneling} shows that
\begin{align*}
    \lim_{\lambda \rightarrow \infty} -\frac{1}{\lambda}\ln\left(E_2(\lambda) - E_1(\lambda)\right) = \rho(a, b).
\end{align*}
We draw the readers attention to the square-root potential and area-like dependence of the decay rate, so as to contrast this with the previously mention Witten-Laplacian gap estimate (Theorem \ref{thm:lang_sgd_gap})
\begin{align*}
    \delta^{(C)}_{\min} = \mathcal{O}(e^{-\lambda H_f}),
\end{align*}
where $H_f$ is the Morse saddle barrier. In this case, $H_f$ is simple the height of the barrier. Thus for the Schr\"odinger operator with the WKB potential
\begin{align*}
    H_W(\lambda) = -\Delta + \lambda^2 \lVert \nabla f(x)\rVert^2 + \lambda\nabla f, 
\end{align*}
we have
\begin{align*}
        \lim_{\lambda \rightarrow \infty} -\frac{1}{\lambda}\ln\left(E^{(W)}_2(\lambda) - E^{(W)}_1(\lambda)\right) = H_f.
\end{align*}
Again, classical is somewhat insensitive to the uniqueness/non-uniqueness of the global minimizer of $f$, due to the WKB potential usual containing multiple global minimizers, as discussed in Section \ref{sec:quat_clas_block_sep}.

To showcase this difference in quantum and classical Langevin gap dependence, we consider a double-well potential from \cite{Liu_2023}. Let $v \in \mathbb{R}^d$, $\lVert v \rVert =1$ and consider two balls $W_- = \mathcal{B}_2(0, a), W_+ = \mathcal{B}_2(2bv, a), $ with $b > a$. Then let 
\begin{align*}
    B_{v} = \{x \in \mathbb{R}^d | w < \langle x, v\rangle < 2b -w, \sqrt{ \lVert x \rVert^2 - \langle x, v \rangle^2} < \sqrt{a^2 - w^2}, x \notin W_+ \cup W_-\}.
\end{align*}

\begin{align*}
    f(x) = \begin{cases}
        \frac{\omega^2}{2}\lVert x\rVert^2 , & x \in W_-\\
        \frac{\omega^2}{2}\lVert x -2bv\rVert^2 , & x \in W_+\\
         H_1, &x \in B_{v}\\
         H_2 &\text{o.w.},
    \end{cases}
\end{align*}
$H_2 > H_1 \approx \frac{\omega^2}{2}a^2$. %

From \cite[Lemma 4.15]{Liu_2023}, the Agmon distance is equal to
\begin{align*}
    \rho(0, 2bv) = \frac{1}{\sqrt{2}}\omega a^2  + 2(b-a)\sqrt{H_1} \approx \frac{1}{\sqrt{2}}\omega a^2  + \sqrt{2}(b-a) \omega a = \Theta(\omega).
\end{align*}
However for Langevin dynamics, $H_f \sim H_1 \sim \frac{1}{2}\omega^2 a^2 = \Theta(\omega^2)$. 

Consider $\omega = \Theta(\sqrt{d}), a, b = \mathcal{O}(1),$ corresponding to a ``tall, thin barrier''. Under these assumptions, the difference between the classical, $\tau_C$, and quantum relaxation, $\tau_Q$, times is
\begin{align*}
    \tau_{\text{C}} = \mathcal{O}(e^{\lambda d}),  \quad \tau_{\text{Q}} = \mathcal{O}(e^{3\lambda \sqrt{d}}).
\end{align*}
Thus in this setting, quantum tunnels \emph{more} than classical, resulting in a reduced relaxation time (contrast this with Section \ref{sec:quat_clas_block_sep}). The quantum relaxation time goes with effectively the exponential of the square-root of the barrier height. However, quantum has a worse dependence on $a$ and $b$. If they were to also grow with $d$, then the quantum relaxation time would be worse than the classical one. This quantifies the intuition that quantum dynamics can have an optimization advantage over classical due to tunneling when the barriers are tall and thin.

As briefly mentioned in Section \ref{sec:quat_clas_block_sep}, both stronger localization in the unique minimum case and weaker localization in the multiple-minima case are likely required for good quantum performance. Specifically, at the point where the quantum gap is smallest, it is possible for a local minima and the global minimum to cross in energy value. This scenario mimics a double-well setting, in which case quantum will need to effectively tunnel to efficiently track the ground state. If quantum's strong localization kicks-in at this stage, then the state can become trapped, potentially resulting in a worse-than brute-force search runtime \cite{altshuler2010anderson}.

Once this crossing has passed, the global minimum will eventually dominate in energy, and in this setting, we want quantum to tightly localize around the global-minimum well. Here, the uniqueness of the minimizer will help quantum to "see" the minimum, whereas classical Langevin can fail to distinguish it from false minima.

\end{document}